\documentclass[12pt,a4paper]{article}
\usepackage{enumerate}
\usepackage{amssymb,amsmath,amsthm,amsfonts,mathrsfs}

\usepackage{pgf,tikz}
\usetikzlibrary{intersections}
\usetikzlibrary{decorations.pathreplacing}
\usetikzlibrary{positioning}
\usepackage{graphicx}
\usepackage{float}

\usepackage[colorlinks=true, allcolors=black]{hyperref}
\usepackage[all]{hypcap}
\usepackage[capitalize,nameinlink]{cleveref}

\newtheorem{theo}{Theorem}
\newtheorem{prop}[theo]{Proposition}
\newtheorem{lem}[theo]{Lemma}
\newtheorem{cor}[theo]{Corollary}

\theoremstyle{definition}

\newtheorem{claim}{Claim}[theo]

\newtheorem{cas}{Case}[claim]

\renewcommand{\qedsymbol}{$\blacksquare$}

\AtBeginDocument{%
  \mathchardef\mathcomma\mathcode`\,
  \mathcode`\,="8000 
}
{\catcode`,=\active
  \gdef,{\mathcomma\discretionary{}{}{}}
}

\begin{document}
\title{Mixed Unit Interval Bigraphs : A Characterization}
\author{Ashok Kumar Das\thanks{Corresponding Author} , Rajkamal Sahu, Amina Khatun\\
Department of Pure Mathematics, University of Calcutta, Kolkata, India\\
Email Address - ashokdas.cu@gmail.com \& rajkamalmath@gmail.com\\ \& aminakhatun6290@gmail.com}
\maketitle
\begin{abstract}
The class of intersection bigraphs of unit intervals of the real line 
whose ends may be open or closed is called a class of mixed unit interval bigraphs. 
This class of bigraphs is a strict superclass of the class of unit 
interval bigraphs. In a previous paper~\cite{dsml} we have provided four 
infinite families of forbidden induced subgraphs including two 
separate forbidden induced subgraphs of mixed unit interval bigraphs. In that paper, we also 
posed a conjecture concerning characterization of mixed unit interval bigraphs and verified parts of it. In the present paper we shall give a 
complete characterization of mixed unit interval bigraphs.
\end{abstract}
\section{Introduction}
A graph $G=(V,E)$ is an \emph{interval graph} if corresponding to each
vertex $v$ of $G$ we can assign an interval $I_v$ of the real line such 
that two vertices are adjacent if and only if their corresponding 
intervals intersect. The family of intervals $I=\{I_v: v\in V\}$ is 
said to be the interval representation of $G$. If every interval of $I$
is of unit length then $G$ is a \emph{unit interval graph}. Also, 
if no interval in $I$ is properly contained in another then $G$ is 
\emph{proper interval graph}. These classes of graphs have been well understood structurally 
\cite{bls,f,g,sz} as well as algorithmically \cite{c,cknos,cos,hhc,hss,ks}. Also have many applications
in real world problem \cite{ggks,k,py}.\par

However, most of the researchers do not specify which type of interval is used, that is,
whether the intervals are open, closed or semi-closed. This is acceptable
because this class of graphs does not actually depend on this \cite{dlprs}. This is no longer true for the class of unit interval graphs. Rautenbach and Szwarcfiter \cite{rs} showed that the class of intersection graphs of unit open and closed intervals is a strict superclass of the class of unit interval graphs. They also characterized this class of graphs, by a finite list of forbidden induced subgraphs. The work of Rautenbach 
and Szwarcfiter~\cite{rs} was generalized by Dourado et al.~\cite{dlprs}  
allowing all four distinct types of unit intervals namely, open, closed, 
closed-open and open-closed in the representation. Which is known as 
mixed unit interval graphs. Felix Joos \cite{j} characterized this class 
and have shown that the number of forbidden 
induced subgraphs of mixed unit interval graphs are infinite.\par

A bipartite graph (in short, bigraph) $B = (X, Y, E)$ is an 
\emph{interval bigraph} if there exists a one-to-one correspondence 
between the vertex set $X\cup Y$ of $B$ and a collection of intervals 
$\{I(v) : v\in X\cup Y\}$ on the real line such that two vertices are 
adjacent if and only if their corresponding intervals intersect and they 
belong to different partite sets. The collection of intervals 
$\{I(v) : v\in X\cup Y\}$ is called an interval representation of $B$. 
We simply denote the interval representation of $B$ by $I$ (which is a 
function from the vertex set $X\cup Y$ to a collection of intervals).The existence of interval copies is allowed here.\par

An interval bigraph $B = (X, Y, E)$ is a \emph{unit interval bigraph} if all the intervals in the interval representation have unit length. An interval bigraph is a \emph{proper interval bigraph} if in the interval representation no interval is properly contained in another. Hell and Huang \cite{hhi} proved that an interval bigraph is a unit interval bigraph if and only if it does not contain the bipartite claw ($H_1$), the bipartite net ($H_2$) or the bipartite tent ($H_3$) as an induced subgraph (see Fig.~1).
In \cite{ds} we have observed that the bigraphs $H_1, H_2$, and $H_3$ have intersection representation with unit open and closed intervals (see Fig.~2, Fig.~3, Fig.~4 respectively). In the same paper we give have given a characterization of the class of finite intersection bigraphs of unit open and closed intervals in terms of forbidden induced bigraphs.\par

In a subsequent paper~\cite{dsml} we generalize the results of~\cite{ds} to the mixed unit interval bigraphs where we allow all four types of unit intervals namely closed, open, left closed-right open and right closed-left open unit interval in the interval representation. In that paper we provided four infinite families of forbidden induced subgraphs including two separate forbidden induced subgraphs of mixed unit interval bigraphs. In that paper~\cite{dsml} we also posed a conjecture concerning characterization of mixed unit interval bigraphs and verified parts of it. Motivated by the work of Joos \cite{j} and Das and Sahu \cite{ds}, in the present paper we extend the results of~\cite{dsml} and give a complete characterization of mixed unit interval bigraphs.\par

In Section 2 we collect all basic definitions, terminology, and results related to our work. In Section~3 we give the list of forbidden induced subgraphs of mixed unit interval bigraphs. In Section~4 we shall provide structural characterization of mixed unit interval bigraphs.

\section{Preliminaries}
We consider only simple, finite and connected bigraphs. For a bigraph $B = (X, Y,E)$ the
neighborhood of a vertex $u\in X\cup Y$ is denoted by $N_B(u)$. Two distinct vertices $u$ and
$v$ of $B$ are \emph{copies} if $N_B(u) = N_B(v)$. Sometime $u$ and $v$ are also called \textit{false twins}. If no two vertices of $B$ are copies then $B$ is copy-free. If $\mathcal{F}$ is a set of graphs and any graph $G$ does not contain a graph in $\mathcal{F}$ as an induced subgraph then $G$ is $\mathcal{F}$-free.\par

Let $\mathcal{M}$ be a family of sets. An $\mathcal{M}$-intersection representation of a bigraph is a function $f:X\cup Y \to \mathcal{M}$ such that for any two distinct vertices $u$ and $v$ of a bigraph $B$, we have $uv\in E$ if and only if $f(u)\cap f(v)\neq \emptyset$ and $u$ and $v$ are vertices of different partite sets. A bigraph is an $\mathcal{M}$-bigraph if it has an $\mathcal{M}$-intersection representation.\par
For two real numbers $a$ and $b$, we denote the open interval $\{x\in \mathbb{R}|a<x<b \}$ by $(a,b)$, the closed interval $\{x\in \mathbb{R}|a\leq x\leq b \}$ by $[a,b]$, the open-closed interval $\{x\in \mathbb{R}|a<x\leq b \}$ by $(a,b]$ and the closed-open interval $\{x\in \mathbb{R}|a\leq x< b \}$ by $[a,b)$. For an interval $I$, let $l(I)=\inf(I)$ and $r(I)=\sup(I)$. We assume $\mathcal{I}^{++}$ is the set of closed intervals, $\mathcal{I}^{--}$ is the set of open intervals, $\mathcal{I}^{+-}$ is the set of closed-open intervals and $\mathcal{I}^{-+}$ is the set of open-closed intervals. Also assume $\mathcal{U}^{++}$ is the set of unit closed intervals, $\mathcal{U}^{--}$ is the set of  unit open intervals, $\mathcal{U}^{+-}$ is the set of unit closed-open intervals and $\mathcal{U}^{-+}$ is the set of unit open-closed intervals. In addition, let $\mathcal{I}^\pm=\mathcal{I}^{++}\cup \mathcal{I}^{--}$, $\mathcal{U}^\pm=\mathcal{U}^{++}\cup \mathcal{U}^{--}$, $\mathcal{I}=\mathcal{I}^{++}\cup \mathcal{I}^{--} \cup \mathcal{I}^{+-} \cup \mathcal{I}^{-+}$, and $\mathcal{U}=\mathcal{U}^{++}\cup \mathcal{U}^{--}\cup \mathcal{U}^{+-}\cup \mathcal{U}^{-+}$.\par
Our first result shows that as in the case of interval graphs, the class of interval bigraphs does not depend on the type of the interval used in the intersection representation.\par

Actually, in \cite{dsml} we have shown that the following class of  bigraphs are equivalent.

\begin{prop}
The classes of $\mathcal{I}^{++}$-bigraphs, $\mathcal{I}^{--}$-bigraphs, $\mathcal{I}^{\pm}$-bigraphs, $\mathcal{I}^{+-}$-bigraphs, $\mathcal{I}^{-+}$-bigraphs and $\mathcal{I}$-bigraphs are the same.
\end{prop}
The following proposition extends the result of Proposition~\ref{p2} of \cite{dlprs} which showed that a bigraph is a $\mathcal{U}^{++}$-bigraph if and only if it is a $\mathcal{U}^{--}$-bigraph.
\begin{prop}[\cite{dsml}]\label{p2}
The classes of $ \mathcal{U}^{++}$-bigraphs, $\mathcal{U}^{--}$-bigraphs, $\mathcal{U}^{+-}$-bigraphs, $\mathcal{U}^{-+}$-bigraphs and $\mathcal{U}^{+-}\cup \mathcal{U}^{-+}$-bigraphs are the same.
\end{prop}  
\begin{figure}[H]
\tikzstyle{every node}=[circle, draw, fill=black,
                        inner sep=0pt, minimum width=2.5pt]
\centering
\begin{tikzpicture}[scale=.40]
\draw (0:0) node {} --++ (45:1) node {} --++ (45:1) node (a) {} --++ 
(90:1) node {} -- ++ (90:1) node {};
\draw (a) --++ (-45:1) node{}--++ (-45:1) node (b){};

\coordinate (c) at ([xshift=3cm]b);
\draw (c) node{} --++ (50:1.2)node(d){}--++ (50:1.2) node(e){}--++(90:1.4) 
node{};
\draw (e)--++(-50:1.2)node(f){}--++(-50:1.2)node(g){};
\draw (d) --(e|- d) node{}--(f) ;

\coordinate (h) at ([xshift=3cm]g);
\draw (h)node{} --++ (90:1.6)node{}--++(0:1.4)node{}--++(0:1.4)node{}--++(-90:1.6)node{}--++(180:1.4)node(i){}--cycle;
\draw (i)--++(90:3.2)node{};
\end{tikzpicture}
\caption{The bipartite claw $(H_1)$, net $(H_2)$ and tent $(H_3)$}
\label{f1}
\end{figure}
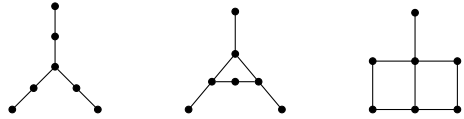

As mentioned in the introduction, the bigraphs $H_1, H_2$ and $H_3$ have $\mathcal{U}^\pm$-intersection
representation, (see Fig.~2, Fig.~3 and Fig.~4 of \cite{ds,dsml} ).
For the sake of convenience we give below their representations.\par

\begin{figure}[H]
\centering
\begin{tikzpicture}[scale=.40]
{\tikzstyle{every node}=[circle, draw, fill=black,
                        inner sep=0pt, minimum width=2.5pt]
                        
\draw (0:0) node[label={[label distance=1pt]180:${\scriptstyle x_2}$}] {} --++ (45:1) node[label={[label distance=1pt]180:${\scriptstyle y_2}$}] {} --++ (45:1) node[label={[label distance=1pt]0:${\scriptstyle x_4}$}] (a) {} --++ 
(90:1) node[label={[label distance=1pt]0:${\scriptstyle y_1}$}] {} -- ++ (90:1) node[label={[label distance=1pt]0:${\scriptstyle x_1}$}] {};
\draw (a) --++ (-45:1) node[label={[label distance=1pt]0:${\scriptstyle y_3}$}]{}--++ (-45:1) node[label={[label distance=1pt]0:${\scriptstyle x_3}$}] (b){};
}
\pgfmathsetmacro{\b}{1.5}
\pgfmathsetmacro{\c}{0.6}
\pgfmathsetmacro{\d}{0.13}
\draw \foreach \p/\q/\r in {6*\b/0*\c/y_2,5*\b/1*\c/y_1,7*\b/1*\c/y_3,4*\b/2*\c/x_1,6*\b/2*\c/x_4,8*\b/2*\c/x_3}
{
(\p,\q)--(\p+\b,\q)
(\p+\d,\q-\d)--(\p,\q-\d)--(\p,\q+\d)--(\p+\d,\q+\d)
(\p+\b-\d,\q-\d)--(\p+\b,\q-\d)--(\p+\b,\q+\d)--(\p+\b-\d,\q+\d)
(\p+0.5*\b,\q+0.25) node{${\scriptstyle \r}$}
};
\draw \foreach \p/\q/\r in {6*\b/3*\c/x_2}
{
(\p,\q)--(\p+\b,\q)
(\p-\d,\q-\d)--(\p,\q-\d)--(\p,\q+\d)--(\p-\d,\q+\d)
(\p+\b+\d,\q-\d)--(\p+\b,\q-\d)--(\p+\b,\q+\d)--(\p+\b+\d,\q+\d)
(\p+0.5*\b,\q+0.25) node{${\scriptstyle \r}$}
};
\end{tikzpicture}
\caption{The bipartite claw $H_1$ with its $\mathcal{U}^\pm$-intersection representation.}\label{f2}
\end{figure}
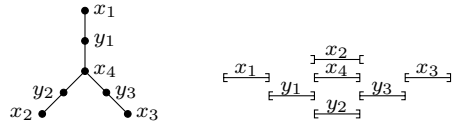

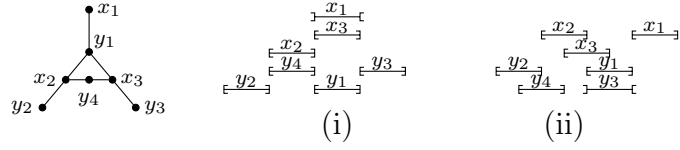
\begin{figure}[H]
\centering
\begin{tikzpicture}[scale=.40]
{\tikzstyle{every node}=[circle, draw, fill=black,
                        inner sep=0pt, minimum width=2.5pt]
\draw (0:0) node[label={[label distance=1pt]180:${\scriptstyle y_2}$}]{} --++ (50:1.2)node[label={[label distance=1pt]180:${\scriptstyle x_2}$}](d){}--++ (50:1.2) node[label={[label distance=1pt]10:${\scriptstyle y_1}$}](e){}--++(90:1.4) 
node[label={[label distance=1pt]0:${\scriptstyle x_1}$}]{};
\draw (e)--++(-50:1.2)node[label={[label distance=1pt]0:${\scriptstyle x_3}$}](f){}--++(-50:1.2)node[label={[label distance=1pt]0:${\scriptstyle y_3}$}](g){};
\draw (d) --(e|- d) node[label={[label distance=1pt]-90:${\scriptstyle y_4}$}]{}--(f) ;
}
\pgfmathsetmacro{\b}{1.5}
\pgfmathsetmacro{\c}{0.6}
\pgfmathsetmacro{\d}{0.13}
\draw \foreach \p/\q/\r in {4*\b/1*\c/y_2,6*\b/1*\c/y_1,5*\b/2*\c/y_4,7*\b/2*\c/y_3,5*\b/3*\c/x_2,6*\b/4*\c/x_3, 11*\b/4*\c/x_2,13*\b/4*\c/x_1,11.5*\b/3*\c/x_3,10*\b/2*\c/y_2,12*\b/2*\c/y_1,10.5*\b/1*\c/y_4}
{
(\p,\q)--(\p+\b,\q)
(\p+\d,\q-\d)--(\p,\q-\d)--(\p,\q+\d)--(\p+\d,\q+\d)
(\p+\b-\d,\q-\d)--(\p+\b,\q-\d)--(\p+\b,\q+\d)--(\p+\b-\d,\q+\d)
(\p+0.5*\b,\q+0.25) node{${\scriptstyle \r}$}
};

\draw \foreach \p/\q/\r in {6*\b/5*\c/x_1, 12*\b/1*\c/y_3}
{
(\p,\q)--(\p+\b,\q)
(\p-\d,\q-\d)--(\p,\q-\d)--(\p,\q+\d)--(\p-\d,\q+\d)
(\p+\b+\d,\q-\d)--(\p+\b,\q-\d)--(\p+\b,\q+\d)--(\p+\b+\d,\q+\d)
(\p+0.5*\b,\q+0.25) node{${\scriptstyle \r}$}
};
\draw (6.5*\b,-1*\c) node {(i)};
\draw (11.5*\b,-1*\c) node {(ii)};
\end{tikzpicture}
\caption{The bipartite net $H_2$ with its two $\mathcal{U}^\pm$-intersection representation.}\label{f3}
\end{figure}

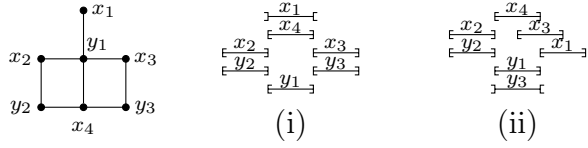
\begin{figure}[H]
\centering
\begin{tikzpicture}[scale=.40]
{\tikzstyle{every node}=[circle, draw, fill=black,
                        inner sep=0pt, minimum width=2.5pt]
\draw (0:0)node[label={[label distance=1pt]180:${\scriptstyle y_2}$}]{} --++ (90:1.6)node[label={[label distance=1pt]180:${\scriptstyle x_2}$}]{}--++(0:1.4)node[label={[label distance=1pt]45:${\scriptstyle y_1}$}]{}--++(0:1.4)node[label={[label distance=1pt]0:${\scriptstyle x_3}$}]{}--++(-90:1.6)node[label={[label distance=1pt]0:${\scriptstyle y_3}$}]{}--++(180:1.4)node[label={[label distance=1pt]-90:${\scriptstyle x_4}$}](i){}--cycle;
\draw (i)--++(90:3.2)node[label={[label distance=1pt]0:${\scriptstyle x_1}$}]{};
}
\pgfmathsetmacro{\b}{1.5}
\pgfmathsetmacro{\c}{0.6}
\pgfmathsetmacro{\d}{0.13}
\draw \foreach \p/\q/\r in {5*\b/1*\c/y_1,4*\b/2*\c/y_2,6*\b/2*\c/y_3,4*\b/3*\c/x_2,6*\b/3*\c/x_3,5*\b/4*\c/x_4, 10*\b/5*\c/x_4,9*\b/4*\c/x_2,10.5*\b/4*\c/x_3,11*\b/3*\c/x_1, 9*\b/3*\c/y_2,10*\b/2*\c/y_1}
{
(\p,\q)--(\p+\b,\q)
(\p+\d,\q-\d)--(\p,\q-\d)--(\p,\q+\d)--(\p+\d,\q+\d)
(\p+\b-\d,\q-\d)--(\p+\b,\q-\d)--(\p+\b,\q+\d)--(\p+\b-\d,\q+\d)
(\p+0.5*\b,\q+0.25) node{${\scriptstyle \r}$}
};
\draw \foreach \p/\q/\r in {5*\b/5*\c/x_1, 10*\b/1*\c/y_3}
{
(\p,\q)--(\p+\b,\q)
(\p-\d,\q-\d)--(\p,\q-\d)--(\p,\q+\d)--(\p-\d,\q+\d)
(\p+\b+\d,\q-\d)--(\p+\b,\q-\d)--(\p+\b,\q+\d)--(\p+\b+\d,\q+\d)
(\p+0.5*\b,\q+0.25) node{${\scriptstyle \r}$}
};
\draw (5.5*\b,-1*\c) node {(i)};
\draw (10.5*\b,-1*\c) node {(ii)};
\end{tikzpicture}
\caption{The bipartite tent $H_3$ with its two $\mathcal{U}^\pm$-intersection representation.}
\label{f4}
\end{figure}

Therefore, the class of $\mathcal{U}^\pm$-bigraphs is a strict superclass of the class of unit interval
bigraphs. We have characterized this class of bigraphs in \cite{ds} in the following Theorem~3.
For an $\mathcal{I}^{++}$-bigraph if two vertices $u$ and $v$ are copies then they belong to the same
partite set. Therefore, in the $\mathcal{I}^{++}$-interval representation we can take the same interval for these
two vertices. Thus we consider that our bigraphs are copy-free.

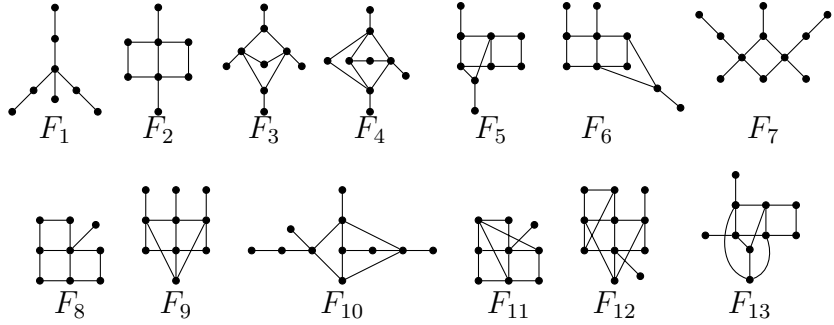
\begin{figure}[H]
\centering
\begin{tikzpicture}[scale=.40]
{\tikzstyle{every node}=[circle, draw, fill=black,
                        inner sep=0pt, minimum width=2.5pt]
\draw (0:0) node {} --++ (45:1) node {} --++ (45:1) node (a) {} --++ 
(90:1) node {} -- ++ (90:1.03) node {};
\draw (a) --++ (-45:1) node{}--++ (-45:1) node (b){};
\draw (a) --++ (-90:1) node {};

\coordinate (c) at ([xshift=2cm]b);
\draw (c) node{} --++(90:1.15)node{}--++(90:1.15)node{}--++(180:1)node{}--++(-90:1.15)
node{}--++(0:2)node{}--++(90:1.15)node{}--++(180:1)--++(90:1.15)node{};

\coordinate (d) at ([xshift=3.5cm]c);
\draw (d)node{}--++(90:.7)node(e){}--++(60:1.5)node(f){}--++(-45:.7)node{};
\path [name path=d1] (d) --++ (90:3.2);
\path [name path=f1] (f) --++ (135:1.5);
\path [name path=f2] (f) --++ (210:2);
\path [name intersections={of=d1 and f1, by=g}];
\path [name intersections={of=d1 and f2, by=i}];
\draw (f)--(g)node{} (e)--++(120:1.5)node(h){}--++(-135:.7)node{} (g)--(h) 
(g)--++(90:.7)node{} (h)--(i)node{}--(f);

\coordinate (j) at ([xshift=3.5cm]d);
\draw (j)node{}--++(90:.7)node(k){}--++(55:1.2)node(l){}--++(-45:.7)node{};
\path [name path=j1] (j) --++ (90:3.2);
\path [name path=l1] (l) --++ (125:1.5);
\path [name path=l2] (l) --++ (180:2);
\path [name intersections={of=j1 and l1, by=m}];
\path [name intersections={of=j1 and l2, by=o}];
\path [name path=k1] (k) --++ (145:3.2);
\path [name path=m1] (m) --++ (35:-3.2);
\path [name intersections={of=k1 and m1, by=o2}];
\draw (l)--(m)node{} (k)--++(125:1.2)node(n){} (m)--(n) 
(m)--++(90:.7)node{} (n)--(o)node{}--(l) (k)--(o2)node{}--(m);

\coordinate (p) at ([xshift=3cm,yshift=3.5cm]j);
\draw (p)node{}--++(-90:1)node{}--++(0:1)node(q){}--++(0:1)node{}--++(-90:1)
node{}--++(180:1)node(r){}--++(180:1)node(s){}--++(90:1)node{} (q)--(r);
\path [name path=q1] (q) --++ (-110:1.7);
\path [name path=s1] (s) --++ (-45:.8);
\path [name intersections={of=q1 and s1, by=t}];
\draw (q)--(t)node{}--(s) (t)--++(-90:1)node{};

\coordinate (u) at ([xshift=3.5cm]p);
\draw (u)node{}--++(-90:1)node{}--++(0:1)node{}--++(0:1)node(v){}--++
(-90:1)node{}--++(180:1)node(w){}--++(180:1)node{}--++(90:1)node{} 
(w)--++(90:2)node{};
\path [name path=v1] (v) --++ (-60:2.2);
\path [name path=w1] (w) --++ (-20:2.2);
\path [name intersections={of=v1 and w1, by=x}];
\draw (v)--(x)node{}--(w) (x)--++(-40:1)node{};

\coordinate (aa) at ([xshift=6.5cm]u);

\draw (aa)node{}--++(-90:1)node(bb){}--++(-45:1)node(cc){}--++(45:1)node{}--++(45:1)node{}
(cc)--++(-45:1)node{}
(cc)--++(45:-1)node(dd){}--++(-45:-1)node(ee){}--++(45:-1)node{}
(ee)--++(-45:-1)node{}--++(-45:-1)node{} (ee)--(bb);
}
\coordinate (z) at (0,-0.6);
\draw (z-|a)node{$F_1$} (z-|c)node{$F_2$} (z-|d)node{$F_3$} (z-|j)node{$F_4$} (z-|q)node{$F_5$} (z-|w)node{$F_6$} (z-|aa)node{$F_7$};

{\tikzstyle{every node}=[circle, draw, fill=black,
                        inner sep=0pt, minimum width=2.5pt]

\coordinate (ff) at ([xshift=-.5cm,yshift=-7cm]a);
\draw (ff)node{}--++(0:1)node(gg){}--++(0:1)node{}--++(90:1)node{}
--++(0:-1)node(hh){}--++(90:1)node{}--++(0:-1)node{}--++
(90:-1)node(kk){}--(ff) 
(gg)--(hh)--(kk) 
(hh)--++(45:1.2)node{};

\coordinate (ll) at ([xshift=3.5cm]gg);
\draw (ll)node{}--++(90:1)node(mm){}--++(90:1)node(nn){}--++
(90:1)node{} (nn)--++(0:1)node(oo){}--++(90:-1)node{}--(mm)--++
(0:-1)node{}--++(90:1)node(pp){}--(nn)
(oo)--(ll)--(pp)--++(90:1)node{} (oo)--++(90:1)node{};

\coordinate (qq) at ([xshift=5.5cm]ll);
\draw (qq)node{}--++(90:1)node(rr){}--++(90:1)node(ss){}--++
(90:1)node{} (rr)--++(0:1)node{}--++(0:1)node(tt){}--++(0:1)node{}
(qq)--(tt)--(ss);
\path [name path=v1] (qq) --++ (-45:-2.2);
\path [name path=w1] (ss) --++ (45:-2.2);
\path [name intersections={of=v1 and w1, by=uu}];
\draw (qq)--(uu)node{}--(ss) (uu)--++(0:-1)node{}--++(0:-1)node{}
(uu)--++(-45:-1)node{};

\coordinate (vv) at ([xshift=4.5cm]qq);
\draw (vv)node{}--++(0:1)node(ww){}--++(0:1)node{}--++
(90:1)node(z2){}
--++(0:-1)node(xx){}--++(90:1)node{}--++(0:-1)node(zz){}--++
(90:-1)node(yy){}--(vv)
(yy)--(xx)--(ww) 
(xx)--++(45:1.2)node{} (z2)--(zz)--(ww);

\coordinate (a1) at ([xshift=4.5cm]vv);
\draw (a1)node{}--++(90:1)node(b1){}--++(0:1)node{}--++
(90:1)node(c1){}--++(0:-1)node(d1){}--++(90:1)node(e1){}--++
(0:-1)node{}--++(90:-1)node(f1){}--(d1)--(b1)--++(0:-1)node(g1){}
--(f1) (g1)--(e1) (c1)--(a1)--(f1) (b1)--++(-45:1.2)node{} (c1)--++(90:1)node{};

\coordinate (n2) at ([xshift=4cm,yshift=3.5cm]a1);
\draw (n2)node{}--++(-90:1)node(p){}--++(0:1)node(q){}--++(0:1)node{}--++(-90:1)
node{}--++(180:1)node(r){}--++(180:1)node(s){}--++(90:1)node{} (q)--(r) (s)--++(0:-1)node{};
\path [name path=q1] (q) --++ (-110:1.7);
\path [name path=s1] (s) --++ (-45:.8);
\path [name intersections={of=q1 and s1, by=t}];
\draw (q)--(t)node{}--(s) (t)--++(-90:1)node(u){};
\draw    (u)to[out=20,in=-70](r) (u)to[out=160,in=-120](p);
}
\coordinate (z3) at ([yshift=-0.8cm]gg);
\draw (z3-|gg)node{$F_8$} (z3-|ll)node{$F_9$} (z3-|qq)node{$F_{10}$} (z3-|ww)node{$F_{11}$} (z3-|a1)node{$F_{12}$} (z3-|t)node{$F_{13}$};
\end{tikzpicture}
\caption{Forbidden induced subgraphs of $\mathcal{U}^\pm$-bigraphs.}
\label{f5}
\end{figure}
Now we define a restricted type
of $\mathcal{I}^\pm$-bigraphs that are almost proper and will be useful in our 
characterization theorem. A bipartite graph $B$ is a $\mathcal{PI}^\pm$-bigraph (\textit{or almost proper interval bigraph}) if it 
has an $\mathcal{I}^\pm$-intersection representation $I:V(B)\to \mathcal{I}^\pm$ such that
\begin{enumerate}[(i)]
\item if $I(u)$ and $I(v)$ are two distinct closed intervals of $I$, then $I(u)
\not\subset I(v)$ and $I(v) \not\subset I(u)$, that is, no closed interval 
properly contains another closed interval, and
\item for every vertex $u$ of $B$ with $I(u)\in \mathcal{I}^-$, there is a vertex $v$ of 
$B$ with $I(v)\in \mathcal{I}^+$, such that $l(I(u))=l(I(v))$, and $r(I(u))=r(I(v))$,
that is, for every open interval, there is a closed interval with the same end
points.
\end{enumerate}
\begin{theo}[\cite{ds}]
For a bipartite graph $B$, the following statements are equivalent.
\begin{enumerate}[(i)]
\item $B$ is a $\{F_1,F_2,F_3,F_4,F_5,F_6,F_7,F_8,F_9,F_{10},F_{11},F_{12},F_{13}\}$-free interval bigraph.
\item $B$ is an almost proper interval bigraph.
\item $B$ is a $\mathcal{U}^{\pm}$-bigraph.
\end{enumerate}
\end{theo}

It has observed in \cite{ds} that the bigraphs $F_6$ has mixed unit interval representations. Thus 
the class of $\mathcal{U}$-bigraphs properly contains the 
class of $\mathcal{U}^\pm$-bigraphs.
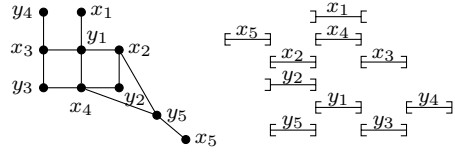
\begin{figure}[H]
\centering
\begin{tikzpicture}[scale=.50]
{\tikzstyle{every node}=[circle, draw, fill=black,
                        inner sep=0pt, minimum width=3pt]
\draw (0:0)node[label={[label distance=1pt]180:${\scriptstyle y_4}$}]{}--++(-90:1)node[label={[label distance=1pt]180:${\scriptstyle x_3}$}]{}--++(0:1)node[label={[label distance=1pt]20:${\scriptstyle y_1}$}]{}--++(0:1)node[label={[label distance=1pt]0:${\scriptstyle x_2}$}](v){}--++
(-90:1)node[label={[label distance=1pt]-20:${\scriptstyle y_2}$}]{}--++(180:1)node[label={[label distance=1pt]-90:${\scriptstyle x_4}$}](w){}--++(180:1)node[label={[label distance=1pt]180:${\scriptstyle y_3}$}]{}--++(90:1) 
(w)--++(90:2)node[label={[label distance=1pt]0:${\scriptstyle x_1}$}]{};
\path [name path=v1] (v) --++ (-60:2.2);
\path [name path=w1] (w) --++ (-20:2.2);
\path [name intersections={of=v1 and w1, by=x}];
\draw (v)--(x)node[label={[label distance=1pt]0:${\scriptstyle y_5}$}]{}--(w) (x)--++(-40:1)node[label={[label distance=1pt]0:${\scriptstyle x_5}$}]{};
}
\pgfmathsetmacro{\b}{1.2}
\pgfmathsetmacro{\c}{0.6}
\pgfmathsetmacro{\d}{0.15}
\draw \foreach \p/\q/\r in {5*\b/-5.2*\c/y_5,7*\b/-5.2*\c/y_3,6*\b/-4.2*\c/y_1,8*\b/-4.2*\c/y_4,5*\b/-2.2*\c/x_2,7*\b/-2.2*\c/x_3,4*\b/-1.2*\c/x_5,6*\b/-1.2*\c/x_4}
{
(\p,\q)--(\p+\b,\q)
(\p+\d,\q-\d)--(\p,\q-\d)--(\p,\q+\d)--(\p+\d,\q+\d)
(\p+\b-\d,\q-\d)--(\p+\b,\q-\d)--(\p+\b,\q+\d)--(\p+\b-\d,\q+\d)
(\p+0.5*\b,\q+0.2) node{${\scriptstyle \r}$}
};
\draw \foreach \p/\q/\r in {6*\b/-0.2*\c/x_1}
{
(\p,\q)--(\p+\b,\q)
(\p-\d,\q-\d)--(\p,\q-\d)--(\p,\q+\d)--(\p-\d,\q+\d)
(\p+\b+\d,\q-\d)--(\p+\b,\q-\d)--(\p+\b,\q+\d)--(\p+\b+\d,\q+\d)
(\p+0.5*\b,\q+0.2) node{${\scriptstyle \r}$}
};
\draw \foreach \p/\q/\r in {5*\b/-3.2*\c/y_2}
{
(\p,\q)--(\p+\b,\q)
(\p-\d,\q-\d)--(\p,\q-\d)--(\p,\q+\d)--(\p-\d,\q+\d)
(\p+\b-\d,\q-\d)--(\p+\b,\q-\d)--(\p+\b,\q+\d)--(\p+\b-\d,\q+\d)
(\p+0.5*\b,\q+0.2) node{${\scriptstyle \r}$}
};
\end{tikzpicture}
\caption{A $\mathcal{U}$-intersection representation of $F_6$.}
\label{f11}
\end{figure}

As mentioned in \cite{ds} it may be noted that in a $\mathcal{U}$-representation of a bigraph we can make some \emph{trivial modifications}. Which are suitable interval shifts that preserve intersections and non-intersections changes in the types (open, closed or open-closed and closed-open) of some intervals, refection of the entire model about a point on the real line, translation of the entire model, and relabeling of some intervals that preserve intersection and non intersection of the intervals corresponding to different partite sets of vertices.\par
We have observed in \cite{ds} that $\mathcal{U}^\pm$ representation of each of the bigraphs $H_1$, $H_2$ and $H_3$ is unique up to trivial modifications.

\section{Forbidden induced subgraphs of Mixed unit interval bigraphs}\label{s3}
It is easy to verify that the bigraphs $F_2$, $F_4$, $F_5$, $F_8$, $F_9$, $F_{11}$ and $F_{12}$ of Figure~5 have no $\mathcal{U}$-intersection representation. 
Also, as mentioned earlier, that in \cite{dsml} we have obtained four forbidden infinite families, namely $\mathcal{L}$, $\mathcal{M}$, $\mathcal{N}$, $\mathcal{H'}$ 
(Figures 7, 8, 9 and 10 respectively) of mixed unit interval bigraphs along with the forbidden bigraphs $B_1$, and  $B_2$ (Figures 17 and 18).
In this section we shall provide six new families of forbidden induced subgraphs, namely the infinite families $\mathcal{K}$, $\mathcal{P}$, $\mathcal{T}$, $\mathcal{Q}$, $\mathcal{R}$, and $\mathcal{S}$ and the bigraphs $K$, $M$, $H_0$.\par

\begin{figure}[H]
\centering
\begin{minipage}[b]{.5\textwidth}
\centering
\vfill
\begin{tikzpicture}[scale=.40]
{\tikzstyle{every node}=[circle, draw, fill=black, inner sep=0pt, minimum width= 2.5pt]

\draw (0:0)node(a){}--++(30:1)node(b){}--++(-30:-1)node{}--++(-30:-1)node{} (b)--++(0:1)node(c){} (c)--++(90:1)node{}--++(90:1)node{} (c)--++(-90:1)node(n){} (c)--++(0:1)node(d){}--++(30:1)node{}--++(30:1)node{} (d)--++(-30:1)node{}--++(-30:1)node{} (a)--++(30:-1)node{};
\coordinate (n1) at ([yshift=-.01cm]n);

\coordinate (z) at ([xshift=7cm]a);
\draw (z)node{}--++(30:1)node(b){}--++(-30:-1)node{}--++(-30:-1)node{} (b)--++(0:1)node(c){} (c)--++(90:1)node{}--++(90:1)node{} (c)--++(-90:1)node(n){} (c)--++(0:1)node(e){}--++(0:1)node(d){}--++(30:1)node{}--++(30:1)node{} (d)--++(-30:1)node{}--++(-30:1)node{} (e)--++(-90:1)node{} (z)--++(30:-1)node{};
\coordinate (n2) at ([yshift=-.01cm]n);

\coordinate (z) at ([yshift=-4.5cm]a);
\draw (z)node{}--++(30:1)node(b){}--++(-30:-1)node{}--++(-30:-1)node{} (b)--++(0:1)node(g){}--++(-90:1)node{};
\draw[loosely dotted, thick] (g)--++(0:2)node(h){};
\draw (h)--++(0:1)node(f){}--++(0:1)node(c){} (c)--++(90:1)node{}--++(90:1)node{} (c)--++(-90:1)node(n){} (c)--++(0:1)node(e){}--++(0:1)node(d){};
\draw[loosely dotted, thick] (d)--++(0:2)node(i){};
\draw (i)--++(0:1)node(j){}--++(30:1)node{}--++(30:1)node{} (j)--++(-30:1)node{}--++(-30:1)node{} (e)--++(-90:1)node{} (f)--++(-90:1)node{} (h)--++(-90:1)node{} (d)--++(-90:1)node{} (i)--++(-90:1)node{} (z)--++(30:-1)node{};
\coordinate (n4) at ([yshift=-.01cm]n);
}
\node[label=below:{$L_{1,1}$}] at (n1){};
\node[label=below:{$L_{1,2}$}] at (n2){};
\node[label=below:{$L_{i,j}$}] at (n4){};
\draw[decoration={brace,raise=3pt,amplitude=6pt},decorate]
  (b) -- node[above=6pt] {${\scriptstyle i}$ {\footnotesize vertices}} (f);
\draw[decoration={brace,raise=3pt,amplitude=6pt},decorate]
  (e) -- node[above=6pt] {${\scriptstyle j}$ {\footnotesize vertices}} (j);

\end{tikzpicture}
\vfill
\caption{The class $\mathcal{L}$}
\end{minipage}%
\begin{minipage}[b]{.5\textwidth}
\centering
\vfill
\begin{tikzpicture}[scale=.40]
{\tikzstyle{every node}=[circle, draw, fill=black, inner sep=0pt, minimum width= 2.5pt]
\draw (0:0)node(e){}--++(45:1)node(a){}--++(45:1)node(b){}--++(-45:-1)node(c){}--++(-45:-1)node{} (a)--++(-45:-1)node{}--(c) (b)--++(90:1)node{} (b)--++(0:1)node(d){}--++(30:1)node{}--++(30:1)node{} (d)--++(-30:1)node{}--++(-30:1)node{};
\coordinate (n1) at ([xshift=1.5cm,yshift=.2cm]e);

\coordinate (f) at ([xshift=6cm]e);
\draw (f)node{}--++(45:1)node(a){}--++(45:1)node(b){}--++(-45:-1)node(c){}--++(-45:-1)node{} (a)--++(-45:-1)node{}--(c) (b)--++(90:1)node{} (b)--++(0:1)node(d){}--++(0:1)node(e1){}--++(-30:1)node{}--++(-30:1)node{} (d)--++(90:1)node{} (e1)--++(30:1)node{}--++(30:1)node{};
\coordinate (n2) at ([xshift=2cm,yshift=.2cm]f);

\coordinate (z) at ([xshift=1.5cm,yshift=-4.5cm]e);
\draw (z)node{}--++(45:1)node(a){}--++(45:1)node(b){}--++(-45:-1)node(c){}--++(-45:-1)node{} (a)--++(-45:-1)node{}--(c) (b)--++(90:1)node{} (b)--++(0:1)node(d){}--++(0:1)node(e1){}--++(90:1)node{} (d)--++(90:1)node{};
\draw[loosely dotted,thick] (e1)--++(0:2)node(g){};
\draw (g)--++(0:1)node(h){}--++(30:1)node{}--++(30:1)node{} (h)--++(-30:1)node{}--++(-30:1)node{} (g)--++(90:1)node{};
\coordinate (n3) at ([xshift=3cm]z);
}
\node[label=below:{$M_1$}] at (n1){};
\node[label=below:{$M_2$}] at (n2){};
\node[label=below:{$M_i$}] at (n3){};
\draw[decoration={brace,mirror,raise=5pt,amplitude=7pt},decorate]
  (d) -- node[below=9pt] {${\scriptstyle i}$ {\footnotesize vertices}} (h);
\end{tikzpicture}
\vfill
\caption{The class $\mathcal{M}$.}
\end{minipage}
\end{figure}

\begin{figure}[H]
\centering
\begin{minipage}[b]{.5\textwidth}
\centering
\vfill
\begin{tikzpicture}[scale=.40]
{\tikzstyle{every node}=[circle, draw, fill=black, inner sep=0pt, minimum width= 2.5pt]

\draw (0:0)node(a){}--++(0:1)node(b){}--++(0:1)node(n){}--++(90:1)node(c){}--++(0:-1)node{}--++(0:-1)node{} (a)--++(90:2)node{} (b)--++(90:2)node{}  ;
\path [name path=d1] (b)--++(-15:3);
\path [name path=f1] (c)--++(-45:3);
\path [name intersections={of=d1 and f1, by=d}];
\draw (b)--(d)node{}--(c) (d)--++(30:1)node{}--++(30:1)node{} (d)--++(-30:1)node{}--++(-30:1)node{};
\coordinate (n1) at ([yshift=-.5cm]n);

\coordinate (z) at ([xshift=7cm]a);
\draw (z)node{}--++(0:1)node(b){}--++(0:1)node(n){}--++(90:1)node(c){}--++(0:-1)node{}--++(0:-1)node{} (z)--++(90:2)node{} (b)--++(90:2)node{}  ;
\path [name path=d1] (b)--++(-15:3);
\path [name path=f1] (c)--++(-45:3);
\path [name intersections={of=d1 and f1, by=d}];
\draw (b)--(d)node{}--(c) (d)--++(0:1)node(e){}--++(30:1)node{}--++(30:1)node{} (d)--++(90:1)node{}
(e)--++(-30:1)node{}--++(-30:1)node{};
\coordinate (n2) at ([yshift=-.5cm]n);

\coordinate (y) at ([xshift=2cm,yshift=-4.5cm]a);
\draw (y)node{}--++(0:1)node(b){}--++(0:1)node(n){}--++(90:1)node(c){}--++(0:-1)node{}--++(0:-1)node{} (y)--++(90:2)node{} (b)--++(90:2)node{}  ;
\path [name path=d1] (b)--++(-15:3);
\path [name path=f1] (c)--++(-45:3);
\path [name intersections={of=d1 and f1, by=d}];
\draw (b)--(d)node{}--(c) (d)--++(0:1)node(e){}--++(90:1)node{} (d)--++(90:1)node{}
(e)--++(0:1)node(f){}--++(90:1)node{};
\draw[loosely dotted, thick] (f)--++(0:2)node(g){};
\draw (g)--++(0:1)node(h){}--++(30:1)node{}--++(30:1)node{} (h)--++(-30:1)node{}--++(-30:1)node{} (g)--++(90:1)node{};
\coordinate (n3) at ([yshift=-.5cm]n);
}
\node[label=below:{$N_1$}] at (n1){};
\node[label=below:{$N_2$}] at (n2){};
\node[label=below:{$N_i$}] at (n3){};
\draw[decoration={brace,mirror,raise=5pt,amplitude=7pt},decorate]
  (d) -- node[below=9pt] {${\scriptstyle i}$ {\footnotesize vertices}} (h);
\end{tikzpicture}
\vfill
\caption{The class $\mathcal{N}$.}
\end{minipage}%
\begin{minipage}[b]{.5\textwidth}
\centering
\vfill
\begin{tikzpicture}[scale=.40]
{\tikzstyle{every node}=[circle, draw, fill=black, inner sep=0pt, minimum width= 2.5pt]
\draw (0:0)node(a){}--++(0:1)node{}--++(90:1)node(b){}--++(90:1)node{}--++(0:-1)node{}--++(-90:1)node(c){}--(a) (c)--(b)--++(0:1)node(d){}--++(30:1)node{}--++(30:1)node{}
(b)--++(30:1)node{} (d)--++(-30:1)node{}--++(-30:1)node{};
\coordinate (n1) at ([xshift=2cm]a);

\coordinate (f) at ([xshift=6cm]a);
\draw (f)node{}--++(0:1)node{}--++(90:1)node(b){}--++(90:1)node{}--++(0:-1)node{}--++(-90:1)node(c){}--(f) (c)--(b)--++(0:1)node(d){}--++(30:1)node{}
(b)--++(30:1)node{} (d)--++(0:1)node(e){}--++(-30:1)node{}--++(-30:1)node{} (e)--++(30:1)node{}--++(30:1)node{};
\coordinate (n2) at ([xshift=2.5cm]f);

\coordinate (z) at ([xshift=1cm,yshift=-4.5cm]a);
\draw (z)node{}--++(0:1)node{}--++(90:1)node(b){}--++(90:1)node{}--++(0:-1)node{}--++(-90:1)node(c){}--(z) (c)--(b)--++(0:1)node(d){}--++(30:1)node{}
(b)--++(30:1)node{} (d)--++(0:1)node(e){}--++(0:1)node(f){} (e)--++(30:1)node{} (f)--++(30:1)node{};
\draw[loosely dotted, thick] (f)--++(0:2)node(g){};
\draw (g)--++(30:1)node{} (g)--++(0:1)node(h){}--++(30:1)node{}--++(30:1)node{} (h)--++(-30:1)node{}--++(-30:1)node{};
\coordinate (n3) at ([xshift=3cm]z);
}
\node[label=below:{$H_1'$}] at (n1){};
\node[label=below:{$H_2'$}] at (n2){};
\node[label=below:{$H_i'$}] at (n3){};
\draw[decoration={brace,mirror,raise=3pt,amplitude=6pt},decorate]
  (d) -- node[below=6pt] {${\scriptstyle i}$ {\footnotesize vertices}} (h);
\end{tikzpicture}
\vfill
\caption{The class $\mathcal{H'}$.}
\end{minipage}
\end{figure}

\begin{figure}[H]
\centering
\begin{minipage}[b]{.5\textwidth}
\centering
\vfill
\begin{tikzpicture}[scale=.40]
{\tikzstyle{every node}=[circle, draw, fill=black, inner sep=0pt, minimum width= 2.5pt]

\draw (0:0)node(a){}--++(45:1)node(b){}--++(-45:1)node(c){}--++(45:-1)node(d){}--(a)--++(0:-1)node(a1){}--++(30:-1)node{}--++(30:-1)node{} (a1)--++(-30:-1)node{}--++(-30:-1)node{} (a)--++(45:-1)node{}
(c)--++(0:1)node(c1){}--++(30:1)node{}--++(30:1)node{} (c1)--++(-30:1)node{}--++(-30:1)node{} (c)--++(-45:1)node{}
(b)--++(90:1)node{};
\coordinate (n1) at ([yshift=-.5cm]d);

\coordinate (e) at ([xshift=9cm]c);
\draw (e)node{}--++(45:1)node(f){}--++(-45:1)node(g){}--++(45:-1)node(h){}--(e)--++(0:-1)node(e1){}--++(0:-1)node(e2){}--++(30:-1)node{}--++(30:-1)node{} (e2)--++(-30:-1)node{}--++(-30:-1)node{} (e)--++(45:-1)node{} (e1)--++(45:-1)node{}
(g)--++(0:1)node(g1){}--++(30:1)node{}--++(30:1)node{} (g1)--++(-30:1)node{}--++(-30:1)node{} (g)--++(-45:1)node{} (f)--++(90:1)node{} ;
\coordinate (n2) at ([yshift=-.5cm]h) ;

\coordinate (i) at ([xshift=4.4cm,yshift=-5cm]a);
\draw (i)node{}--++(45:1)node(j){}--++(-45:1)node(k){}--++(45:-1)node(l){}--(i)--++(0:-1)node(i1){} (k)--++(0:1)node(k1){} (j)--++(90:1)node{};
\draw[loosely dotted, thick] (i1)--++(0:-2)node(i2){} (k1)--++(0:2)node(k2){};
\draw (i2)--++(0:-1)node(i3){}--++(30:-1)node{}--++(30:-1)node{} (i3)--++(-30:-1)node{}--++(-30:-1)node{} (i)--++(45:-1)node{} (i1)--++(45:-1)node{} (i2)--++(45:-1)node{}
(k2)--++(0:1)node(k3){}--++(30:1)node{}--++(30:1)node{} (k3)--++(-30:1)node{}--++(-30:1)node{} (k)--++(-45:1)node{} (k1)--++(-45:1)node{} (k2)--++(-45:1)node{};
\coordinate (n3) at ([yshift=-.5cm]l) ;
}

\node[label=below:{$K_{1,1}$}] at (n1){};
\node[label=below:{$K_{2,1}$}] at (n2){};
\node[label=below:{$K_{i,j}$}] at (n3){};
\draw[decoration={brace,raise=5pt,amplitude=7pt},decorate]
  (i3) -- node[above=10pt] {${\scriptstyle i}$ {\footnotesize vertices}} (i1);
\draw[decoration={brace,raise=5pt,amplitude=7pt},decorate]
  (k1) -- node[above=9pt] {${\scriptstyle j}$ {\footnotesize vertices}} (k3);
\end{tikzpicture}
\vfill
\caption{The class $\mathcal{K}$}
\end{minipage}%
\begin{minipage}[b]{.5\textwidth}
\centering
\vfill
\begin{figure}[H]
\centering
\begin{tikzpicture}[scale=.40]
{\tikzstyle{every node}=[circle, draw, fill=black, inner sep=0pt, minimum width= 2.5pt]

\draw (0:0)node(q){}--++(90:1)node(r){}--++(90:1)node(s){}--++
(90:1)node{} (r)--++(0:1)node(t){}--++(0:1)node{}
(q)--(t)--(s);
\draw (r)--++(0:.5)node(r'){};
\path [name path=v1] (q) --++ (-45:-2.2);
\path [name path=w1] (s) --++ (45:-2.2);
\path [name intersections={of=v1 and w1, by=u}];
\draw (q)--(u)node{}--(s) (u)--++(30:-1)node(u1){}--++(30:-1)node{}
(u)--++(-30:-1)node{}--++(-30:-1)node{} ;
\coordinate (n1) at ([yshift=-.5cm]q);

\coordinate (a) at ([xshift=10cm]q);
\draw (a)node{}--++(90:1)node(b){}--++(90:1)node(c){}--++
(90:1)node{} (b)--++(0:1)node(d){}--++(0:1)node{}
(a)--(d)--(c);
\draw (b)--++(0:.5)node(b'){};
\path [name path=v1] (a) --++ (-45:-2.2);
\path [name path=w1] (c) --++ (45:-2.2);
\path [name intersections={of=v1 and w1, by=e}];
\draw (a)--(e)node{}--(c) (e)--++(0:-1)node(e1){}--++(30:-1)node{}--++(30:-1)node{}
(e)--++(-45:-1)node{} (e1)--++(-30:-1)node{}--++(-30:-1)node{} ;
\coordinate (n2) at ([yshift=-.5cm]a);

\coordinate (a1) at ([xshift=6.5cm,yshift=-4.5cm]q);
\draw (a1)node{}--++(90:1)node(b){}--++(90:1)node(c){}--++
(90:1)node{} (b)--++(0:1)node(d){}--++(0:1)node{}
(a1)--(d)--(c);

\draw (b)--++(0:.5)node(b'){};

\path [name path=v1] (a1) --++ (-45:-2.2);
\path [name path=w1] (c) --++ (45:-2.2);
\path [name intersections={of=v1 and w1, by=e}];
\draw (a1)--(e)node{}--(c) (e)--++(0:-1)node(e1){}--++(0:-1)node(e2){};
\draw[loosely dotted, thick] (e2)--++(0:-2)node(e3){};

\draw (e3)node{}--++(0:-1)node(e4){} (e)--++(-45:-1)node{} (e1)--++(-45:-1)node{} (e2)--++(-45:-1)node{} (e3)--++(-45:-1)node{} (e4)--++(-30:-1)node{}--++(-30:-1)node{} (e4)--++(30:-1)node{}--++(30:-1)node{};
\coordinate (n3) at ([xshift=-1.5cm,yshift=-.5cm]a1);
}

\node[label=below:{$P_{1}$}] at (n1){};
\node[label=below:{$P_{2}$}] at (n2){};
\node[label=below:{$P_{i}$}] at (n3){};
\draw[decoration={brace,mirror,raise=5pt,amplitude=7pt},decorate]
  (e4) -- node[below=10pt] {${\scriptstyle i}$ {\footnotesize vertices}} (e);
\end{tikzpicture}
\end{figure}
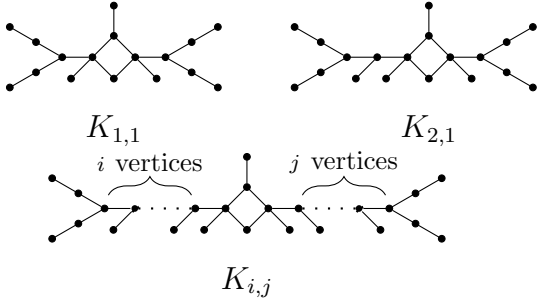
\vfill
\caption{The class $\mathcal{P}$}
\end{minipage}
\end{figure}

\begin{figure}[H]
\centering
\begin{minipage}[b]{.5\textwidth}
\centering
\vfill
\begin{tikzpicture}[scale=.40]
{\tikzstyle{every node}=[circle, draw, fill=black, inner sep=0pt, minimum width= 2.5pt]

\draw (0:0)node(a){}--++(90:1)node(b){}--++(90:1)node{}--++(90:1)node{} (b)--++(0:-1)node(c){}--++(0:-1)node(c1){}--++(0:-1)node{} (c)--++(-50:-1)node{}--++(0:-1)node{}--(c1) (c)--++(90:1)node{} (b)--++(0:1)node(m2){}--++(0:1)node(n){}--++(0:1)node{} (m2)--++(50:1)node{}--++(0:1)node{}--(n) (m2)--++(90:1)node{};
\coordinate (n1) at ([yshift=-.01cm]a);

\coordinate (j) at ([xshift=8cm]a);
\draw (j)node{}--++(90:1)node(k){}--++(90:1)node{}--++(90:1)node{} (k)--++(0:-1)node(k1){}--++(0:-1)node(k2){}--++(0:-1)node(l){}--++(0:-1)node{} (k1)--++(90:1)node{}
(k2)--++(90:1)node{} (k2)--++(-50:-1)node{}--++(0:-1)node{}--(l) (k)--++(0:1)node(m2){}--++(0:1)node(n){}--++(0:1)node{} (m2)--++(50:1)node{}--++(0:1)node{}--(n) (m2)--++(90:1)node{};
\coordinate (n3) at ([yshift=-.01cm]j);

\coordinate (p) at ([xshift=3.5cm,yshift=-5cm]a);
\draw (p)node{}--++(90:1)node(q){}--++(90:1)node{}--++(90:1)node{} (q)--++(0:-1)node(q1){}--++(0:-1)node(q2){} (q1)--++(90:1)node{} (q2)--++(90:1)node{} (q)--++(0:1)node(s1){}--++(0:1)node(s2){} (s2)--++(90:1)node{} (s1)--++(90:1)node{} ;

\draw[loosely dotted, thick] (q2)--++(0:-2)node(q3){} (s2)--++(0:2)node(s3){};

\draw (q3)--++(0:-1)node(q4){}--++(0:-1)node(r){}--++(0:-1)node{} (q4)--++(90:1)node{} (q4)--++(-50:-1)node{}--++(0:-1)node{}--(r) (q3)--++(90:1)node{}  (s3)--++(0:1)node(s4){}--++(0:1)node(t){}--++(0:1)node{} (s4)--++(50:1)node{}--++(0:1)node{}--(t) (s3)--++(90:1)node{} (s4)--++(90:1)node{};
\coordinate (n4) at ([yshift=-.01cm]p);
}
\node[label=below:{$T_{1,1}$}] at (n1){};
\node[label=below:{$T_{2,1}$}] at (n3){};
\node[label=below:{$T_{i,j}$}] at (n4){};
\draw[decoration={brace,mirror,raise=5pt,amplitude=7pt},decorate]
  (q4) -- node[below=10pt] {${\scriptstyle i}$ {\footnotesize vertices}} (q1);
\draw[decoration={brace,mirror,raise=5pt,amplitude=7pt},decorate]
  (s1) -- node[below=10pt] {${\scriptstyle j}$ {\footnotesize vertices}} (s4);
\end{tikzpicture}
\vfill
\caption{The class $\mathcal{T}$}
\end{minipage}%
\begin{minipage}[b]{.5\textwidth}
\centering
\vfill
\begin{tikzpicture}[scale=.40]
{\tikzstyle{every node}=[circle, draw, fill=black,
                        inner sep=0pt, minimum width=3pt]

\draw (0:0)node{}--++(-90:1)node(a){}--++(0:1)node(b){}--++(0:1)node(c){}--++(-90:1)
node{}--++(180:1)node(d){}--++(180:1)node(e){}--++(a) (b)--(d) (e)--++(0:-1)node{};
\path [name path=b1] (b) --++ (-110:1.7);
\path [name path=e1] (e) --++ (-45:.8);
\path [name intersections={of=b1 and e1, by=f}];

\path [name path=c1] (c) --++ (-55:3);
\path [name path=d1] (d) --++ (-20:3);
\path [name intersections={of=c1 and d1, by=g}];
\draw (b)--(f)node{}--(e) (f)--++(-90:1)node(h){} ;
\draw    (h)to[out=30,in=-80](d) (h)to[out=160,in=-120](a) (c)to[out=-20,in=100](g)node{} (g)to[out=180,in=-50](d) (g)--++(20:1)node{}--++(20:1)node{} (g)--++(-20:1)node{}--++(-20:1)node{};
\coordinate (n1) at ([yshift=-.4cm]h);

\coordinate (z) at ([xshift=8cm,yshift=1cm]a);
\draw (z)node{}--++(-90:1)node(a1){}--++(0:1)node(b){}--++(0:1)node(c){}--++(-90:1)
node{}--++(180:1)node(d){}--++(180:1)node(e){}--++(a1) (b)--(d) (e)--++(0:-1)node{};
\path [name path=b1] (b) --++ (-110:1.7);
\path [name path=e1] (e) --++ (-45:.8);
\path [name intersections={of=b1 and e1, by=f}];

\path [name path=c1] (c) --++ (-55:3);
\path [name path=d1] (d) --++ (-20:3);
\path [name intersections={of=c1 and d1, by=g}];
\draw (b)--(f)node{}--(e) (f)--++(-90:1)node(h){} ;
\draw    (h)to[out=30,in=-80](d) (h)to[out=160,in=-120](a1) (c)to[out=-20,in=100](g)node{} (g)to[out=180,in=-50](d) (g)--++(0:1)node(i){}--++(20:1)node{}--++(20:1)node{} (g)--++(-90:1)node{} (i)--++(-20:1)node{}--++(-20:1)node{};
\coordinate (n2) at ([yshift=-.4cm]h);

\coordinate (z) at ([xshift=3cm,yshift=-6cm]a);
\draw (z)node{}--++(-90:1)node(a){}--++(0:1)node(b){}--++(0:1)node(c){}--++(-90:1)
node{}--++(180:1)node(d){}--++(180:1)node(e){}--++(a) (b)--(d) (e)--++(0:-1)node{};
\path [name path=b1] (b) --++ (-110:1.7);
\path [name path=e1] (e) --++ (-45:.8);
\path [name intersections={of=b1 and e1, by=f}];

\path [name path=c1] (c) --++ (-55:3);
\path [name path=d1] (d) --++ (-20:3);
\path [name intersections={of=c1 and d1, by=g}];
\draw (b)--(f)node{}--(e) (f)--++(-90:1)node(h){} ;
\draw    (h)to[out=30,in=-80](d) (h)to[out=160,in=-120](a) (c)to[out=-20,in=100](g)node{} (g)to[out=180,in=-50](d) (g)--++(0:1)node(i){}--++(0:1)node(j){}--++(-90:1)node(k){} (g)--++(-90:1)node{} (i)--++(-90:1)node{};
\draw[loosely dotted, thick] (j)--++(0:2)node(k){};
\draw (k)--++(-90:1)node{} (k)--++(0:1)node(l){}--++(20:1)node{}--++(20:1)node{} (l)--++(-20:1)node{}--++(-20:1)node{};

\coordinate (n3) at ([yshift=-.4cm]h);

}
\draw[decoration={brace,raise=13pt,amplitude=7pt},decorate]
  (g) -- node[above=18pt] {${\scriptstyle i}$ {\footnotesize vertices}} (l);

\node[label=below:{$Q_1$}] at (n1){};
\node[label=below:{$Q_2$}] at (n2){};
\node[label=below:{$Q_i$}] at (n3){};

\end{tikzpicture}
\vfill
\caption{The class $\mathcal{Q}$.}
\end{minipage}
\end{figure}

\begin{figure}[H]
\centering
\begin{minipage}[b]{.5\textwidth}
\centering
\vfill
\begin{tikzpicture}[scale=.40]
{\tikzstyle{every node}=[circle, draw, fill=black, inner sep=0pt, minimum width= 2.5pt]

\draw (0:0)node(a){}--++(0:1)node(b){}--++(0:1)node(c){}--++(90:1)node{}--++(90:1)node{} (c)--++(-50:-1)node{}--++(0:-1)node{}--(b) (c)--++(50:1)node{} (c)--++(0:1)node(d){}--++(20:1)node{}--++(20:1)node{}
(d)--++(-20:1)node{}--++(-20:1)node{} ;
\coordinate (n2) at ([yshift=-.4cm]c);

\coordinate (a1) at ([xshift=8cm]a);
\draw (a1)node{}--++(0:1)node(b1){}--++(0:1)node(c1){}--++(90:1)node{}--++(90:1)node{} (c1)--++(-50:-1)node{}--++(0:-1)node{}--(b1) (c1)--++(50:1)node{} (c1)--++(0:1)node(d1){}--++(0:1)node(e1){}--++(20:1)node{}--++(20:1)node{}
(e1)--++(-20:1)node{}--++(-20:1)node{} (d1)--++(50:1)node{} ;
\coordinate (n3) at ([yshift=-.4cm]d1);

\coordinate (a2) at ([xshift=3.5cm,yshift=-5cm]a);
\draw (a2)node{}--++(0:1)node(b2){}--++(0:1)node(c2){}--++(90:1)node{}--++(90:1)node{} (c2)--++(-50:-1)node{}--++(0:-1)node{}--(b2) (c2)--++(50:1)node{} (c2)--++(0:1)node(d2){}--++(50:1)node{};

\draw[loosely dotted, thick] (d2)--++(0:2)node(e2){} ;

\draw (e2)--++(0:1)node(f2){}--++(20:1)node{}--++(20:1)node{}
(f2)--++(-20:1)node{}--++(-20:1)node{} (e2)--++(50:1)node{};
}
\coordinate (n4) at ([yshift=-1.5cm]d2);

\node[label=below:{$R_1$}] at (n2){};
\node[label=below:{$R_2$}] at (n3){};
\node[label=below:{$R_i$}] at (n4){};
\draw[decoration={brace,mirror,raise=6pt,amplitude=7pt},decorate]
  (d2) -- node[below=9pt] {${\scriptstyle i}$ {\footnotesize vertices}} (f2);
\end{tikzpicture}
\vfill
\caption{The class $\mathcal{R}$}
\end{minipage}%
\begin{minipage}[b]{.5\textwidth}
\centering
\vfill
\begin{tikzpicture}[scale=.40]
{\tikzstyle{every node}=[circle, draw, fill=black, inner sep=0pt, minimum width= 2.5pt]

\draw (0:0)node(a){}--++(0:1)node{}--++(0:1)node(b){}--++(90:1)node{}--++(90:1)node{} (b)--++(60:1)node(d){} (b)--++(0:2)node(c){}--++(120:1)node(e){}--(d) (c)--++(20:1)node{}--++(20:1)node{}
(c)--++(-20:1)node{}--++(-20:1)node{};
\coordinate (n2) at ([yshift=-.2cm]b);

\coordinate (g) at ([xshift=8cm]a);
\draw (g)node{}--++(0:1)node{}--++(0:1)node(h){}--++(90:1)node{}--++(90:1)node{} (h)--++(60:1)node(j){} (h)--++(0:2)node(i){}--++(0:1)node(l){} (i)--++(90:1)node{} (i)--++(120:1)node(k){}--(j)(l)--++(20:1)node{}--++(20:1)node{}
(l)--++(-20:1)node{}--++(-20:1)node{};
\coordinate (n3) at ([yshift=-.2cm]i);

\coordinate (p) at ([xshift=3.5cm,yshift=-5cm]a);
\draw (p)node{}--++(0:1)node{}--++(0:1)node(h1){}--++(90:1)node{}--++(90:1)node{} (h1)--++(60:1)node(j1){} (h1)--++(0:2)node(i1){}--++(0:1)node(l1){}--++(90:1)node{} (i1)--++(90:1)node{} (i1)--++(120:1)node(k1){}--(j1) ;

\draw[loosely dotted, thick] (l1)--++(0:2)node(m1){} ;

\draw (m1)--++(0:1)node(n1){}--++(20:1)node{}--++(20:1)node{}
(n1)--++(-20:1)node{}--++(-20:1)node{} (m1)--++(90:1)node{};
}
\coordinate (n4) at ([yshift=-1.5cm]i1);

\node[label=below:{$S_1$}] at (n2){};
\node[label=below:{$S_2$}] at (n3){};
\node[label=below:{$S_i$}] at (n4){};
\draw[decoration={brace,mirror,raise=5pt,amplitude=7pt},decorate]
  (i1) -- node[below=10pt] {${\scriptstyle i}$ {\footnotesize vertices}} (n1);
\end{tikzpicture}
\vfill
\caption{The class $\mathcal{S}$}
\end{minipage}
\end{figure}

\begin{figure}[H]
\centering
\begin{minipage}[b]{.5\textwidth}
\centering
\vfill
\begin{tikzpicture}[scale=.40]
{\tikzstyle{every node}=[circle, draw, fill=black, inner sep=0pt, minimum width= 2.5pt]
\draw (0:0)node{}--++(45:1)node{}--++(45:1)node(a){}--++(90:1)node{}--++(90:1)node{} (a)--++(-45:1)node{}--++(-45:1)node{}
(a)--++(-90:1)node{}--++(-90:1)node{};
}
\end{tikzpicture}
\vfill
\caption{The bigraph $B_1$}
\end{minipage}%
\begin{minipage}[b]{.5\textwidth}
\centering
\vfill
\begin{tikzpicture}[scale=.40]
{\tikzstyle{every node}=[circle, draw, fill=black, inner sep=0pt, minimum width= 2.5pt]
\draw (0:0)node{}--++(45:1)node(b){}--++(45:1)node(a){}--++(90:1)node{}--++(90:1)node{} (a)--++(-45:1)node(c){}--++(-45:1)node{}
(a)--++(20:1)node{}--++(-45:1)node{}--(c) (a)--++(-20:-1)node{}--++(45:-1)node{}--(b) (a)--++(-90:1)node{};
}
\end{tikzpicture}
\vfill
\caption{The bigraph $B_2$}
\end{minipage}
\end{figure}

\begin{figure}[H]
\centering
\begin{minipage}[b]{.5\textwidth}
\centering
\vfill
\begin{tikzpicture}[scale=.50]
{\tikzstyle{every node}=[circle, draw, fill=black, inner sep=0pt, minimum width= 2.5pt]

\draw (0:0)node(a){}--++(45:1)node(b)
{}--++(-45:1)node(c){}--++(45:-1)node(d){}--(a)--++(0:-1)node(a1){}--++(0:-1)node(a2){} (a)--++(30:-1)node{}--++(30:-1)node{} (c)--++(0:1)node{}--++(0:1)node{} (c)--++(-30:1)node{} (b)--++(90:1)node{};
}
\end{tikzpicture}
\vfill
\caption{The bigraph $K$.}
\end{minipage}%
\begin{minipage}[b]{.5\textwidth}
\centering
\vfill
\begin{tikzpicture}[scale=.50]
{\tikzstyle{every node}=[circle, draw, fill=black, inner sep=0pt, minimum width= 2.5pt]

\draw (0:0)node(a){}--++(0:1)node(b){}--++(0:1)node{}--++(90:1)node{}--++(0:-1)node(c){}--++(0:-1)node{}--(a) (b)--(c)--++(90:1)node(d){}--++(90:1)node{} (c)--++(30:1)node{}--++(90:1)node{}--(d) (c)--++(-45:-1)node{};
}
\end{tikzpicture}
\vfill
\caption{The bigraph $H_0$.}
\end{minipage}
\end{figure}

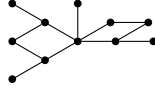
\begin{figure}[H]
\centering

\begin{tikzpicture}[scale=.50]
{\tikzstyle{every node}=[circle, draw, fill=black, inner sep=0pt, minimum width= 2.5pt]

\draw (0:0)node{}--++(30:1)node(a){}--++(30:1)node(b){}--++(-30:-1)node(c){}--++(-30:-1)node{} (a)--++(150:1)node{}--(c) (b)--++(90:1)node{} (b)--++(0:1)node(d){}--++(30:1)node{}--++(0:-1)node{}--(b) (d)--++(0:1)node{};
}
\end{tikzpicture}

\caption{The bigraph $M$.}

\end{figure}

In following Lemmas we shall prove that each of the infinite families 
$\mathcal{K}$, $\mathcal{P}$, $\mathcal{T}$, $\mathcal{Q}$, $\mathcal{R}$, and $\mathcal{S}$ are also forbidden families of subgraphs of mixed unit interval bigraphs.\par
In the bigraphs $K_{i,j}'$, $P'_i$, $Q'_i$, $S'_i$, and $R'_i$ we call the two 
vertices $v'$ and $v''$ adjacent to $u$ are the \emph{special vertices} of the 
respective bigraphs. Now in Lemma~4, 5, 6, 7, 8 and 9 respectively we show 
that the bigraphs $K'_{i,j}$, $P'_i$, $T'_{i,j}$, $Q'_i$, $S'_i$ and 
$R'_i$ have unique $\mathcal{U}$-representation up to 
trivial modifications and the bigraphs $K_{i,j}$, $P_i$, $T_{i,j}$, 
$Q_i$, $S_i$ and $R_i$ are forbidden induced subgraphs 
of $\mathcal{U}$-bigraphs.\par

\begin{figure}[H]
\centering
\begin{tikzpicture}[scale=.50]
{\tikzstyle{every node}=[circle, draw, fill=black, inner sep=0pt, minimum width= 2.5pt]

\draw (0:0)node(a)[label={[label distance=1pt]90:${\scriptstyle x_1''}$}]{}--++(45:1)node(b)[label={[label distance=1pt]10:${\scriptstyle y_0}$}]{}--++(-45:1)node(c)[label={[label distance=1pt]80:${\scriptstyle x_1}$}]{}--++(45:-1)node(d)[label={[label distance=1pt]-90:${\scriptstyle y}$}]{}--(a)--++(0:-1)node(a1)[label={[label distance=1pt]90:${\scriptstyle y_1''}$}]{}--++(30:-1)node[label={[label distance=1pt]-80:${\scriptstyle x''}$}]{}--++(30:-1)node[label={[label distance=1pt]180:${\scriptstyle y''}$}]{} (a1)--++(-30:-1)node[label={[label distance=1pt]80:${\scriptstyle x'}$}]{}--++(-30:-1)node[label={[label distance=1pt]180:${\scriptstyle y'}$}]{} (a)--++(45:-1)node[label={[label distance=1pt]-90:${\scriptstyle y_1'''}$}]{}
(c)--++(0:1)node(c1)[label={[label distance=1pt]90:${\scriptstyle y_1}$}]{}--++(30:1)node[label={[label distance=1pt]0:${\scriptstyle x_0'}$}]{} (c1)--++(-30:1)node[label={[label distance=1pt]0:${\scriptstyle x_0''}$}]{} (c)--++(-45:1)node[label={[label distance=1pt]-90:${\scriptstyle y_1'}$}]{}
(b)--++(90:1)node[label={[label distance=1pt]180:${\scriptstyle x_0}$}]{};
\coordinate (n1) at ([yshift=-.5cm]d);

\coordinate (e) at ([xshift=9cm]c);
\draw (e)node[label={[label distance=1pt]90:${\scriptstyle x_1''}$}]{}--++(45:1)node(f)[label={[label distance=1pt]10:${\scriptstyle y_0}$}]{}--++(-45:1)node(g)[label={[label distance=1pt]80:${\scriptstyle x_1}$}]{}--++(45:-1)node(h)[label={[label distance=1pt]-90:${\scriptstyle y}$}]{}--(e)--++(0:-1)node(e1)[label={[label distance=1pt]90:${\scriptstyle y_1''}$}]{}--++(0:-1)node(e2)[label={[label distance=1pt]90:${\scriptstyle x_2''}$}]{}--++(30:-1)node[label={[label distance=1pt]-80:${\scriptstyle y''}$}]{}--++(30:-1)node[label={[label distance=1pt]180:${\scriptstyle x''}$}]{} (e2)--++(-30:-1)node[label={[label distance=1pt]80:${\scriptstyle y'}$}]{}--++(-30:-1)node[label={[label distance=1pt]180:${\scriptstyle x'}$}]{} (e)--++(45:-1)node[label={[label distance=1pt]-90:${\scriptstyle y_1'''}$}]{} (e1)--++(45:-1)node[label={[label distance=1pt]-90:${\scriptstyle x_1'''}$}]{}
(g)--++(0:1)node(g1)[label={[label distance=1pt]90:${\scriptstyle y_1}$}]{}--++(30:1)node[label={[label distance=1pt]0:${\scriptstyle x_0'}$}]{} (g1)--++(-30:1)node[label={[label distance=1pt]0:${\scriptstyle x_0''}$}]{} (g)--++(-45:1)node[label={[label distance=1pt]-90:${\scriptstyle y_1'}$}]{} (f)--++(90:1)node[label={[label distance=1pt]180:${\scriptstyle x_0}$}]{} ;
\coordinate (n2) at ([yshift=-.5cm]h) ;

\coordinate (i) at ([xshift=4.4cm,yshift=-5cm]a);
\draw (i)node[label={[label distance=1pt]90:${\scriptstyle x_1''}$}]{}--++(45:1)node(j)[label={[label distance=1pt]10:${\scriptstyle y_0}$}]{}--++(-45:1)node(k)[label={[label distance=1pt]80:${\scriptstyle x_1}$}]{}--++(45:-1)node(l)[label={[label distance=1pt]-90:${\scriptstyle y}$}]{}--(i)--++(0:-1)node(i1)[label={[label distance=1pt]90:${\scriptstyle y_1''}$}]{} (k)--++(0:1)node(k1)[label={[label distance=1pt]90:${\scriptstyle y_1}$}]{} (j)--++(90:1)node[label={[label distance=1pt]180:${\scriptstyle x_0}$}]{};
\draw[loosely dotted, thick] (i1)--++(0:-2)node(i2){} (k1)--++(0:2)node(k2){};
\draw (i2)--++(0:-1)node(i3)[label={[label distance=1pt]90:${\scriptstyle z}$}]{}--++(30:-1)node[label={[label distance=1pt]-90:${\scriptstyle w'}$}]{}--++(30:-1)node[label={[label distance=1pt]180:${\scriptstyle z''}$}]{} (i3)--++(-30:-1)node[label={[label distance=1pt]90:${\scriptstyle w}$}]{}--++(-30:-1)node[label={[label distance=1pt]180:${\scriptstyle z'}$}]{} (i)--++(45:-1)node[label={[label distance=1pt]-90:${\scriptstyle y_1'''}$}]{} (i1)--++(45:-1)node[label={[label distance=1pt]-90:${\scriptstyle x_1'''}$}]{} (i2)--++(45:-1)node{}
(k2)--++(0:1)node(k3)[label={[label distance=1pt]90:${\scriptstyle u}$}]{}--++(30:1)node[label={[label distance=1pt]0:${\scriptstyle v'}$}]{} (k3)--++(-30:1)node[label={[label distance=1pt]0:${\scriptstyle v''}$}]{} (k)--++(-45:1)node[label={[label distance=1pt]-90:${\scriptstyle y_1'}$}]{} (k1)--++(-45:1)node[label={[label distance=1pt]-90:${\scriptstyle x_1'}$}]{} (k2)--++(-45:1)node{};
\coordinate (n3) at ([yshift=-.5cm]l) ;
}

\node[label=below:{$K_{1,1}'$}] at (n1){};
\node[label=below:{$K_{2,1}'$}] at (n2){};
\node[label=below:{$K_{i,j}'$}] at (n3){};
\draw[decoration={brace,raise=13pt,amplitude=7pt},decorate]
  (i3) -- node[above=18pt] {${\scriptstyle i}$ {\footnotesize vertices}} (i1);
\draw[decoration={brace,raise=13pt,amplitude=7pt},decorate]
  (k1) -- node[above=17pt] {${\scriptstyle j}$ {\footnotesize vertices}} (k3);
\end{tikzpicture}
\caption{The class $\mathcal{K}'$ of bigraphs}
\end{figure}
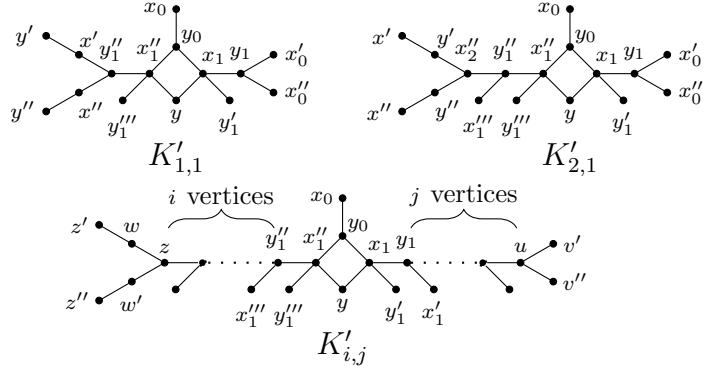

\begin{lem}
Let $i,j\in \mathbb{N}$
\begin{enumerate}[(a)]

\item A $\mathcal{U}$-intersection representation $I:V(K'_{i,j})\to \mathcal{U}$ of $K'_{i,j}$, where $I(V(K'_{i,j}))$ consists of the following intervals
\begin{itemize}
\item $I(x_0)=\left(0,1\right)$, $I(x_1)=[0,1]$, $I(y_0)=\left[-\frac{1}{2},\frac{1}{2}\right]$, $I(y)=[-1,0]$, $I(y_1)=[1,2]$, $I(y_1')=[1,2)$
\item $I(x_n)=[2(n-1),2(n-1)+1]$, $I(y_n)=[2(n-1)+1,2(n-1)+2]=[2n-1,2n]$, $(n\geq 1)$
\item $I(x_n')=[2n,2n+1)$, $I(y_n')=[2n-1,2n)$, $(n\geq 1)$
\item $I(x_n'')=\left[-(2n-1)-\frac{1}{2},-2(n-1)-\frac{1}{2}\right]$, $I(x_n''')=\left(-2n-\frac{1}{2},-(2n-1)-\frac{1}{2}\right]$, $(n\geq 1)$
\item $I(y_n'')=\left[-2n-\frac{1}{2},-(2n-1)-\frac{1}{2}\right]$, $I(y_n''')=\left(-(2n-1)-\frac{1}{2},-2(n-1)-\frac{1}{2}\right]$, $(n\geq 1)$
\item $I(u)=[j,j+1],\ I(v')=I(v'')=[j+1,j+2]$ or $[j+1,j+2)$
\item $I(z)=\left[-i-\frac{3}{2},-i-\frac{1}{2}\right],\ I(w)=\left[-i-\frac{5}{2},-i-\frac{3}{2}\right]$, $I(w')=\left(-i-\frac{3}{2},-i-\frac{1}{2}\right)$,\\ $I(z')=\left[-i-\frac{7}{2},-i-\frac{5}{2}\right]$, $I(z'')=\left(-i-\frac{3}{2},-i-\frac{1}{2}\right)$
\end{itemize}
is unique up to trivial modifications.
\item $K_{i,j}$ is a  forbidden induced subgraph of $\mathcal{U}$-bigraphs.
\end{enumerate}

\end{lem}

\begin{proof}
(a) One can easily observe that in each of the bigraph $K'_{i,j}$ the 
vertices $x_0, y_0, x_1, y_1, x_1'', y_1'', y$ induce $H_2$ as a 
subgraph. Without loss of generality we take the following representation of it (which is a trivial modifications of the first  representation
of $H_2$), where $I(x_1)=[0,1]$, $I(x_0)=(0,1)$, $I(y)=[-1,0]$, 
$I(y_1)=[1,2]$, $I(y_0)=\left[-\frac{1}{2},\frac{1}{2}\right]$, $I(x_1'')=[-\frac{3}{2},-\frac{1}{2}]$, $I(y'')=[-\frac{5}{2},-\frac{3}{2}]$.\par
Next, consider the path $x_1, y_1, x_2, y_2, \ldots x_n, y_n, \ldots u$. This 
path has an interval representation : $I(x_1)=[0,1]$, $I(y_1)=[1,2]$, 
$I(x_2)=[2,3]$, $I(y_2)=[3,4]$, and by induction we have $I(x_n)=[2(n-1),2(n-1)+1]$ 
and $I(y_n)=[2(n-1)+1,2(n-1)+2]=[2n-1,2n]$. If $j=2r\ (r\geq 1)$ then 
$u=x_{r+1}$ and $I(u)=I(x_{r+1})=[2r,2r+1]=[j,j+1]$; again if $j=2r+1\ (r\geq 1)$, 
then $u=y_{r+1}$ and $I(u)=[2r+1,2r+2]=[j,j+1]$. The intervals $I(v')$ and 
$I(v'')$ may be taken as $[j+1,j+2]$ or $[j+1,j+2)$.\par
Since $y_n'$ is adjacent to $x_n$ only, $I(y_n')=[2n-1,2n)\ (n\geq 1)$. 
Similarly $x_n'$ is adjacent to $y_n$ only, $I(x_n')=[2n,2n+1)\ (n\geq 1)$. 
Next, consider the path $x_1'', y_1'', x_2'', y_2'', \ldots x_n'', y_n'', \ldots\
z$. This path has the interval representation given by 
$I(x_1'')=\left[-\frac{3}{2},-\frac{1}{2}\right]$, $I(y_1'')=\left[-\frac{5}{2},-\frac{3}{2}\right]$ and by induction $I(x_n'')=\left[-(2n-1)-\frac{1}{2},-2(n-1)-\frac{1}{2}\right]$ and $I(y_n'')=\left[-2n-\frac{1}{2},
-(2n-1)-\frac{1}{2}\right]$. If $i=2r\ (r\geq 1)$ then $z=x_{r+1}$ and $I(z)=
\left[-(2(r+1)-1)-\frac{1}{2}\right.,$ $\left.
-(2r+1-1)-\frac{1}{2}\right] =\left[-2r-\frac{3}{2}, -2r-
\frac{1}{2}\right]=\left[-i-\frac{3}{2},-i-\frac{1}{2}\right]$. \newline  Again if 
$i=2r+1\ (r\geq 1)$, then $z=y_{r+1}''$ and $I(z)=\left[-2(r+1)-\frac{1}{2}, -(2r+2-1)-\frac{1}{2}\right]=\left[-2r-\frac{5}{2}, -2r-\frac{3}{2}\right]=\left[-i-\frac{3}{2},-i-\frac{1}{2}\right]$, $I(w)=\left[-i-\frac{5}{2},-i-\frac{3}{2}\right]$, $I(z')=\left[-i-\frac{7}{2},-i-\frac{5}{2}\right]$, $I(w')=I(z'')=\left(-i-\frac{3}{2},-i-\frac{1}{2}\right)$.\par
Next, $y_1'''$ is adjacent to $x_1''$ only, we take $I(y_1''')$ as the 
open-closed copy of $I(x_1'')$. Also $I(y_2''')$ is the open-closed copy of 
$I(x_2'')$; therefore by induction we have $I(y_n''')=\left(-(2n-1)-\frac{1}{2}\right. ,$ $\left.
-2(n-1)-\frac{1}{2}\right]$. Similarly, as $x_n'''$ is adjacent to $y_n''$ 
only, $I(x_n''')=\left(-2n-\frac{1}{2}\right.,$ $\left.-(2n-1)-\frac{1}{2}\right]$. Thus 
the proof of (a) is complete.\par 
(b) From the above representation of $K'_{i,j}$, it implies that 
$K_{i,j}$ is a  forbidden induced subgraph of $\mathcal{U}$-bigraphs.
\end{proof}
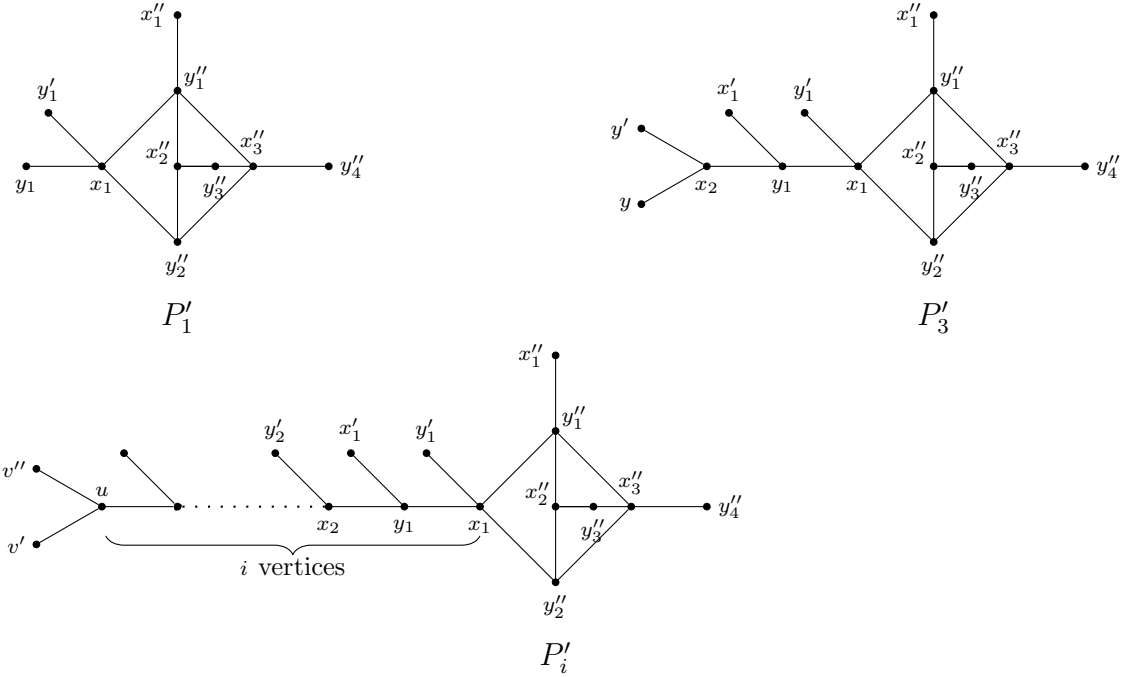
\begin{figure}[H]
\centering
\begin{tikzpicture}[scale=1]
{\tikzstyle{every node}=[circle, draw, fill=black, inner sep=0pt, minimum width= 2.5pt]

\draw (0:0)node(q)[label={[label distance=1pt]-90:${\scriptstyle y_2''}$}]{}--++(90:1)node(r)[label={[label distance=1pt]170:${\scriptstyle x_2''}$}]{}--++(90:1)node(s)[label={[label distance=1pt]10:${\scriptstyle y_1''}$}]{}--++
(90:1)node[label={[label distance=1pt]180:${\scriptstyle x_1''}$}]{} (r)--++(0:1)node(t)[label={[label distance=1pt]90:${\scriptstyle x_3''}$}]{}--++(0:1)node[label={[label distance=1pt]0:${\scriptstyle y_4''}$}]{}
(q)--(t)--(s);

\draw (r)--++(0:.5)node(r')[label={[label distance=.1pt]-90:${\scriptstyle y_3''}$}]{};

\path [name path=v1] (q) --++ (-45:-2.2);
\path [name path=w1] (s) --++ (45:-2.2);
\path [name intersections={of=v1 and w1, by=u}];
\draw (q)--(u)node[label={[label distance=1pt]-90:${\scriptstyle x_1}$}]{}--(s) (u)--++(0:-1)node(u1)[label={[label distance=1pt]-90:${\scriptstyle y_1}$}]{}
(u)--++(-45:-1)node[label={[label distance=1pt]90:${\scriptstyle y_1'}$}]{} ;
\coordinate (n1) at ([yshift=-.5cm]q);

\coordinate (a) at ([xshift=10cm]q);

\draw (a)node[label={[label distance=1pt]-90:${\scriptstyle y_2''}$}]{}--++(90:1)node(b)[label={[label distance=1pt]170:${\scriptstyle x_2''}$}]{}--++(90:1)node(c)[label={[label distance=1pt]10:${\scriptstyle y_1''}$}]{}--++
(90:1)node[label={[label distance=1pt]180:${\scriptstyle x_1''}$}]{} (b)--++(0:1)node(d)[label={[label distance=1pt]90:${\scriptstyle x_3''}$}]{}--++(0:1)node[label={[label distance=1pt]0:${\scriptstyle y_4''}$}]{}
(a)--(d)--(c);

\draw (b)--++(0:.5)node(b')[label={[label distance=.1pt]-90:${\scriptstyle y_3''}$}]{};

\path [name path=v1] (a) --++ (-45:-2.2);
\path [name path=w1] (c) --++ (45:-2.2);
\path [name intersections={of=v1 and w1, by=e}];
\draw (a)--(e)node[label={[label distance=1pt]-90:${\scriptstyle x_1}$}]{}--(c) (e)--++(0:-1)node[label={[label distance=1pt]-90:${\scriptstyle y_1}$}](e1){}--++(0:-1)node(e2)[label={[label distance=1pt]-90:${\scriptstyle x_2}$}]{}--++(30:-1)node[label={[label distance=1pt]180:${\scriptstyle y}$}]{}
(e)--++(-45:-1)node[label={[label distance=1pt]90:${\scriptstyle y_1'}$}]{} (e1)--++(-45:-1)node[label={[label distance=1pt]90:${\scriptstyle x_1'}$}]{} (e2)--++(-30:-1)node[label={[label distance=1pt]180:${\scriptstyle y'}$}]{} ;
\coordinate (n2) at ([yshift=-.5cm]a);

\coordinate (a1) at ([xshift=5cm,yshift=-4.5cm]q);
\draw (a1)node[label={[label distance=1pt]-90:${\scriptstyle y_2''}$}]{}--++(90:1)node(b)[label={[label distance=.5pt]170:${\scriptstyle x_2''}$}]{}--++(90:1)node(c)[label={[label distance=1pt]10:${\scriptstyle y_1''}$}]{}--++
(90:1)node[label={[label distance=1pt]180:${\scriptstyle x_1''}$}]{} (b)--++(0:1)node(d)[label={[label distance=1pt]90:${\scriptstyle x_3''}$}]{}--++(0:1)node[label={[label distance=1pt]0:${\scriptstyle y_4''}$}]{}
(a1)--(d)--(c);

\draw (b)--++(0:.5)node(b')[label={[label distance=1pt]-90:${\scriptstyle y_3''}$}]{};

\path [name path=v1] (a1) --++ (-45:-2.2);
\path [name path=w1] (c) --++ (45:-2.2);
\path [name intersections={of=v1 and w1, by=e}];
\draw (a1)--(e)node[label={[label distance=1pt]-90:${\scriptstyle x_1}$}]{}--(c) (e)--++(0:-1)node(e1)[label={[label distance=1pt]-90:${\scriptstyle y_1}$}]{}--++(0:-1)node(e2)[label={[label distance=1pt]-90:${\scriptstyle x_2}$}]{};
\draw[loosely dotted, thick] (e2)--++(0:-2)node(e3){};

\draw (e3)node{}--++(0:-1)node(e4)[label={[label distance=1pt]90:${\scriptstyle u}$}]{} (e)--++(-45:-1)node[label={[label distance=1pt]90:${\scriptstyle y_1'}$}]{} (e1)--++(-45:-1)node[label={[label distance=1pt]90:${\scriptstyle x_1'}$}]{} (e2)--++(-45:-1)node[label={[label distance=1pt]90:${\scriptstyle y_2'}$}]{} (e3)--++(-45:-1)node{} (e4)--++(-30:-1)node[label={[label distance=1pt]180:${\scriptstyle v''}$}]{} (e4)--++(30:-1)node[label={[label distance=1pt]180:${\scriptstyle v'}$}]{};
\coordinate (n3) at ([yshift=-.5cm]a1);
}

\node[label=below:{$P_{1}'$}] at (n1){};
\node[label=below:{$P_{3}'$}] at (n2){};
\node[label=below:{$P_{i}'$}] at (n3){};
\draw[decoration={brace,mirror,raise=11pt,amplitude=7pt},decorate]
  (e4) -- node[below=15pt] {${\scriptstyle i}$ {\footnotesize vertices}} (e);
\end{tikzpicture}
\caption{The class $\mathcal{P}'$ of bigraphs}
\end{figure}

\begin{lem}
Let $i\in \mathbb{N}$
\begin{enumerate}[(a)]

\item A $\mathcal{U}$-intersection representation $I:V(P'_{i})\to \mathcal{U}$ of $P'_{i}$, where $I(V(P'_{i}))$ consists of the following intervals
\begin{itemize}
\item $I(x_1)=\left[0,1\right]$, $I(x''_1)=(0,1)$, $I(y_1'')=\left[-\frac{1}{2},\frac{1}{2}\right]$, $I(y_1)=[1,2]$, $I(y_2'')=[-1,0]$, 
$I(x_3'')=[-1,0]$, $I(y_4'')=[-2,-1]$
\item $I(x_2'')=[-\frac{1}{2},\frac{1}{2}]$, $I(y_3'')=[-\frac{3}{2},-\frac{1}{2}]$
\item $I(x_n)=[2(n-1),2n-1]$, $I(y_n)=[2n-1,2n]$, $(n\geq 1)$
\item $I(x_n')=\left[2n,2n+1\right)$, $I(y_n')=\left[2n-1,2n\right)$, $(n\geq 1)$
\item $I(u)=\left[i-1,i\right]$, $I(v')$ and $I(v'')$ are $[i,i+1]$ or 
$[i,i+1)$.
\end{itemize}
is unique up to trivial modifications.
\item $P_{i}$ is a  forbidden induced subgraph of $\mathcal{U}$-bigraphs.
\end{enumerate}

\end{lem}

\begin{proof}
(a) One can observe that in each of the bigraph $P_i'$, the 
vertices $x_1'', y_1'', x_1, y_1, y''_2, x''_3, y''_4$ induce $H_2$ 
as a subgraph. We consider the following representation of $H_2$ (which is a 
trivial modification of the first representation of Fig.~3), where $I(x''_1)=
(0,1)$, $I(x_1)=[0,1]$, $I(y_2'')=[-1,0]$, $I(x_3'')=[-1,0]$, $I(y_1)=[1,2]$, $I(y_1'')=\left[ -\frac{1}{2},\frac{1}{2}\right]$, $I(y_4'')=[-2,-1]$.
Since $x_2''$ is adjacent to $y_1''$ and $y_2''$ we can take $I(x_2'')=\left[
-\frac{1}{2},\frac{1}{2}\right]$.
Next, we consider the path $x_1, y_1, x_2, y_2, \ldots x_n, y_n, \ldots u$. It has 
the representation $I(x_1)=[0,1]$, $I(y_1)=[1,2]$, $I(x_2)=[2,3]$, 
$I(y_2)=[3,4]$ and by induction $I(x_n)=[2(n-1),2n-1]$ and 
$I(y_n)=[2n-1,2n]$. Now, if $i$ is even, say $i=2r\ (r\geq 1)$ then 
$u=y_{r}$ and $I(u)=I(y_{r})=[2r-1,2r]=[i-1,i]$. In the 
other case, if $i=2r+1\ (r\geq 1)$, $u=x_{r+1}$ and $I(x_{r+1})=[2(r+1-1),2(r+1)-1]=[2r,2r+1]
=[i-1,i]$. Thus the intervals $I(v')$ and $I(v'')$ are $[i,i+1]$ or 
$[i,i+1)$.\par
Next, consider the vertex $y_n'$, it is adjacent to $x_n$ only. So, we take 
$I(y_n')$ as the closed-open copy of $I(y_n)$, i.e. $I(y_n')=[2n-1,2n)$. 
Similarly, $I(x_n')$ is the closed-open copy of $I(x_{n+1})$, i.e. $I(x_n')=
[2n,2n+1)$. Therefore the proof of (a) is complete.\par
(b) From the above representation of $P_i'$, it follows that $P_i$ is a  
forbidden induced subgraph of $\mathcal{U}$-bigraphs.
\end{proof}
\begin{lem}
The bigraph $F_1$ has unique $\mathcal{U}$-intersection up to trivial modifications.
\end{lem}
\begin{figure}[H]
\centering
\begin{tikzpicture}[scale=.40]
{\tikzstyle{every node}=[circle, draw, fill=black, inner sep=0pt, minimum width= 2.5pt]

\draw (0:0)node(a)[label={[label distance=1pt]-180:${\scriptstyle x_2}$}]{}--++(45:1)node[label={[label distance= 1pt]180:${\scriptstyle y_1}$}]{}--++(45:1)node(b)[label={[label distance= 1pt]0:${\scriptstyle x_1}$}]{}--++(90:1)node[label={[label distance= 1pt]0:${\scriptstyle y_0}$}]{}--++(90:1)node[label={[label distance= 1pt]0:${\scriptstyle x_0}$}]{} (b)--++(-45:1)node[label={[label distance= 1pt]0:${\scriptstyle y_2}$}]{}--++(-45:1)node[label={[label distance= 1pt]0:${\scriptstyle x_3}$}]{} (b)--++(-90:1)node[label={[label distance= 1pt]-90:${\scriptstyle y_3}$}]{};
}
\pgfmathsetmacro{\b}{1.5}
\pgfmathsetmacro{\c}{0.6}
\pgfmathsetmacro{\d}{0.13}
\draw \foreach \p/\q/\r in {6*\b/3*\c/x_1,6*\b/2*\c/y_0,5*\b/1*\c/y_1,7*\b/1*\c/y_2,8*\b/4*\c/x_3,4*\b/4*\c/x_2}
{
(\p,\q)--(\p+\b,\q)
(\p+\d,\q-\d)--(\p,\q-\d)--(\p,\q+\d)--(\p+\d,\q+\d)
(\p+\b-\d,\q-\d)--(\p+\b,\q-\d)--(\p+\b,\q+\d)--(\p+\b-\d,\q+\d)
(\p+0.5*\b,\q+0.25) node{${\scriptstyle \r}$}
};
\draw \foreach \p/\q/\r in {5*\b/0*\c/y_3}
{
(\p,\q)--(\p+\b,\q)
(\p-\d,\q-\d)--(\p,\q-\d)--(\p,\q+\d)--(\p-\d,\q+\d)
(\p+\b-\d,\q-\d)--(\p+\b,\q-\d)--(\p+\b,\q+\d)--(\p+\b-\d,\q+\d)
(\p+0.5*\b,\q+0.25) node{${\scriptstyle \r}$}
};
\draw \foreach \p/\q/\r in {6*\b/4*\c/x_0}
{
(\p,\q)--(\p+\b,\q)
(\p-\d,\q-\d)--(\p,\q-\d)--(\p,\q+\d)--(\p-\d,\q+\d)
(\p+\b+\d,\q-\d)--(\p+\b,\q-\d)--(\p+\b,\q+\d)--(\p+\b+\d,\q+\d)
(\p+0.5*\b,\q+0.25) node{${\scriptstyle \r}$}
};

\end{tikzpicture}
\caption{The graph $F_1$ and its $\mathcal{U}$-representation.}\label{f13}
\end{figure}
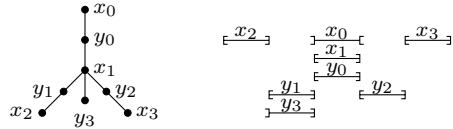
\begin{proof}
The proof follows from the Proposition 4 of \cite{ds}. The bigraph $F_1-y_3$ is the bipartite claw $(H_1)$. Thus from that Proposition, $H_1$ has a unique $\mathcal{U}$-intersection as shown in~\cref{f2}. Since the vertex $y_3$ is adjacent to $x_1$ only, so we can take $I(y_3)$ as in~\cref{f13} to get the $\mathcal{U}$-intersection representation of $F_1$. Again $I(y_3)$ can be taken as closed-open copy of $I(y_2)$ and make some trivial modifications to get the same representation of $F_1$ as earlier. This completes the proof.
\end{proof}
\begin{par}
\textbf{Observation.} From the representation of $F_1$ (\cref{f13}) it follows that the bigraph $B_1$ has no $\mathcal{U}$-representation. \end{par} 
\begin{figure}[H]
\centering
\begin{tikzpicture}[scale=.50]
{\tikzstyle{every node}=[circle, draw, fill=black, inner sep=0pt, minimum width= 2.5pt]

\draw (0:0)node(j)[label={[label distance=1pt]0:${\scriptstyle y}$}]{}--++(90:1)node(k)[label={[label distance=1pt]-135:${\scriptstyle x}$}]{}--++(90:1)node[label={[label distance=1pt]50:${\scriptstyle y_0}$}]{}--++(90:1)node[label={[label distance=1pt]180:${\scriptstyle x_0}$}]{} (k)--++(0:-1)node(k2)[label={[label distance=1pt]-90:${\scriptstyle y_1''}$}]{}--++(0:-1)node(l)[label={[label distance=1pt]-90:${\scriptstyle x_1''}$}]{}--++(0:-1)node[label={[label distance=1pt]180:${\scriptstyle y_2'}$}]{} (k2)--++(-50:-1)node[label={[label distance=1pt]100:${\scriptstyle x_1'''}$}]{}--++(0:-1)node[label={[label distance=1pt]180:${\scriptstyle y_2'''}$}]{}--(l) (k)--++(0:1)node(m2)[label={[label distance=1pt]-90:${\scriptstyle y_1}$}]{}--++(0:1)node(n)[label={[label distance=1pt]-90:${\scriptstyle x_1}$}]{}--++(0:1)node[label={[label distance=1pt]0:${\scriptstyle y_2}$}]{} (m2)--++(50:1)node[label={[label distance=1pt]50:${\scriptstyle x_1'}$}]{}--++(0:1)node[label={[label distance=1pt]0:${\scriptstyle y_2'}$}]{}--(n) (m2)--++(90:1)node[label={[label distance=1pt]90:${\scriptstyle x'}$}]{};
\coordinate (n3) at ([yshift=-.01cm]j);

\coordinate (a) at ([xshift=10cm]j);
\draw (a)node[label={[label distance=1pt]0:${\scriptstyle y}$}]{}--++(90:1)node(b)[label={[label distance=1pt]-135:${\scriptstyle x}$}]{}--++(90:1)node[label={[label distance=1pt]50:${\scriptstyle y_0}$}]{}--++(90:1)node[label={[label distance=1pt]180:${\scriptstyle x_0}$}]{} (b)--++(0:-1)node(b1)[label={[label distance=1pt]-90:${\scriptstyle y_1''}$}]{}--++(0:-1)node[label={[label distance=1pt]-90:${\scriptstyle x_1''}$}](c){}--++(0:-1)node(c1)[label={[label distance=1pt]-90:${\scriptstyle y_2''}$}]{}--++(0:-1)node[label={[label distance=1pt]180:${\scriptstyle x_2''}$}]{} (c)--++(-50:-1)node[label={[label distance=1pt]100:${\scriptstyle y_2'''}$}]{}--++(0:-1)node[label={[label distance=1pt]180:${\scriptstyle x_2'''}$}]{}--(c1) (b1)--++(90:1)node[label={[label distance=.5pt]130:${\scriptstyle x_1'''}$}]{} (b)--++(0:1)node(m2)[label={[label distance=1pt]-90:${\scriptstyle y_1}$}]{}--++(0:1)node(n)[label={[label distance=1pt]-90:${\scriptstyle x_1}$}]{}--++(0:1)node[label={[label distance=1pt]0:${\scriptstyle y_2}$}]{} (m2)--++(50:1)node[label={[label distance=1pt]50:${\scriptstyle x_1'}$}]{}--++(0:1)node[label={[label distance=1pt]0:${\scriptstyle y_2'}$}]{}--(n) (m2)--++(90:1)node[label={[label distance=1pt]90:${\scriptstyle x'}$}]{};
\coordinate (n1) at ([yshift=-.01cm]a);

\coordinate (p) at ([xshift=4.5cm,yshift=-5cm]j);
\draw (p)node[label={[label distance=1pt]0:${\scriptstyle y}$}]{}--++(90:1)node(q)[label={[label distance=1pt]-135:${\scriptstyle x}$}]{}--++(90:1)node[label={[label distance=1pt]50:${\scriptstyle y_0}$}]{}--++(90:1)node[label={[label distance=1pt]0:${\scriptstyle x_0}$}]{} (q)--++(0:-1)node(q1)[label={[label distance=1pt]-90:${\scriptstyle y_1''}$}]{}--++(0:-1)node(q2)[label={[label distance=1pt]-90:${\scriptstyle x_1''}$}]{} (q1)--++(90:1)node[label={[label distance=1pt]90:${\scriptstyle x_1'''}$}]{} (q2)--++(90:1)node[label={[label distance=1pt]90:${\scriptstyle y_1'''}$}]{} (q)--++(0:1)node(s1)[label={[label distance=1pt]-90:${\scriptstyle y_1}$}]{}--++(0:1)node(s2)[label={[label distance=1pt]-90:${\scriptstyle x_1}$}]{} (s2)--++(90:1)node[label={[label distance=1pt]90:${\scriptstyle y_1'}$}]{} (s1)--++(90:1)node[label={[label distance=1pt]80:${\scriptstyle x_1'}$}]{} ;

\draw[loosely dotted, thick] (q2)--++(0:-2)node(q3){} (s2)--++(0:2)node(s3){};

\draw (q3)--++(0:-1)node(q4)[label={[label distance=1pt]-90:${\scriptstyle z}$}]{}--++(0:-1)node(r)[label={[label distance=1pt]-90:${\scriptstyle w_0}$}]{}--++(0:-1)node[label={[label distance=1pt]180:${\scriptstyle z_0}$}]{} (q4)--++(-50:-1)node[label={[label distance=1pt]90:${\scriptstyle w_0'}$}]{}--++(0:-1)node[label={[label distance=1pt]180:${\scriptstyle z_0'}$}]{}--(r) (q3)--++(90:1)node{}  (s3)--++(0:1)node(s4)[label={[label distance=1pt]-90:${\scriptstyle u}$}]{}--++(0:1)node(t)[label={[label distance=1pt]-90:${\scriptstyle v_0}$}]{}--++(0:1)node[label={[label distance=1pt]0:${\scriptstyle u_0}$}]{} (s4)--++(50:1)node[label={[label distance=1pt]50:${\scriptstyle v_0'}$}]{}--++(0:1)node[label={[label distance=1pt]0:${\scriptstyle u_0'}$}]{}--(t) (s3)--++(90:1)node{} (s4)--++(90:1)node[label={[label distance=1pt]90:${\scriptstyle v'}$}]{};
\coordinate (n4) at ([yshift=-.01cm]p);
}
\node[label=below:{$T'_{2,1}$}] at (n1){};
\node[label=below:{$T'_{1,1}$}] at (n3){};
\node[label=below:{$T'_{i,j}$}] at (n4){};
\draw[decoration={brace,mirror,raise=12pt,amplitude=7pt},decorate]
  (q4) -- node[below=15pt] {${\scriptstyle i}$ {\footnotesize vertices}} (q1);
\draw[decoration={brace,mirror,raise=12pt,amplitude=7pt},decorate]
  (s1) -- node[below=15pt] {${\scriptstyle j}$ {\footnotesize vertices}} (s4);
\end{tikzpicture}
\caption{The class $\mathcal{T}'$ of bigraphs}
\end{figure}
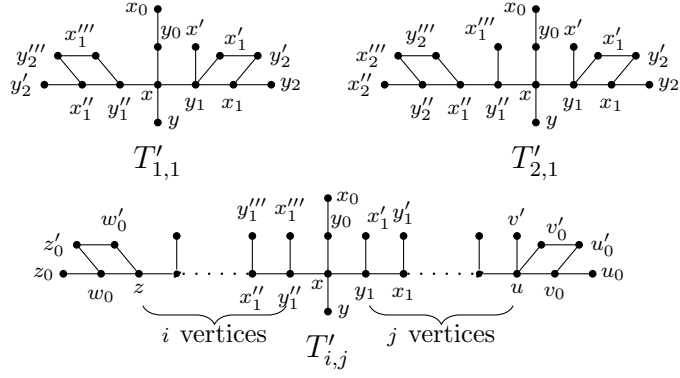

\begin{lem}
Let $i,j\in \mathbb{N}$
\begin{enumerate}[(a)]

\item A $\mathcal{U}$-intersection representation $I:V(T'_{i,j})\to \mathcal{U}$ of $T'_{i,j}$, where $I(V(T'_{i,j}))$ consists of the following intervals
\begin{itemize}
\item $I(x)=\left[0,1\right]$, $I(x_0)=(0,1)$, $I(y_0)=\left(0,1\right)$, $I(y)=(-1,0]$
\item $I(y_n)=[2n-1,2n]$, $I(x_n)=[2n,2n+1]$, $I(x_n')=[2n-1,2n)$, $I(y_n')=[2n,2n+1)$ $(n\geq 1)$
\item $I(y_n'')=[-2(n-1)-1,-2(n-1)]$, $I(x_n'')=[-2(n-1)-2,-2(n-1)-1]$, $(n\geq 1)$
\item $I(x_n''')=\left(-2(n-1)-2,-2(n-1)-1\right]$, $I(y_n''')=\left(-2n-1,-2n\right]$, $(n\geq 1)$
\item $I(u)=\left[j,j+1\right]$, $I(v_0)=\left[j+1,j+2\right]$, 
$I(v_0')=[j+1,j+2)$, $I(u_0)=[j+2,j+3]$, $I(u_0')=(j+1,j+2)$, 
$I(v')=(j,j+1]$
\item $I(z)=[-i,-i+1]$, $I(w_0)=[-i-1,-i]$, $I(w_0')=(-i-1,-i]$, 
$I(z_0')=(-i-1,-i)$, $I(z_0)=[-i-2,-i-1]$

\end{itemize}
is unique up to trivial modifications.
\item $T_{i,j}$ is a  forbidden induced subgraph for the class of $\mathcal{U}$-bigraphs.
\end{enumerate}
\end{lem}

\begin{proof}
(a) Now, one can observe that in each of the bigraph $T_{i,j}'$ the 
vertices $x, x_0, x_1, x_1'', y, y_0, y_1, y_1''$ induce $F_1$ as a 
subgraph. Without loss of generality we consider the following representation of $F_1$, where $I(x)=[0,1]$, $I(x_0)=(0,1)$, $I(y_0)=(0,1)$, 
$I(y_1)=[1,2]$, $I(x_1)=[2,3]$, $I(y_1'')=[-1,0]$, $I(x_1'')=[-2,-1]$, 
$I(y)=(-1,0]$. Next, consider the path $y_1, x_1, y_2, x_2, \ldots y_n, x_n, 
\ldots u, v_0, u_0$. Intervals corresponding to the vertices $y_1,x_1,y_2$ 
and $x_2$ are respectively $I(y_1)=[1,2]$, $I(x_1)=[2,3]$, $I(y_2)=[3,4]$, 
$I(x_2)=[4,5]$ and by induction $I(y_n)=[2n-1,2n]$ and $I(x_n)=[2n,2n+1]\ 
(n\geq 1)$. If $j=2r$, then $u=x_r$ and $I(x_r)=[2r,2r+1]=[j,j+1]$ and if 
$j=2r+1$, then $u=y_{r+1}$ and $I(y_{r+1})=[2(r+1)-1,2(r+1)]=[2r+1,2r+2]=
[j,j+1]$. Thus $I(u)=[j,j+1]$, $I(v_0)=[j+1,j+2]$ and $I(u_0)=[j+2,j+3]$. 
Also, we take $I(v_0')=[j+1,j+2)$, $I(u_0')=(j+1,j+2)$ and as $v'$ is 
adjacent to $u$ only, $I(v')$ can be taken as $(j,j+1]$.\par
Next, consider the path $x, y_1'', x_1'', y_2'', x_2'', \ldots y_n'', x_n'',
\ldots z, w_0, z_0$. This path has the representation given by $I(x)=[0,1]$, 
$I(y_1'')=[-1,0]$, $I(x_1'')=[-2,-1]$, $I(y_2'')=[-3,-2]$, $I(x_2'')=
[-4,-3]$ and by induction $I(y_n'')=[-2(n-1)-1,-2(n-1)]$, $I(x_n'')=
[-2(n-1)-2, -2(n-1)-1]$. Now, if $i=2n$, then $I(z)=I(x_n'')=[-2(n-1)-2,
-2(n-1)-1]=[-2n,-2n+1]=[-i,-i+1]$, and if $i=2n+1$, then $I(z)=I(y_{n+1}'')
=[-2n-1,-2n]=[-i,-i+1]$; thus we have $I(z)=[-i,-i+1]$ in any case. So, 
$I(w_0)=[-i-1,-i]$ and $I(z_0)=[-i-2,-i-1]$. Again, $w_0'$ is adjacent to 
$z$ only in this path, we take $I(w_0')=(-i-1,-i]$, and then $I(z_0')=(-i-1,-i)$. 
Now, $x_n'''$ is adjacent to $y_n''$ only, we have to take $I(x_n''')=
(-2(n-1)-2,-2(n-1)-1]$, and similarly $I(y_n''')=(-2n-1,-2n]$. Thus the
proof of (a) is complete.\par
(b) $T_{i,j}$ is obtained from $T_{i,j}'$ by adjoining a vertex $w'$ to 
$z$ only. Now the above $\mathcal{U}$-intersection representation of 
$T'_{i,j}$, it implies that $T_{i,j}$ is a  forbidden induced 
subgraph of $\mathcal{U}$-bigraph.
\end{proof}

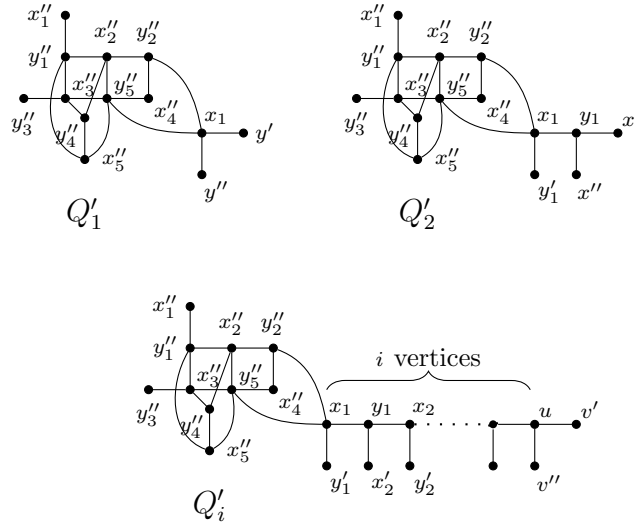
\begin{figure}[H]
\centering
\begin{tikzpicture}[scale=.55]
{\tikzstyle{every node}=[circle, draw, fill=black,
                        inner sep=0pt, minimum width=3pt]

\draw (0:0)node[label={[label distance=1pt]180:${\scriptstyle x_1''}$}]{}--++(-90:1)node(a)[label={[label distance=1pt]180:${\scriptstyle y_1''}$}]{}--++(0:1)node(b)[label={[label distance=1pt]90:${\scriptstyle x_2''}$}]{}--++(0:1)node(c)[label={[label distance=1pt]90:${\scriptstyle y_2''}$}]{}--++(-90:1)
node[label={[label distance=.5pt]-10:${\scriptstyle x_4''}$}]{}--++(180:1)node(d)[label={[label distance=1pt]10:${\scriptstyle y_5''}$}]{}--++(180:1)node(e)[label={[label distance=1pt]10:${\scriptstyle x_3''}$}]{}--++(a) (b)--(d) (e)--++(0:-1)node[label={[label distance=1pt]-90:${\scriptstyle y_3''}$}]{};
\path [name path=b1] (b) --++ (-110:1.7);
\path [name path=e1] (e) --++ (-45:.8);
\path [name intersections={of=b1 and e1, by=f}];

\path [name path=c1] (c) --++ (-55:3);
\path [name path=d1] (d) --++ (-20:3);
\path [name intersections={of=c1 and d1, by=g}];
\draw (b)--(f)node[label={[label distance=1pt]-135:${\scriptstyle y_4''}$}]{}--(e) (f)--++(-90:1)node(h)[label={[label distance=3pt]0:${\scriptstyle x_5''}$}]{} ;
\draw    (h)to[out=30,in=-80](d) (h)to[out=160,in=-120](a) (c)to[out=-20,in=100](g)node[label={[label distance=1pt]50:${\scriptstyle x_1}$}]{} (g)to[out=180,in=-50](d) (g)--++(0:1)node[label={[label distance=1pt]0:${\scriptstyle y'}$}]{} (g)--++(-90:1)node[label={[label distance=.5pt]-60:${\scriptstyle y''}$}]{};
\coordinate (n1) at ([yshift=-.4cm]h);

\coordinate (z) at ([xshift=8cm,yshift=1cm]a);
\draw (z)node[label={[label distance=1pt]180:${\scriptstyle x_1''}$}]{}--++(-90:1)node(a1)[label={[label distance=1pt]180:${\scriptstyle y_1''}$}]{}--++(0:1)node(b)[label={[label distance=1pt]90:${\scriptstyle x_2''}$}]{}--++(0:1)node(c)[label={[label distance=1pt]90:${\scriptstyle y_2''}$}]{}--++(-90:1)
node[label={[label distance=.5pt]-10:${\scriptstyle x_4''}$}]{}--++(180:1)node(d)[label={[label distance=1pt]10:${\scriptstyle y_5''}$}]{}--++(180:1)node(e)[label={[label distance=1pt]10:${\scriptstyle x_3''}$}]{}--++(a1) (b)--(d) (e)--++(0:-1)node[label={[label distance=1pt]-90:${\scriptstyle y_3''}$}]{};
\path [name path=b1] (b) --++ (-110:1.7);
\path [name path=e1] (e) --++ (-45:.8);
\path [name intersections={of=b1 and e1, by=f}];

\path [name path=c1] (c) --++ (-55:3);
\path [name path=d1] (d) --++ (-20:3);
\path [name intersections={of=c1 and d1, by=g}];
\draw (b)--(f)node[label={[label distance=1pt]-135:${\scriptstyle y_4''}$}]{}--(e) (f)--++(-90:1)node(h)[label={[label distance=3pt]0:${\scriptstyle x_5''}$}]{} ;
\draw    (h)to[out=30,in=-80](d) (h)to[out=160,in=-120](a1) (c)to[out=-20,in=100](g)node[label={[label distance=1pt]50:${\scriptstyle x_1}$}]{} (g)to[out=180,in=-50](d) (g)--++(0:1)node(i)[label={[label distance=1pt]50:${\scriptstyle y_1}$}]{}--++(0:1)node[label={[label distance=1pt]50:${\scriptstyle x'}$}]{} (g)--++(-90:1)node[label={[label distance=.5pt]-60:${\scriptstyle y_1'}$}]{} (i)--++(-90:1)node[label={[label distance=.5pt]-60:${\scriptstyle x''}$}]{};
\coordinate (n2) at ([yshift=-.4cm]h);

\coordinate (z) at ([xshift=3cm,yshift=-6cm]a);
\draw (z)node[label={[label distance=1pt]180:${\scriptstyle x_1''}$}]{}--++(-90:1)node(a)[label={[label distance=1pt]180:${\scriptstyle y_1''}$}]{}--++(0:1)node(b)[label={[label distance=1pt]90:${\scriptstyle x_2''}$}]{}--++(0:1)node(c)[label={[label distance=1pt]90:${\scriptstyle y_2''}$}]{}--++(-90:1)
node[label={[label distance=.5pt]-10:${\scriptstyle x_4''}$}]{}--++(180:1)node(d)[label={[label distance=1pt]10:${\scriptstyle y_5''}$}]{}--++(180:1)node(e)[label={[label distance=1pt]10:${\scriptstyle x_3''}$}]{}--++(a) (b)--(d) (e)--++(0:-1)node[label={[label distance=1pt]-90:${\scriptstyle y_3''}$}]{};
\path [name path=b1] (b) --++ (-110:1.7);
\path [name path=e1] (e) --++ (-45:.8);
\path [name intersections={of=b1 and e1, by=f}];

\path [name path=c1] (c) --++ (-55:3);
\path [name path=d1] (d) --++ (-20:3);
\path [name intersections={of=c1 and d1, by=g}];
\draw (b)--(f)node[label={[label distance=1pt]-135:${\scriptstyle y_4''}$}]{}--(e) (f)--++(-90:1)node(h)[label={[label distance=3pt]0:${\scriptstyle x_5''}$}]{} ;
\draw    (h)to[out=30,in=-80](d) (h)to[out=160,in=-120](a) (c)to[out=-20,in=100](g)node[label={[label distance=1pt]50:${\scriptstyle x_1}$}]{} (g)to[out=180,in=-50](d) (g)--++(0:1)node(i)[label={[label distance=1pt]50:${\scriptstyle y_1}$}]{}--++(0:1)node(j)[label={[label distance=1pt]50:${\scriptstyle x_2}$}]{}--++(-90:1)node(k)[label={[label distance=1pt]-60:${\scriptstyle y_2'}$}]{} (g)--++(-90:1)node[label={[label distance=.5pt]-60:${\scriptstyle y_1'}$}]{} (i)--++(-90:1)node[label={[label distance=.5pt]-60:${\scriptstyle x_2'}$}]{};
\draw[loosely dotted, thick] (j)--++(0:2)node(k){};
\draw (k)--++(-90:1)node{} (k)--++(0:1)node(l)[label={[label distance=1pt]50:${\scriptstyle u}$}]{}--++(0:1)node[label={[label distance=1pt]50:${\scriptstyle v'}$}]{} (l)--++(-90:1)node[label={[label distance=1pt]-60:${\scriptstyle v''}$}]{};

\coordinate (n3) at ([yshift=-.4cm]h);

}
\draw[decoration={brace,raise=13pt,amplitude=7pt},decorate]
  (g) -- node[above=18pt] {${\scriptstyle i}$ {\footnotesize vertices}} (l);

\node[label=below:{$Q_1'$}] at (n1){};
\node[label=below:{$Q_2'$}] at (n2){};
\node[label=below:{$Q_i'$}] at (n3){};

\end{tikzpicture}
\caption{The class $\mathcal{Q'}$ of bigraphs.}
\end{figure}

\begin{lem}
Let $i\in \mathbb{N}$
\begin{enumerate}[(a)]

\item A $\mathcal{U}$-intersection representation $I:V(Q'_{i})\to \mathcal{U}$ of $Q'_{i}$, where $I(V(Q'_{i}))$ consists of the following intervals
\begin{itemize}
\item $I(x''_1)=\left(-1,0\right)$, $I(y_1'')=I(x_3'')=\left[-1,0\right]$, $I(y_3'')=[-2,-1]$
\item $I(x_2'')=I(y_5'')=[0,1]$, $I(x_5'')=[0,1)=I(y_4'')$
\item $I(y_2'')=[1,2]$, $I(x_4'')=[1,2)$
\item $I(x_n)=\left[2n-1,2n\right]$, $I(y_n)=\left[2n,2n+1\right]$, $(n\geq 1)$
\item $I(y_n')=\left[2n,2n+1\right)$, $(n\geq 1)$ and $I(x_n')=\left[2n-1,2n\right)$, $(n\geq 2)$
\item $I(u)=[i,i+1]$, $I(v')$ and $I(v'')=[i+1,i+2]$ or $[i+1,i+2)$
\end{itemize}
is unique up to trivial modifications.
\item $Q_{i}$ is a  forbidden induced subgraph of $\mathcal{U}$-bigraphs.
\end{enumerate}
\end{lem}

\begin{proof}
(a) One can observe	 that in each of the bigraph $Q_i$, the 
vertices $x_1'', x_2'', x_3'', y_1'', y_2'', y_3''$ and $y_4''$  induce 
$H_2$ as a subgraph. Without loss of generality from Proposition~5 of \cite{ds} we 
take the following intervals for these vertices: $I(y_1'')=I(x_3'')=[-1,0]$, $I(y_3'')=[-2,-1]$ or $(-2,-1]$, $I(x_2'')=[0,1]$, $I(x_1'')=(-1,0)$, 
$I(y_2'')=[1,2]$ and $I(y_4'')=[0,1)$. Since $y_5''$ is adjacent to $x_2''$ and 
$x_3''$, we take $I(y_5'')=[0,1]$. Now, $x_5''$ is adjacent to $y_1''$, $y_4''$ and 
$y_5''$, $I(x_5'')$ can be taken as $[0,1)$. Again, since $x_4''$ is adjacent to 
$y_2''$ and $y_5''$ only, we take $I(x_4'')=[1,2)$. Also, $x_1$ is adjacent 
to $y_2''$ and $y_5''$ and to other vertices, we take $I(x_1)=[1,2]$. Next, we take the interval representation of the path $x_1, y_1, x_2, y_2, \ldots x_n, y_n, \ldots u$; where $I(x_1)=[1,2]$, $I(y_1)=[2,3]$, $I(x_2)=[3,4]$, $I(y_2)=[4,5]$. By induction, $I(x_n)=[2n-1,2n]$ and $I(y_n)=[2n,2n+1]$. Now, if 
$i=2r$, then $u=y_r$. So $I(u)=I(y_r)=[2r,2r+1]=[i,i+1]$. Again, if $i=2r+1$, 
then $u=x_{r+1}$. Thus $I(u)=I(x_{r+1})=[2r+2-1,2r+2]=[2r+1,2r+2]=[i,i+1]$. 
Since $v$ and $v'$ are adjacent to $u$, we take $I(v)=I(v')=[i+1,i+2]$ or 
$[i+1,i+2)$. Again, since $y_n'$ is adjacent to $x_n$ only, we take $I(y_n')=[2n,2n+1)$, $(n\geq 1)$. Similarly, as $x_n'$ is adjacent to $y_{n-1}$ only for 
$n\geq 2$, we take $I(x_n')=[2n-1,2n)$, $(n\geq 2)$. Hence the proof 
of (a) is complete.\par
(b) From the above representation of $Q_i'$, it implies that $Q_i$ is a  
forbidden induced subgraphs for $\mathcal{U}$-bigraphs.
\end{proof}

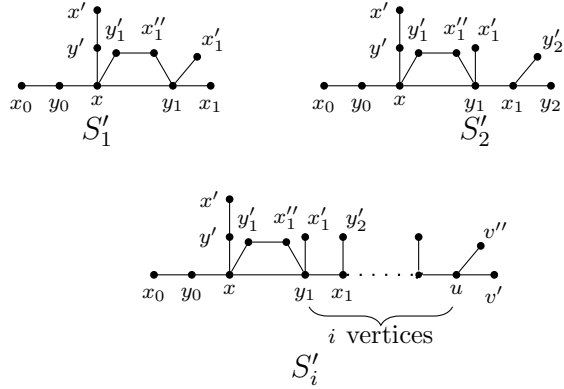
\begin{figure}[H]
\centering
\begin{tikzpicture}[scale=.50]
{\tikzstyle{every node}=[circle, draw, fill=black, inner sep=0pt, minimum width= 2.5pt]

\draw (0:0)node(a)[label={[label distance=1pt]-90:${\scriptstyle x_0}$}]{}--++(0:1)node[label={[label distance=1pt]-90:${\scriptstyle y_0}$}]{}--++(0:1)node(b)[label={[label distance=1pt]-90:${\scriptstyle x}$}]{}--++(90:1)node[label={[label distance=1pt]180:${\scriptstyle y'}$}]{}--++(90:1)node[label={[label distance=1pt]180:${\scriptstyle x'}$}]{} (b)--++(60:1)node(d)[label={[label distance=1pt]90:${\scriptstyle y_1'}$}]{} (b)--++(0:2)node(c)[label={[label distance=1pt]-90:${\scriptstyle y_1}$}]{}--++(120:1)node(e)[label={[label distance=1pt]90:${\scriptstyle x_1''}$}]{}--(d) (c)--++(0:1)node(f)[label={[label distance=1pt]-90:${\scriptstyle x_1}$}]{} (c)--++(50:1)node[label={[label distance=1pt]50:${\scriptstyle x_1'}$}]{};
\coordinate (n2) at ([yshift=-.2cm]b);

\coordinate (g) at ([xshift=8cm]a);
\draw (g)node[label={[label distance=1pt]-90:${\scriptstyle x_0}$}]{}--++(0:1)node[label={[label distance=1pt]-90:${\scriptstyle y_0}$}]{}--++(0:1)node(h)[label={[label distance=1pt]-90:${\scriptstyle x}$}]{}--++(90:1)node[label={[label distance=1pt]180:${\scriptstyle y'}$}]{}--++(90:1)node[label={[label distance=1pt]180:${\scriptstyle x'}$}]{} (h)--++(60:1)node(j)[label={[label distance=1pt]90:${\scriptstyle y_1'}$}]{} (h)--++(0:2)node(i)[label={[label distance=1pt]-90:${\scriptstyle y_1}$}]{}--++(0:1)node(l)[label={[label distance=1pt]-90:${\scriptstyle x_1}$}]{} (i)--++(90:1)node[label={[label distance=1pt]70:${\scriptstyle x_1'}$}]{} (i)--++(120:1)node(k)[label={[label distance=1pt]90:${\scriptstyle x_1''}$}]{}--(j) (l)--++(0:1)node(m)[label={[label distance=1pt]-90:${\scriptstyle y_2}$}]{} (l)--++(50:1)node[label={[label distance=1pt]50:${\scriptstyle y_2'}$}]{};
\coordinate (n3) at ([yshift=-.2cm]i);

\coordinate (p) at ([xshift=3.5cm,yshift=-5cm]a);
\draw (p)node[label={[label distance=1pt]-90:${\scriptstyle x_0}$}]{}--++(0:1)node[label={[label distance=1pt]-90:${\scriptstyle y_0}$}]{}--++(0:1)node(h1)[label={[label distance=1pt]-90:${\scriptstyle x}$}]{}--++(90:1)node[label={[label distance=1pt]180:${\scriptstyle y'}$}]{}--++(90:1)node[label={[label distance=1pt]180:${\scriptstyle x'}$}]{} (h1)--++(60:1)node(j1)[label={[label distance=1pt]90:${\scriptstyle y_1'}$}]{} (h1)--++(0:2)node(i1)[label={[label distance=1pt]-90:${\scriptstyle y_1}$}]{}--++(0:1)node(l1)[label={[label distance=1pt]-90:${\scriptstyle x_1}$}]{}--++(90:1)node[label={[label distance=1pt]70:${\scriptstyle y_2'}$}]{} (i1)--++(90:1)node[label={[label distance=1pt]70:${\scriptstyle x_1'}$}]{} (i1)--++(120:1)node(k1)[label={[label distance=1pt]90:${\scriptstyle x_1''}$}]{}--(j1) ;

\draw[loosely dotted, thick] (l1)--++(0:2)node(m1){} ;

\draw (m1)--++(0:1)node(n1)[label={[label distance=1pt]-90:${\scriptstyle u}$}]{}--++(0:1)node(o1)[label={[label distance=1pt]-90:${\scriptstyle v'}$}]{} (n1)--++(50:1)node[label={[label distance=1pt]50:${\scriptstyle v''}$}]{} (m1)--++(90:1)node{};
}
\coordinate (n4) at ([yshift=-1.5cm]i1);

\node[label=below:{$S'_1$}] at (n2){};
\node[label=below:{$S'_2$}] at (n3){};
\node[label=below:{$S'_i$}] at (n4){};
\draw[decoration={brace,mirror,raise=12pt,amplitude=7pt},decorate]
  (i1) -- node[below=15pt] {${\scriptstyle i}$ {\footnotesize vertices}} (n1);
\end{tikzpicture}
\caption{The class $\mathcal{S}'$ of bigraphs}
\end{figure}
 
\begin{lem}
Let $i\in \mathbb{N}$
\begin{enumerate}[(a)]

\item A $\mathcal{U}$-intersection representation $I:V(S'_i)\to \mathcal{U}$ of $S'_i$, where $I(V(S'_i))$ consists of the following intervals
\begin{itemize}
\item $I(x)=\left[0,1\right]$, $I(y')=I(x')=(0,1)$, $I(y_1)=\left[1,2\right]$, $I(x_1)=[2,3]$
\item $I(y_0)=[-1,0]$, $I(x_0)=[-2,-1]$ or $(-2,-1]$
\item $I(x_n)=[2n,2n+1]$, $I(y_n)=[2n-1,2n]$, $(n\geq 1)$
\item $I(y_n')=\left[2n-1,2n\right)$, $I(x_n')=\left[2n,2n+1\right)$, $(n\geq 1)$, $I(x_1'')=(1,2)$
\item $I(u)=\left[i,i+1\right]$, $I(v')=\left[i+1,i+2\right]$, 
$I(v'')=[i+1,i+2)$ 

\end{itemize}
is unique up to trivial modifications.
\item $S_i$ is a  forbidden induced subgraph of $\mathcal{U}$-bigraphs.
\end{enumerate}
\end{lem}

\begin{proof}
(a) We can observe that in each of the bigraph $S_i'$ the vertices 
$x_0, y_0, x, y', x', y_1, x_1$ induce $H_1$ as a subgraph. From Proposition~4 of \cite{dsml}, $H_1$ has a unique 
$\mathcal{U}$-representation up to trivial modifications. Thus, without loss of generality we consider the representation of $H_1$, such that $I(x)=[0,1]$, $I(x')=I(y')=(0,1)$, $I(y_1)=[1,2]$, 
$I(x_1)=[2,3]$, $I(y_0)=[-1,0]$, $I(x_0)=[-2,-1]$ or $(-2,-1]$. Next, 
consider the representation of the path $y_1, x_1, y_2, x_2, \ldots y_n, x_n,
\ldots u$; where $I(y_1)=[1,2]$, $I(x_1)=[2,3]$, $I(y_2)=[3,4]$, $I(x_2)=[4,5]$.
By induction, we have $I(y_n)=[2n-1,2n]$ and $I(x_n)=[2n,2n+1],\ (n\geq 1)$. 
Now, $I(y_1')$ can be taken $[1,2)$. Similarly, as $x_1'$ is adjacent to $y_1$ only,  
we take $I(x_1')=[2,3)$ and since $y_2'$ is adjacent to $x_1$ only, 
$I(y_2')=[3,4)$. Also, by induction we have $I(y_n')=[2n-1,2n)$ and 
$I(x_n')=[2n,2n+1),\ (n\geq 1)$. Next, $x_1''$ is adjacent to $y_1$ and 
$y_1'$, we take $I(x_1'')=(2n-1,2n)$. Now, if $i=2r$, then $u=x_r$. Then 
$I(x_r)=[2r,2r+1]=[i,i+1]$. In the other case, if $i=2r+1$, then $u=y_{r+1}$ 
and $I(u)=[2(r+1)-1,2(r+1)]=[2r+1,2r+2]=[i,i+1]$. Thus in any case $I(u)=[i,i+1]$ and then $I(v')=[i+1,i+2]$.
As before we take $I(v'')=[i+1,i+2)$. Thus 
the proof of (a) is complete.
\par
(b) From the above representation of $S_i'$, it implies that $S_i$ is a  forbidden induced subgraph for the class of $\mathcal{U}$-bigraphs.
\end{proof}

\begin{figure}[H]
\centering
\begin{tikzpicture}[scale=.50]
{\tikzstyle{every node}=[circle, draw, fill=black, inner sep=0pt, minimum width= 2.5pt]

\draw (0:0)node(a)[label={[label distance=1pt]-90:${\scriptstyle x_1''}$}]{}--++(0:1)node(b)[label={[label distance=1pt]-90:${\scriptstyle y_1''}$}]{}--++(0:1)node(c)[label={[label distance=1pt]-90:${\scriptstyle x}$}]{}--++(90:1)node[label={[label distance=1pt]10:${\scriptstyle y'}$}]{}--++(90:1)node[label={[label distance=1pt]10:${\scriptstyle x'}$}]{} (c)--++(-50:-1)node[label={[label distance=1pt]130:${\scriptstyle y_2''}$}]{}--++(0:-1)node[label={[label distance=1pt]180:${\scriptstyle x''}$}]{}--(b) (c)--++(50:1)node[label={[label distance=1pt]50:${\scriptstyle y_1'}$}]{} (c)--++(0:1)node(d)[label={[label distance=1pt]-90:${\scriptstyle y_1}$}]{}--++(0:1)node(e)[label={[label distance=1pt]-90:${\scriptstyle x_1}$}]{} (d)--++(50:1)node[label={[label distance=1pt]50:${\scriptstyle x_1'}$}]{} ;
\coordinate (n2) at ([yshift=-.4cm]c);

\coordinate (a1) at ([xshift=8cm]a);
\draw (a1)node[label={[label distance=1pt]-90:${\scriptstyle x_1''}$}]{}--++(0:1)node(b1)[label={[label distance=1pt]-90:${\scriptstyle y_1''}$}]{}--++(0:1)node(c1)[label={[label distance=1pt]-90:${\scriptstyle x}$}]{}--++(90:1)node[label={[label distance=1pt]10:${\scriptstyle y'}$}]{}--++(90:1)node[label={[label distance=1pt]10:${\scriptstyle x'}$}]{} (c1)--++(-50:-1)node[label={[label distance=1pt]130:${\scriptstyle y_2''}$}]{}--++(0:-1)node[label={[label distance=1pt]180:${\scriptstyle x''}$}]{}--(b1) (c1)--++(50:1)node[label={[label distance=1pt]50:${\scriptstyle y_1'}$}]{} (c1)--++(0:1)node(d1)[label={[label distance=1pt]-90:${\scriptstyle y_1}$}]{}--++(0:1)node(e1)[label={[label distance=1pt]-90:${\scriptstyle x_1}$}]{}--++(0:1)node(f1)[label={[label distance=1pt]-90:${\scriptstyle y_2}$}]{} (d1)--++(50:1)node[label={[label distance=1pt]50:${\scriptstyle x_1'}$}]{} (e1)--++(50:1)node[label={[label distance=1pt]50:${\scriptstyle y_2'}$}]{};
\coordinate (n3) at ([yshift=-.4cm]d1);

\coordinate (a2) at ([xshift=3.5cm,yshift=-5cm]a);
\draw (a2)node[label={[label distance=1pt]-90:${\scriptstyle x_1''}$}]{}--++(0:1)node(b2)[label={[label distance=1pt]-90:${\scriptstyle y_1''}$}]{}--++(0:1)node(c2)[label={[label distance=1pt]-90:${\scriptstyle x}$}]{}--++(90:1)node[label={[label distance=1pt]10:${\scriptstyle y'}$}]{}--++(90:1)node[label={[label distance=1pt]10:${\scriptstyle x'}$}]{} (c2)--++(-50:-1)node[label={[label distance=1pt]130:${\scriptstyle y_2''}$}]{}--++(0:-1)node[label={[label distance=1pt]180:${\scriptstyle x''}$}]{}--(b2) (c2)--++(50:1)node[label={[label distance=1pt]50:${\scriptstyle y_1'}$}]{} (c2)--++(0:1)node(d2)[label={[label distance=1pt]-90:${\scriptstyle y_1}$}]{}--++(50:1)node[label={[label distance=1pt]50:${\scriptstyle x_1'}$}]{};

\draw[loosely dotted, thick] (d2)--++(0:2)node(e2){} ;

\draw (e2)--++(0:1)node(f2)[label={[label distance=1pt]-90:${\scriptstyle u}$}]{}--++(0:1)node(g2)[label={[label distance=1pt]-90:${\scriptstyle v'}$}]{} (f2)--++(50:1)node[label={[label distance=1pt]50:${\scriptstyle v''}$}]{} (e2)--++(90:1)node{};
}
\coordinate (n4) at ([yshift=-1.5cm]d2);

\node[label=below:{$R'_1$}] at (n2){};
\node[label=below:{$R'_2$}] at (n3){};
\node[label=below:{$R'_i$}] at (n4){};
\draw[decoration={brace,mirror,raise=12pt,amplitude=7pt},decorate]
  (d2) -- node[below=15pt] {${\scriptstyle i}$ {\footnotesize vertices}} (f2);
\end{tikzpicture}
\caption{The class $\mathcal{R}'$ of bigraphs}
\end{figure}
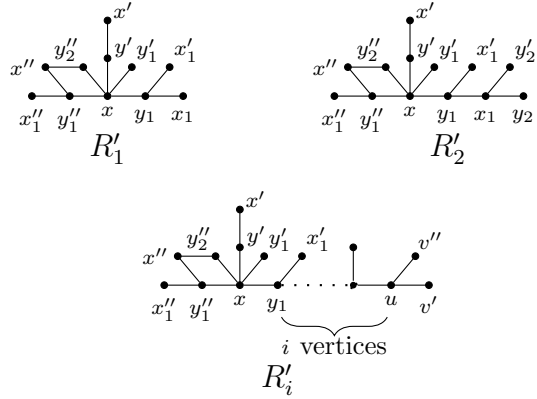

\begin{lem}
Let $i\in \mathbb{N}$
\begin{enumerate}[(a)]

\item A $\mathcal{U}$-intersection representation $I:V(R'_i)\to \mathcal{U}$ of $R'_i$, where $I(V(R'_i))$ consists of the following intervals
\begin{itemize}
\item $I(x)=\left[0,1\right]$, $I(y')=I(x')=(0,1)$, $I(y_1)=\left[1,2\right]$, $I(x_1)=[2,3]$
\item $I(y_1'')=[-1,0]$, $I(x_1'')=[-2,-1]$ or $(-2,-1]$, $I(y_2'')=(-1,0]$, $I(x'')=(-1,0)$.
\item $I(y_n)=[2n-1,2n]$, $I(x_n)=[2n,2n+1]$, $(n\geq 1)$
\item $I(y_n')=\left[2n-1,2n\right)$, $I(x_n')=\left[2n,2n+1\right)$, $(n\geq 1)$
\item $I(u)=\left[i,i+1\right]$, $I(v')=\left[i+1,i+2\right]$, 
$I(v'')=[i+1,i+2)$

\end{itemize}
is unique up to trivial modifications.
\item $R_i$ is a  forbidden induced subgraph for the class of $\mathcal{U}$-bigraphs.
\end{enumerate}
\end{lem}

\begin{proof}
(a) We can observe that in the bigraph $R_i'$ the vertices $x, y_1, x_1, y', x', y_1'', x_1''$ induce $H_1$ as a subgraph. Without loss of generality, (as before using Proposition~4 of \cite{dsml}) we take the following intervals corresponding to these vertices of $H_1$: $I(x)=[0,1]$, $I(y')=I(x')=(0,1)$, $I(y_1)=[1,2]$, 
$I(x_1)=[2,3]$, $I(y_1'')=[-1,0]$ and $I(x_1'')=[-2,-1]$ or $(-2,-1]$. Since $y_2''$ is also adjacent to $x$ we take $I(y_2'')=(-1,0]$ and $x''$ is adjacent $y_1''$ and $y_2''$ only, we take $I(x'')=(-1,0)$. 
Next, in $R_i'$ consider the path $y_1, x_1, y_2, x_2, \ldots\ y_n, x_n,
\ldots u$. The intervals corresponding to these vertices are $I(y_1)=[1,2]$, $I(x_1)=[2,3]$, $I(y_2)=[3,4]$, $I(x_2)=[4,5]$; also by induction $I(y_n)=[2n-1,2n]$ and $I(x_n)=[2n,2n+1],\ (n\geq 1)$. Now, if $i=2r$ (say), 
then $u=x_r$ and $I(x_r)=[2r,2r+1]=[i,i+1]$, and if $i=2r+1$, then $u=y_{r+1}$ and $I(y_{r+1})=[2(r+1)-1,2(r+1)]=[2r+1,2r+2]=[i,i+1]$. Thus, $I(u)=[i,i+1]$ 
and we take $I(v')=[i+1,i+2]$ and $I(v'')=[i+1,i+2)$. Hence the proof of 
(a) is complete.\par
(b) From the above representation of $R_i'$, it implies that $R_i$ is a  forbidden induced subgraph for the class of $\mathcal{U}$-bigraphs.
\end{proof}

\begin{figure}[H]
\centering
\begin{tikzpicture}[scale=.50]
{\tikzstyle{every node}=[circle, draw, fill=black, inner sep=0pt, minimum width= 2.5pt]

\draw (0:0)node(a)[label={[label distance= 1pt]120:${\scriptstyle x_2}$}]{}--++(45:1)node(b)[label={[label distance= 1pt]10:${\scriptstyle y_1}$}]
{}--++(-45:1)node(c)[label={[label distance= 1pt]60:${\scriptstyle x_3}$}]{}--++(45:-1)node(d)[label={[label distance= 1pt]-90:${\scriptstyle y_4}$}]{}--(a)--++(0:-1)node(a1)[label={[label distance= 1pt]120:${\scriptstyle y_2}$}]{}--++(0:-1)node(a2)[label={[label distance= 1pt]180:${\scriptstyle x_4}$}]{} (a)--++(30:-1)node[label={[label distance= 1pt]-45:${\scriptstyle y_2'}$}]{}--++(30:-1)node[label={[label distance= 1pt]180:${\scriptstyle x_4'}$}]{} (c)--++(0:1)node[label={[label distance= 1pt]60:${\scriptstyle y_3}$}]{}--++(0:1)node[label={[label distance= 1pt]0:${\scriptstyle x_5}$}]{} (c)--++(-30:1)node[label={[label distance= 1pt]0:${\scriptstyle y_3'}$}]{} (b)--++(90:1)node[label={[label distance= 1pt]0:${\scriptstyle x_1}$}]{};
}
\end{tikzpicture}
\caption{The bigraph $K$.}
\end{figure}
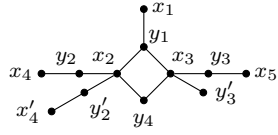
\begin{lem}
The bigraph $K$ is a  forbidden induced subgraph for $\mathcal{U}$-bigraphs.
\end{lem}
\begin{proof}
In the bigraph $K$ the vertices $x_1, x_2, x_3, y_1, y_2, y_3, y_4$ 
induce $H_2$ as a subgraph. Since $H_2$ has a unique 
$\mathcal{U}^\pm$-representation up to trivial modifications \cite{ds}, we 
may consider the following representation of $H_2$ : $I(x_1)=(0,1)$, $I(x_3)=I(y_1)=[0,1]$, $I(x_2)=I(y_4)=[-1,0]$, $I(y_3)=[1,2]$ and $I(y_2)=[-2,-1]$. Now, 
$y_3'$ is adjacent to $x_3$ only, we take $I(y_3')=[1,2)$. Again, 
$x_5$ is adjacent to $y_3$ only, we take $I(x_5)=[2,3]$ or $[2,3)$. 
Similarly,  $I(x_4)=[-3,-2]$ or $(-3,-2]$. Since 
$y_2'$ is adjacent to $x_2$ but not to $x_4$, we take
$I(y_2')=(-2,1]$. Now, the interval
representation for the vertex $x_4'$ does not exist. Hence the proof of 
the lemma is complete.
\end{proof}

\begin{figure}[H]
\centering

\begin{tikzpicture}[scale=.50]
{\tikzstyle{every node}=[circle, draw, fill=black, inner sep=0pt, minimum width= 2.5pt]

\draw (0:0)node[label={[label distance=1pt]180:${\scriptstyle y_3}$}]{}--++(30:1)node(a)[label={[label distance=1pt]-60:${\scriptstyle x_3}$}]{}--++(30:1)node(b)[label={[label distance=1pt]-90:${\scriptstyle y_1}$}]{}--++(-30:-1)node(c)[label={[label distance=1pt]60:${\scriptstyle x_2}$}]{}--++(-30:-1)node[label={[label distance=1pt]180:${\scriptstyle y_2}$}]{} (a)--++(150:1)node[label={[label distance=1pt]180:${\scriptstyle y_4}$}]{}--(c) (b)--++(90:1)node[label={[label distance=1pt]90:${\scriptstyle x_5}$}]{} (b)--++(0:1)node(d)[label={[label distance=1pt]-90:${\scriptstyle x_1}$}]{}--++(30:1)node[label={[label distance=1pt]30:${\scriptstyle y_5}$}]{}--++(0:-1)node[label={[label distance=1pt]60:${\scriptstyle x_4}$}]{}--(b) (d)--++(0:1)node[label={[label distance=1pt]0:${\scriptstyle y_6}$}]{};
}
\end{tikzpicture}

\caption{The bigraph $M$.}

\end{figure}
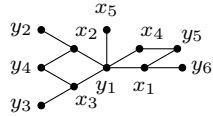

\begin{lem}
The bigraph $M$ is a  forbidden induced subgraph for $\mathcal{U}$-bigraphs.
\end{lem}
\begin{proof}
In the bigraph $M$ the vertices $x_1, x_2, x_3, y_1, y_2, y_3$ and  
$y_4$ induce $H_2$ as subgraph. Now, $H_2$ has a unique 
$\mathcal{U}^\pm$-representation up to trivial modifications . From the 
second representation of Fig.~3 we consider the following 
representation of $H_2$ : $I(x_2) = [1,2]$, $I(x_3)=[1.5,2.5]$, 
$I(y_1)=[2,3]$, $I(y_3)=(2,3)$, $I(y_2)=[0,1]$, $I(y_4)=[3,4]$. Since 
$y_6$ is adjacent to $x_1$ only, we take $I(y_6)=[4,5]$. Next, $x_4$ is 
adjacent to $y_1$ only, we take $I(x_4)=[3,4)$. Again, since $y_5$ is 
adjacent to $x_1$ and $x_4$, we take $I(y_5)=(3,4)$. Now, the interval representation for the vertex $x_5$ is not possible. Thus the proof of
 the lemma is complete.
\end{proof}

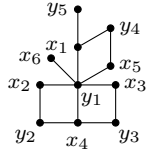
\begin{figure}[H]
\centering

\begin{tikzpicture}[scale=.50]
{\tikzstyle{every node}=[circle, draw, fill=black, inner sep=0pt, minimum width= 2.5pt]

\draw (0:0)node(a)[label={[label distance=1pt]-135:${\scriptstyle y_2}$}]{}--++(0:1)node(b)[label={[label distance=1pt]-90:${\scriptstyle x_4}$}]{}--++(0:1)node[label={[label distance=1pt]-45:${\scriptstyle y_3}$}]{}--++(90:1)node[label={[label distance=1pt]0:${\scriptstyle x_3}$}]{}--++(0:-1)node(c)[label={[label distance=1pt]-45:${\scriptstyle y_1}$}]{}--++(0:-1)node[label={[label distance=1pt]180:${\scriptstyle x_2}$}]{}--(a) (b)--(c)--++(90:1)node(d)[label={[label distance=1pt]180:${\scriptstyle x_1}$}]{}--++(90:1)node[label={[label distance=1pt]180:${\scriptstyle y_5}$}]{} (c)--++(30:1)node[label={[label distance=1pt]0:${\scriptstyle x_5}$}]{}--++(90:1)node[label={[label distance=1pt]0:${\scriptstyle y_4}$}]{}--(d) (c)--++(-45:-1)node[label={[label distance=1pt]180:${\scriptstyle x_6}$}]{};
}
\end{tikzpicture}

\caption{The bigraph $H_0$.}

\end{figure}

\begin{lem}
The bigraph $H_0$ is a  forbidden induced subgraph for 
$\mathcal{U}$-bigraphs.
\end{lem}
\begin{proof}
In $H_0$, the vertices $x_1, x_2, x_3, x_4, y_1, y_2, y_3$ induce $H_3$ 
as a subraph. Now, $H_3$ has a unique 
$\mathcal{U}^\pm$-representation up to trivial modifications \cite{ds}. From the 
second representation of Fig.~4, we consider the following representation of 
$H_3$: $I(x_2)=I(y_2)=[-1,0]$, $I(x_4)=I(y_1)=[0,1]$, $I(x_3)=[0.5,1.5]$, 
$I(y_3)=(0,1)$, $I(x_1)=[1,2]$. Since $x_6$ and $x_5$ are adjacent to 
$y_1$ but $x_5$ is also adjacent to $y_4$, we take $I(x_5)=[1,2]$ and 
$I(x_6)=[1,2)$. Now, $y_4$ is adjacent $x_1$ and $x_5$, we take $I(y_4)=
[2,3]$ or $[2,3)$. Then, the interval representation of $y_5$ is not possible. 
Thus the proof of the lemma is complete.
\end{proof}
\begin{figure}[H]
\centering
\begin{tikzpicture}[scale=.40]
{\tikzstyle{every node}=[circle, draw, fill=black, inner sep=0pt, minimum width= 2.5pt]
\draw (0:0)node(a)[label={[label distance= 1pt]-90:${\scriptstyle u}$}]{}--++(0:1)node(b)[label={[label distance= 1pt]-90:${\scriptstyle v_0''}$}]{}--++(0:1)node[label={[label distance= 1pt]-90:${\scriptstyle u_0'}$}]{} (a)--++(90:1)node[label={[label distance= 1pt]180:${\scriptstyle v_0}$}]{} (a)--++(60:1)node[label={[label distance= 1pt]60:${\scriptstyle v_0'}$}]{}--++(0:1)node[label={[label distance= 1pt]0:${\scriptstyle u_0}$}]{}--(b);
}
\end{tikzpicture}
\caption{The bigraph $B_0$}
\end{figure}
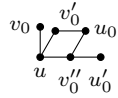
\begin{lem}
In any of the bigraphs $F_{i,j},\ M_i',\ N_i',\ H_i'',\ K_{i,j}',\ P_i',\ R_i',\ S_i',\ Q_i'$; if we delete the vertices $v'$ and $v''$ and then 
form the union of it with the bigraph $B_0$ then the resulting bigraph is also a forbidden induced subgraph for $\mathcal{U}$-bigraphs.
\end{lem}
For the figures and their $\mathcal{U}$- intersection representations of the bigraphs $F_{i,j},\ M_i',\ N_i'$\ and\ $H_i''$ see \cite{dsml}. Now we prove the above lemma.
\begin{proof}
In the $\mathcal{U}$-intersection representation of any of the bigraphs
$F_{i,j}, M_i', N_i', H_i'', K_i', P_i', Q_i', R_i'$ or $S_i'$, let 
the interval corresponding to $u$ is $I(u)=[a,a+1]$. Since $v_0'$ and 
$v_0''$ are adjacent to $u$, $u_0$ is adjacent to $v_0'$ and 
$v_0''$, also $u_0'$ is adjacent to $v_0''$ only, we take intervals 
corresponding to these vertices as follows : $I(v_0'')=[a+1,a+2]$, 
$I(v_0')=[a+1,a+2)$, $I(u_0)=(a+1,a+2)$ and $I(u_0')=[a+2,a+3]$ or 
$[a+2,a+3)$. Now, the interval representations of $v_0$ is not possible 
as there exists an interval $I(u')=[a,a+1)$ in the interval 
representation of each of the bigraphs $F_{i,j}, M_i', N_i', H_i'', K_i', P_i', Q_i', R_i'$ or $S_i'$. Which completes the proof of the lemma.
\end{proof}

Now, we denote the bigraph 
$H\cup B_0$ by $\widetilde{L}_{i,j}$, where $H = F_{i,j}\setminus \{v', v''\}$. Also, we denote the class of 
bigraphs so obtained by $\widetilde{\mathcal{L}}$. In the similar way we 
can obtain the classes $\mathcal{\widetilde{M},\, \widetilde{N},\, \widetilde{H'},\, \widetilde{K},\, \widetilde{P},\, \widetilde{Q},\, \widetilde{R},\, \widetilde{S}}$ of bigraphs.\par

From Lemma~14 we have the following corollary.

\begin{cor}
Any bigraph of the classes $\mathcal{\widetilde{L},\, \widetilde{M},\, \widetilde{N},\, \widetilde{H'},\, \widetilde{K},\, \widetilde{P},\, \widetilde{Q},\, \widetilde{R},\, \widetilde{S}}$ is a forbidden induced subgraph of $\mathcal{U}$-bigraph.
\end{cor}
Let $\mathcal{B'}$ be the class of bigraphs, where $\mathcal{B'}=\mathcal{F'}\cup \mathcal{M'}\cup \mathcal{N'}\cup \mathcal{H''} \cup \mathcal{K'} \cup \mathcal{P'} \cup \mathcal{Q'} \cup \mathcal{R'} \cup \mathcal{S'}$ (See \cite{dsml} for the classes $\mathcal{F'}$,  $\mathcal{M'}$, $\mathcal{N'}$ and $\mathcal{H''}$ of bigraphs).

\begin{lem}
If $B$ is a mixed unit interval bigraph, then $B$ does not contain any 
bigraph of the set $\{ F_2, F_4, F_5, F_8, F_9,  F_{11}, F_{12}, B_1, B_2, K, M, H_0  \} \cup 
 \mathcal{L} \cup \mathcal{M} \cup \mathcal{N} \cup \mathcal{H'} \cup \mathcal{K} \cup \mathcal{P} \cup \mathcal{Q} \cup \mathcal{R} \cup \mathcal{S} 
\cup \mathcal{T} \cup \mathcal{\widetilde{L}} \cup \mathcal{\widetilde{M}} \cup \mathcal{\widetilde{N}} \cup \mathcal{\widetilde{H'}} \cup \mathcal{\widetilde{K}} \cup \mathcal{\widetilde{P}} \cup \mathcal{\widetilde{Q}} \cup \mathcal{\widetilde{R}} \cup \mathcal{\widetilde{S}} $
\end{lem}

\begin{proof}
Let $B$ be a $\mathcal{U}$-bigraph, and let $I$ be a $\mathcal{U}$-intersection representation of $B$. Then, we left as an easy exercise that  $B$ is $F_2, F_4$, $F_5,F_8,F_9,F_{11},F_{12}$-free interval bigraphs, i.e these graphs have no $\mathcal{U}$- intersection representation. Also, Observation of Lemma~6, and Lemma~10 of \cite{dsml} imply that $B$ must not contain 
$B_1$ and $B_2$ as induced subgraph. Again Lemmata~10, 11, and 12 respectively imply that $B$ is $K$, $M$, and $H_0$-free. Also from Lemmata~6, 7, 8, 9 of \cite{dsml} and Lemmata~4, 5, 7, 8, 9 and 6, $B$ is $\mathcal{L}\cup \mathcal{M}\cup \mathcal{N}\cup \mathcal{H'}\cup \mathcal{K}\cup \mathcal{P}\cup \mathcal{Q}\cup \mathcal{R}\cup \mathcal{S}\cup \mathcal{T}$-free interval bigraph.\par
Now, let $H$ be an induced subgraph of $B$ that is isomorphic to any bigraph of the class $\mathcal{B'}$. Also, let the vertices of $H$ be denoted as in the definition of the bigraphs in the class $\mathcal{B'}$. Then, the two pendant vertices $v'$ and $v''$ are special vertices which are adjacent to $u$. Next, we delete the vertices $v'$ and $v''$ from 
$H$ to get the bigraph $H'$. Now, consider the bigraph $H'\cup B_0$ and 
try to give its interval representation from the representation of $H'$.
By Lemmata~6, 7, 8, 9 of \cite{dsml} and Lemmata~4, 5, 6, 7, 8 and 9 we may assume that $I(u)=[a,a+1]$,
where $a\in \mathbb{R}$. Now intervals for $v_0'$ and $v_0''$ are respectively $[a+1,a+2)$ and $[a+1,a+2]$ in the representation of $H'$. 
Since $u_0$ is adjacent to $v_0'$ and $v_0''$, also $u_0'$ is adjacent to 
$v_0''$ only, we take $I(u_0)=(a+1,a+2)$ and $I(u_0')=[a+2,a+3]$ or 
$[a+2,a+3)$. From the Lemma~14, we see that the interval representation 
of $v_0$ is not possible and hence $H'\cup B_0$ is forbidden induced 
subgraph of $\mathcal{U}$-bigraphs. Thus the proof of the Lemma is 
complete.
\end{proof}

\section{Main Result}
Now, we present the characterization theorem of mixed unit interval bigraphs. But, before that as in~\cite{dsml} we introduce a new definition. For notational convenience we write $l(I(v)) = l(v)$ and $r(I(v)) = r(v)$.\par
A bigraph $B=(X,Y,E)$ is a \emph{mixed proper interval bigraph} if it has an $\mathcal{I}$-intersection representation $I:V(B)\to\mathcal{I}$ such that
\begin{enumerate}[(i)]
\item for two distinct vertices $u$ and $v$ of $B$ with $I(u),I(v)\in \mathcal{I}^{++},\ I(u)\not\subset I(v)$ and $I(v)\not\subset I(u)$, and
\item for every vertex $u$ of $B$ with $I(u)\not\in \mathcal{I}^{++}$, there is a vertex $v$ of $B$ with $I(v)\in \mathcal{I}^{++},\ l(u)=l(v)$ and $r(u)=r(v)$, that is no closed interval is properly contained in another closed interval and for any non closed interval, there is a closed interval with same end points.
\end{enumerate}

\begin{theo}
For a bigraph $B$, the following statements are equivalent:
\begin{enumerate}[(a)]
\item $B$ is $\{ F_2, F_4, F_5, F_8, F_9, F_{11}, F_{12}, B_1, B_2, K, M, H_0 \} \cup \mathcal{L} \cup \mathcal{M} \cup \mathcal{N} \cup \mathcal{H'} \cup \mathcal{K} \cup \mathcal{P} \cup \mathcal{Q} \cup \mathcal{R} \cup \mathcal{S} \cup \mathcal{T} \cup \mathcal{\widetilde{L}} \cup  \mathcal{\widetilde{M}} \cup \mathcal{\widetilde{N}} \cup \mathcal{\widetilde{H'}} \cup \mathcal{\widetilde{K}} \cup \mathcal{\widetilde{P}} \cup \mathcal{\widetilde{Q}} \cup \mathcal{\widetilde{R}} \cup \mathcal{\widetilde{S}}$-free interval bigraph.
\item $B$ is a mixed proper interval bigraph.
\item $B$ is a mixed unit interval bigraph.
\end{enumerate}
\end{theo}
\begin{proof}\renewcommand{\qedsymbol}{}
In the Lemma~16 we have proved the implication $(c)\to(a)$. Also, in the Theorem~14 of~\cite{dsml}, we have proved that the statements $(b)$ and $(c)$ are equivalent.  Therefore to prove the Theorem, it remains to prove that if a bigraph $B$ is $\{ F_2, F_4, F_5, F_8, F_9, F_{11}, F_{12}, B_1, B_2, K, M, H_0 \} \cup \mathcal{L} \cup \mathcal{M} \cup \mathcal{N} \cup \mathcal{H'} \cup \mathcal{K} \cup \mathcal{P} \cup \mathcal{Q} \cup \mathcal{R} \cup \mathcal{S} \cup \mathcal{T}  \cup \mathcal{\widetilde{L} \cup  \widetilde{M} \cup \widetilde{N} \cup \widetilde{H'} \cup \widetilde{K} \cup \widetilde{P}} \cup \mathcal{\widetilde{Q}} \cup \mathcal{\widetilde{R}} \cup \mathcal{\widetilde{S}}$-free interval bigraph then $B$ is a mixed proper interval bigraph. The proof is basically similar to the proof of Lemma~8 of~\cite{ds}. The sketch of the proof is as follows. We assume that $I$ is an $\mathcal{I}^{++}$ representation of $B$ with the number of bad pair is as small as possible. If there is no bad pair, then $B$ is proper interval bigraph and hence also a mixed proper interval bigraph. Therefore, we may assume that there is at least one bad pair say, $(u,v)$ where $I(u)$ contained in  $I(v)$ and $u, v \in X\cup Y$. In the Claim~1 to 5 we shall study the structure of the representation $I$ and  the bigraph $B$ corresponding to the occurrence of this bad pair. Then in the claims~6 and 7 we shall describe how the $\mathcal{I}^{++}$-intersection representation of $B$ can be modified to a mixed proper interval representation.
\end{proof}
Note that, in \cite{ds} we have shown that all the intervals of $\mathcal{I}^{++}$ representation of $B$ can be chosen such that the intervals have distinct end points. 
\begin{claim}
If $(u,v)$ is a bad pair, then there are vertices $z_1$ and $z_2$ such that $l(v) < r(z_1) < l(u)$
and $r(u) < l(z_2) < r(v)$, where $u$ is in a partite set different from $z_1$ or $z_2$.
\end{claim}
\begin{proof}
The proof of the claim has already been described in \cite{ds}. So we omit the details.\qedhere
\end{proof}
Now, as in~\cite{ds}, let $I(u)$ be an interval of $I$. Also, let $I(v)$ be an interval contained in $I(u)$ and $l(v)$ be the smallest subject to $I(v)$ contained in $I(u)$. Then, we say $I(v)$ is the leftmost subject to $I(v)$ contained in  $I(u)$. Similarly, if $I(v)$ is contained in  $I(u)$ and such that $r(v)$ is the largest, then we say $I(v)$ is the rightmost subject to the condition that $I(v)$ is contained in $I(u)$.

\begin{claim}
\textit{No interval contains two distinct intervals.}
\end{claim}
\begin{proof}\renewcommand{\qedsymbol}{}
To prove this Claim we assume to the contrary that an interval contains
two distinct intervals. We consider the following possible cases.
\begin{cas}
Let $(x_1,x)$ and $(x_2,x)$ be bad pairs, where $x_1, x_2$ are distinct 
vertices. Again let $x_1$ be the vertex for which the interval $I(x_1)$ 
is the leftmost subject to the condition $I(x_1)\subset I(x)$ and $x_2$ 
be the vertex for the interval $I(x_2)$ is the rightmost subject to the 
condition $I(x_2)\subset I(x)$, also $x_1, x_2 \in X$. Claim~1 applied 
to the bad pairs $(x_1,x)$ and $(x_2,x)$ implies the existance of the 
vertices $y_1$ and $y_2$ of $B$ such that $l(x)<r(y_1)<l(x_1)$ and 
$r(x_2)<l(y_2)<r(x)$. Since $x_1$ and $x_2$ are not isolated vertices, 
there exists vertices $y'$ and $y''$ so that $I(y')$ intersects $I(x_1)$ 
but not $I(x_2)$ and $I(y'')$ intersects $I(x_2)$ but not $I(x_1)$, and also 
it is possible that the interval $I(y')$ intersects both $I(x_1)$ and 
$I(x_2)$. \par
Consider the first possibility. Here we assume $l(y')<l(x)$ and 
$r(y'')>r(x)$, as ${ I}$ contains minimum number of bad pairs. Since the vertices $y_1$ and $y_2$ are not copies in $B$, 
there exists a vertex $x_3$ such that $I(x_3)$ intersects $I(y_1)$. 
As before, assume $(y_1,y')$ is not a bad pair. Thus, $I(x_3)$ is such that 
$l
(y_1)<r(x_3)<l(y')$. Similarly, $(y_2,y'')$ is not a bad pair, there exists a 
vertex $x_4$ such that $r(y'')<l(x_4)<r(y_2)$.Obviously  $I(x_4)$ intersects $I(y_2)$.
Now, the vertices $x, x_1, x_2, x_3, x_4, y_1, y_2, y', y''$ induce $B_1$,
which is a contradiction.\par
Next, suppose $I(y')$ intersects $I(x_1)$ and $I(x_2)$. Here $I(y')\not\subset I(x)$ and $I(x_2)\not\subset I((y')$. So, we have 
$r(y')<r(x_2)$ and $l(y')<l(x)$. Since the vertices $x_1$ and $x_2$ are not 
copies, there exist a vertex $y''$ such that $I(y'')$ intersects 
$I(x_2)$. Also $l(x_2)<l(y'')$ and $r(y'')>r(x)$. As before, $(y_1,y')$ 
and $(y_2,y'')$ are not bad pairs. Thus, there exist vertices $x'$ and 
$x''$ such that $I(x')$ intersects $I(y_1)$ and $I(x'')$ intersects 
$I(y_2)$ satisfying the conditions $r(x')<l(y')$ and $r(y'')<l(x'')$ 
respectively. Since $l(y')< l(x)$, $l(y')$ of the interval $I(y')$ cannot be shorten 
to the  right, there exists a vertex $x_1'$ such that $l(y')<r(x_1')$ 
but $l(x')<l(x_1')<l(y_1)$. Thus no new bad pair is formed. Now, the vertices $x, x_1', x_1, x_2, y', 
y'', y_1, y_2$ induce $F_2$, which is a contradiction.
\end{cas}
\begin{cas}
Let $(y_1,y)$ and $(y_2,y)$ be bad pairs, where $y_1$ and $y_2$ are distinct 
vertices. In this case we can arrive at a contradiction similar to the previous 
case (here only $x,x_1$ and $x_2$ are interchanged respectively 
by $y,y_1$ and $y_2$). So detail is omitted.
\end{cas}
\begin{cas}
Let $(y_1,x)$ and $(y_2,x)$ be bad pairs, where $x, y_1, y_2$ are 
distinct vertices of $B$. Suppose $y_1\in Y$ be the vertex for which 
the interval $I(y_1)$ is the leftmost subject to the condition 
$I(y_1)\subset I(x)$ and $y_2$ be the vertex for which the interval 
$I(y_2)$ is the rightmost subject to the condition $I(y_2)\subset I(x)$. 
Now, by Claim~1, there exists vertices $x_1$ and $x_2$ such that 
$l(x)<r(x_1)<l(y_1)$ and $r(y_2)<l(x_2)<r(x)$.\par
Since the bigraph $B$ is connected, there must exists vertices $y_3$ and $y_4$ 
such that $I(y_3)$ intersects $I(x)$ and $I(x_1)$, and $I(y_4)$ 
intersects $I(x)$ and $I(x_2)$. Since the vertices $y_1$ and $y_2$ 
are copies, there must exists a vertex $x_3$ such that $I(x_3)$ intersects 
$I(y_1)$ but not $I(y_2)$. Again $(x_3,x)$ is not a bad pair. We assume $l(x_1)<l(y_3)<l(x_3)<l(x)$ and $r(x_1)<r(y_3)<r(x_3)$, so that no two of the intervals $I(x_1)$, $I(y_3)$ and $I(x_3)$ form a bad pair. Now $I(x_3)$ does not intersect $I(y_2)$, thus we have an interval $I(y_5)$ such that $r(x_3)<l(y_5)$ and $r(x)<r(y_5)$. Also, no two of the intervals   $I(y_5),I(y_4),I(x_2)$ forms a bad pair. But now $(y_2,y_5)$  is a bad pair. Thus by claim 1, we have an interval $I(x_5)$ such that $l(y_5)<r(x_5)<l(y_2)$. Now if $I(x_5)\cap I(y_3) =\emptyset$, then the vertices $x, x_2, x_3, x_5, y_1, y_2, y_3, y_4$ and $y_5$ induce $F_8$. Next, assume $I(x_5)\cap I(y_3)\neq \emptyset$. If $(x_2,y_4)$ is a bad pair, then by the claim 1 there exists an interval $I(y')$ such that $r(x_2)<l(y')<r(y_4)$. Then $y'$ is an isolated vertex. Which is a contradiction as the bigraph is connected. Similarly, $(x_2,y_5)$,$(y_4,x_2)$ and $(x_2,y_4)$ cannot be a bad pair. Thus assume $(y_4,y_5)$ is a bad pair (otherwise we can derive a contradiction). So,by claim 1,we have an interval $I(x')$ such that $r(y_4)<l(x')<r(y_5)$. Now the vertices $x, x', x_2, x_5, y_2, y_3, y_4$ and $y_5$ induce $F_2$. Thus we have a contradiction in both possibilities.
\par Next, we consider the case where bad pairs are formed by intervals $I(x_3),I(x_1)$ and $I(y_3)$. Assume $(x_1,x_3)$ is a bad pair but $(y_3,x_3)$ is not a bad pair if $(y_3,x_3)$ is a bad pair then as before we have an isolated vertex in $B$ i.e. $r(x_3)<r(y_3)<r(y_1)$ and $l(x_3)<l(y_3)$. Now for the bad pair $(x_1,x_3)$, by claim 1, we have an interval $I(y_1')$ such that $l(x_3)<r(y_1')<l(x_1)$. Also assume $(x_1,y_3)$ is a bad pair, otherwise we can have an isolated vertex in $B$, by claim 1 we have an interval $I(y')$ such that $l(y_3)<r(y')<l(x_1)$. Also $r(y_1')<r(y')$ and $l(y_1')<l(y')$ as $(y_1',y')$ is not a bad pair. Now the vertices $y'$ and $y_1'$ are copies, but the bigraph is copy-free. So we have a vertex $x_1'$ such that $I(x_1')$ intersects $I(y_1')$ only. Again as $(y',y_1')$ is not a bad pair, we have a vertex $x'$ such that $r(y_1')<l(x')$ and if $I(x')$ does not intersect $I(y_3)$ then the vertices $x, x_1, x_2, x_3, x', x_1', y_1, y_2, y_3, y_4, y'$ and $y_1'$ induce the graph $K$. In the other case if $I(x')$ intersects $I(y_3)$ then the vertices $x, x_1, x_3, x', y_1, y_3, y'$ and $y_1'$ induce $F_2$. Which is also a contradiction.\par 
\end{cas}
\begin{cas}
Let $(x_1,y)$ and $(x_2,y)$ be bad pairs where $x_1$ and $x_2$ are distinct 
vertices. This case is similar to the previous one so omitted.
\end{cas}
\begin{cas}
Let $(x_1',x)$ and $(y_1',x)$ be bad pairs, where $x, x_1'$ and $y_1'$ 
are distinct vertices.\par
Let $x_1$ be the vertex of $B$ such that $I(x_1)$ is the leftmost subject 
to the condition $I(x_1)\subset I(x)$, where $x_1\in X$ and $y_1$ be the 
vertex such that $I(y_1)$ is the rightmost subject to the condition 
$I(y_1)\subset I(x)$, where $y_1\in Y$. First, we assume that 
$I(x_1)\cap I(y_1)=\emptyset$. Now by Claim~1, there exists a vertex $y'$ 
such that $l(x)<r(y')<l(x_1)$, also there exists another vertex $x'$ 
such that $r(y_1)<l(x')<r(x)$. Since $x'$ is not an isolated vertex and 
the bigraph is connected, there exists $y_2$, such that $I(y_2)$ 
intersects $I(x)$ and $I(x')$. Again, since $x_1$ is not an isolated vertex, 
there exists a vertex $y''$ such that $I(y'')$ intersects $I(x_1)$ but not contain it. Now we must have $I(y'')\not\subset I(x)$, i.e. $l(y'')<l(x)$ and $r(y'')<r(x_1)$. Since $r(y'')$ cannot be extended to the right, there exist a vertex $x''$ such that $r(y'')<l(x'')$. As $x''$ is not an isolated vertex, $I(x'')\cap I(y_1)\neq\emptyset$. Also $(y_1,x'')$ is not a bad pair, thus $r(x'')<r(y_1)$. Again $(y',y'')$ is not a bad pair, there exist a vertex $x_2$ such that $r(x_2)<l(y'')$ and $I(x_2)$ intersects $I(y')$ with $l(x_2)<l(y')$. Now, 
the vertices $x, x_1, x', x_2, x'', y_1, y_2, y'$ and $y''$ induce $B_1$, 
which is a contradiction.\par
Next, assume $I(x_1)\cap I(y_1)\neq \emptyset$ and $I(x_1)\neq I(y_1)$. 
Since $I(y_1)$ is the rightmost and $I(x_1)\neq I(y_1)$, there exists a 
vertex $y_2$ such that $l(y_2)>r(x_1)$ and $r(y_2)<r(y_1)$. By a similar 
argument there exists $x_2$ such that $r(x_2)<l(y_1)$ and 
$l(x_2)>l(x_1)$. Now, by Claim~1, there exists a vertex $y'$ such that 
$l(x)<r(y')<l(x_1)$, and also there exist a vertex $x'$ such that 
$r(y_1)<l(x')<r(x)$. Again, Claim~1 applied to the bad pair $(x_2,x_1)$ 
we have a vertex $y_1'$ such that $l(x_1)<r(y_1')<l(x_2)$. Also, suppose 
$l(y_1')<l(x)$ as $(y_1',x)$ is not a bad pair. Now, since $x_2$ is not an isolated vertex, there exist a 
vertex $y_1''$ such that $I(y_1'')\cap I(x_2)\neq \emptyset$ and 
$l(y_1'')<l(x)$ and $r(y_1'')<r(x_2)$ as $(y_1'',x)$ and $(x_2,y_1'')$ are not a bad pair. Also, we may assume that no two of the intervals $I(y'), I(y_1')$ and $I(y_1'')$ forms a bad pair. Since $x'$ is not an isolated vertex, there exists 
$y_3$ such that $I(y_3)$ intersects $I(x)$ and $I(x')$. Again, the 
vertices $y'$ and $y_2$ are copies but since $B$ is copy-free, there exists 
a vertex $x''$ such that $I(x'')$ intersects $I(y')$ only. \\
Now for the bad pair $(y_2,y_1)$, by claim 1, we have a vertex $x_1'$ such that $r(y_2)<l(x_1')<r(y_1)$. Also $r(x_1')>r(x)$ as no new bad pair is created. Then $I(x_1')$ must intersects $I(y_3)$. Again $(y_1',y_1'')$ is not a bad pair, thus $r(x'')<l(y_1'')$ and $I(y_1')$ intersects $I(x'')$. Then the vertices $x,x_1,x_1',x'',y_1,y_2,y_3,y_1'$ and $y'$ induce $F_8$. Next, assume $I(y_1')$ does not intersect $I(x'')$. As $(y_1',y_1'')$ is not a bad pair, there exists an interval $I(x_3)$ such that $r(x_3)<l(y_1'')$. Then $I(x_3)$ intersects $I(y')$ and $I(y_1')$, and the vertices $x, x_1, x_2, x_3, x_1', y_1, y', y_1'$ and $y_1''$ induce $F_5$. Thus we have a contradiction in any possibility. \par
Finally, we assume $I(x_1)=I(y_1)$. Now, the existence of $x_1'$ and 
$y_1'$ and the choice of $x_1$ and $y_1$ imply that $(x_1',x_1)$ and 
$(y_1',y_1)$ are bad pairs. So, by Claim~1, there exists two vertices 
$y_2$ and $y_3$ such that $l(x_1)<r(x_2)<l(x_1')$ and 
$r(x_1')<l(y_3)<r(x_1)$. For the similar reason, there exists $x_2$ and 
$x_3$ such that $l(y_1)<r(x_2)<l(y_1')$ and $r(y_1')<l(x_3)<r(y_1)$. Now, 
the vertices $x_1, x_2, x_3, x_1', y_1, y_2, y_3$ and $y_1'$ induce 
$F_2$, which is a contradiction.
\end{cas}
\begin{cas}
Let $(x_1,y)$ and $(y_1,y)$ be bad pairs where $x_1,y_1,y$ are distinct vertices. This case is similar to the previous one so omitted. Hence proof of 
the Claim~2 is complete.
\end{cas}
\end{proof}\vspace{-1cm}
\begin{claim}
\textit{No interval is contained in two intervals.}
\end{claim}
\begin{proof}\renewcommand{\qedsymbol}{}
To prove the Claim we assume to the contrary that an interval is contained in two 
distinct intervals. Now we consider the following cases.\par
\begin{cas}
Suppose $(x,x_1)$ and $(x,x_2)$ are bad pairs, where $x_1$ and $x_2$ 
are distinct vertices. Since neither $(x_1,x_2)$ nor $(x_2,x_1)$ is a bad 
pair we may assume that $l(x_1)<l(x_2)<l(x)<r(x)<r(x_1)<r(x_2)$. By 
Claim~1, there exists vertices $y_1$ and $y_2$ such that 
$l(x_2)<r(y_1)<l(x)$ and $r(x)<l(y_2)<r(x_1)$. Also, by Claim~2, 
$l(y_1)<l(x_1)$ and $r(y_2)>r(x_2)$. Since the bigraph $B$ is connected, 
$x$ cannot be an isolated vertex. Thus, there exists a vertex $y$ such 
that $I(y)$ satisfies the conditions: $l(x_1)<l(y)<l(x_2)$ and 
$r(x_1)<r(y)<r(x_2)$. Now, since $x_1$ and $x_2$ are copies but $B$ is 
copy-free, there exists a vertex $y_3$ such that $I(y_3)$ intersects 
$I(x_1)$ only. Again, $y_1$ and $y_2$ are copies, there exists 
a vertex $x_3$ such that $I(x_3)$ intersects $I(y_2)$ only. Now, since 
$r(y)$ cannot be shorten to the left, there exists a vertex $x_4$ 
such that $l(x_4)<r(y)$. Also, by Claim~2, $r(x_4)>r(x_2)$. Then, obviously
$I(x_4)$ intersects $I(y_2)$. Now, the vertices $x, x_1, x_2, x_3, x_4, 
y, y_1, y_2$ and $y_3$ induce $F_4$, which is a contradiction.\par
Next, we consider another situation. Since $l(y)$ cannot be shorten to the 
right, there exists a vertex $x_5$ such that $l(y)<r(x_5)$. Also, $l(y_1)$
cannot be shorten to the right there exists a vertex $x_6$ such that 
$r(x_6)>l(y_1)$. Now if, $I(x_5)\cap I(y_3)=\emptyset$ and 
$I(x_6)\cap I(y_3)=\emptyset$, then the vertices $x_1, x_2, x_3, x_4, 
x_5, x_6, y, y_1, y_2, y_3$ induce $F_9$, which is a contradiction. 
Again if, $I(x_5)\cap I(y_3)\neq \emptyset$ and $I(x_6)\cap I(y_3)
\neq \emptyset$, then the vertices $x, x_1, x_2, x_3, x_4, x_5, x_6, y, 
y_1, y_2, y_3$ induce $F_{12}$. In the other case if, $I(x_5)\cap I(y_3)
=\emptyset$, and $I(x_6)\cap I(y_3)\neq \emptyset$, then the vertices 
$x, x_1, x_4, x_5, y, y_1, y_2, y_3$ induce $F_2$. Lastly, suppose 
$I(x_5)\cap I(y_3)\neq \emptyset$, and $I(x_6)\cap I(y_3)=\emptyset$. 
Now, as $l(x_2)$ cannot be prolonged to the left, there exists a vertex 
$y_4$ such that $r(y_4)<l(x_2)$. Let $I(y_4)\cap I(x_5)=\emptyset$, 
then the vertices $x, x_1, x_4, 
x_5, y, y_1, y_2, y_4$ induce $F_2$. In the other cases we can similarly 
derive $F_2$ as an induce subgraph. Thus, we have a contradiction in any 
case. \par
Next, we consider another possibility, where $l(x)<r(y)<r(x)$ and $l(y)<
l(x_1)$. Since $x_1, x_2$ are copies, there exist $y_3$ such 
that $I(y_3)\cap I(x_1)\neq \emptyset$; also since $y_1$ and $y_2$ are 
copies, there exists a vertex $x_3$ such that $I(x_3)$ intersects 
$I(y_2)$ only. Next, let $l(y_1)$ and $l(y_3)$ are less than $l(y)$ as $I$ contains minimum number of bad pairs. 
Since $l(y_3)$ and $l(y_1)$ cannot be shorten to the left, there exist 
a vertex $x_4$ such that $l(y_3)<l(y_1)<r(x_4)$ and $r(x_4)<l(y)$. Now, the vertices $x, x_1, x_2, x_3, x_4, y, y_1, y_2, y_3$ 
induce $F_5$, which is a contradiction.\par
Finally, we consider the following possibility. As before, $I(y)$ is 
such that $l(x_1)<l(y)<l(x_2)$ and $r(x_1)<r(y)<r(x_2)$. Since $r(y)$ 
cannot be shorten to the left, we have a vertex $x_3$ such that 
$l(x_3)<r(y)$ and $r(x_3)>r(x_2)$. Again, by Claim~2, $l(x_3)$ must not 
be prolonged to the left, we have an interval $I(y_4)$ such that 
$r(y_4)<l(x_3)$ and $l(y_4)<l(x_1)$ also $ r(y_4)<r(x_1)$. Since $y_1$ and $y_4$ are not copies, 
$I(y_4)$ intersects $I(x)$. Now, $x_1$ and $x_2$ are copies, we have 
a vertex $y_3$ such that $I(y_3)$ intersects $I(x_1)$ only. Then, the 
vertices $x, x_1, x_2, x_3, y, y_1, y_2, y_3$ and $y_4$ induce $F_{11}$, 
which is a contradiction.
\end{cas}
\begin{cas}
Suppose $(y,y_1)$ and $(y,y_2)$ are bad pairs, where $y_1$ and $y_2$ are distinct vertices. \par
This case is similar to the previous one (here $x,x_1$ and $x_2$
are respectively replaced by $y,y_1$ and $y_2$) so omitted.
\end{cas}
\begin{cas}
Suppose $(y,x_1)$ and $(y,x_2)$ are bad pairs, where $x_1$ and $x_2$ are 
distinct vertices.\par
Since neither $(x_1,x_2)$ nor $(x_2,x_1)$ is a bad pair we may assume that 
$l(x_1)<l(x_2)<l(y)<r(y)<r(x_1)<r(x_2)$. By Claim~1, there exist vertices 
$x_3$ and $x_4$ such that $l(x_2)<r(x_3)<l(y)$ and $r(y)<l(x_4)<r(x_1)$. 
Again, by Claim~2, no interval contains two distinct intervals, we have 
$l(x_3)<l(x_1)$ and $r(x_4)>r(x_2)$. Since $x_4$ is not an isolated vertex, 
we have a vertex $y_2$ such that $I(y_2)$ intersects $I(x_4)$. Again, 
$r(x_1)$ cannot be shorten to the left and the bigraph is connected, we have $l(y_2)<r(x_1)$. 
Therefore $I(y_2)$ intersects $I(x_1)$, $I(x_2)$ and 
$I(x_4)$. By a symmetric argument, we have an interval $I(y_1)$ such that 
$l(x_2)<r(y_1)$ and $I(y_1)$ intersects $I(x_1)$, $I(x_2)$ and $I(x_3)$. 
Also, by Claim~2, $l(y_1)<l(x_1)$ and $r(x_2)<r(y_2)$. Now, the vertices 
$x_1$ and $x_2$ are copies but the bigraph is copy-free, we have a vertex 
$y_3$ such that $I(y_3)$ intersects $I(x_2)$ and consequently it  
intersects $I(x_4)$ and $r(y_2)<r(x_4)<r(y_3)$ and $r(x_1)<l(y_3)<r(x_2)$. Since $l(x_1)$ should not be 
shrinked towards right, there exists a vertex $y_4$ such that 
$l(x_1)<r(y_4)$. Then, $I(y_4)$ intersects $I(x_1)$ and $I(x_3)$. Again, 
$(y_4,y_1)$ is not a bad pair, as $l(y_4)<l(y_1)$. So there exists vertex $x_1'$ such that $r(x_1')<l(y_1)$. Again $(y,y_2)$ is not a bad pair, so there exists a vertex $x$ such that $l(y)<r(x)<r(y)$ and $l(x)<l(x_1)$. Next $(y_1,x)$ is not a bad pair, there exist a vertex $y_1'$ such that $l(x_3)<r(y_1')<l(x)$. Since $l(x_1)$ cannot be extended to the left, there exists a vertex $y_5$ such that $r(y_5)<l(x_1)$. Again no two of the intervals $I(y_1)$,$I(y_4)$,$I(y_5)$ and $I(x_3)$ forms a bad pair, so assume $l(y_5)<l(y_4)<l(x_3)<l(y_1)$. Now $l(y_1)$ cannot be prolonged to the left, so there exists a vertex $x_1'$ such that $r(x_1')<l(y_1)$. Similarly there exists a vertex $x_1''$ such that $r(x_1'')<l(y_4)$. Again $l(x_1')<l(y_5)$ and $l(x_1')$ cannot be shorten to the left, so there exists a vertex $y_1''$ such that $l(x_1')<r(y_1'')$. If $I(y_1')$ does not intersects $I(x_1')$ then the vertices $x$, $x_1$, $x_3$, $x_1'$, $y$, $y_4$, $y_5$, $y_1'$, $y_1''$ induce $F_5$, which is a contradiction. 
\par In the other case if $I(y_1')$ intersects $I(x_1')$ then ($x_1',y_1')$ is not a bad pair. Thus we have an interval $I(x_0)$ such that $l(y_4)<r(x_0)<l(y_1')$. Now the vertices $x,x_1,x_3,x_0,y,y_4,y_5,y_1',y_1''$ induce $F_5$.
 
\end{cas}

\begin{cas}
Suppose $(x,y_1)$ and $(x,y_2)$ are bad pairs, where $y_1$ and $y_2$ are distinct vertices. This case is similar to the previous one so details are omitted.
\end{cas}

\begin{cas}
Suppose $(y,x_1)$ and $(y,y_1)$ are bad pairs, where $x_1$ and $y_1$ are 
distinct vertices. Since $(x_1,y_1)$ or $(y_1,x_1)$ is not a bad pair we 
may assume that $l(x_1)<l(y_1)<l(y)<r(y)<r(x_1)<r(y_1)$. Since $(y,y_1)$ is 
a bad pair, by Claim~1, there exists a vertex $x_2$ such that 
$l(y_1)<r(x_2)<l(y)$. Also, by Claim~2, $l(x_2)<l(x_1)$. By a similar 
reason, there exists a vertex $x_3$ such that $r(y)<l(x_3)<r(x_1)$. Also, by 
Claim~2, $r(x_3)>r(y_1)$. As $r(y_1)$ cannot be shrinked towards left, there exists a vertex $x_4$ such that $l(x_4)<r(y_1)$ and $r(x_3)<r(x_4)$. Again, $(x_3,x_4)$ is not a bad pair, there exists a vertex $y'$ such that $r(x_1)<r(y')<l(x_4)$ and also by claim 2, $l(x_1)<l(y')<l(y_1)$. Then $I(y')$ intersects both $I(x_2)$ and $I(x_3)$. Next, $r(x_1)$ cannot be extended to the right, thus we have a vertex $y_3$ such that $r(x_1)<l(y_3)<r(y')$. Now, $I(y_3)$ intersects $I(x_3)$ and $I(x_4)$. Again $l(x_1)$ cannot be extended to the left, thus we have a vertex $y_2$ such that $r(y_2)<l(x_1)$ and $I(y_2)$ intersects $I(x_2)$. Now the vertices $x_1, x_2, x_3, x_4, y, y', y_1, y_2$ and $y_3$ induce $F_5$. Which is a contradiction. 

\end{cas}
\begin{cas}
Suppose $(x,x_1)$ and $(x,y_1)$ are bad pairs, where $x_1$ and $y_1$ are 
distinct vertices. This case is similar to the previous one so detail proof 
is omitted. Thus proof of Claim~3 is complete.
\end{cas}
\end{proof}\vspace{-.5cm}
As in \cite{ds} a bad pair $(u,v)$ is \emph{clean}, if there exist vertices $z_1$ and $z_2$ such that $r(z_1)=l(v)$ and $r(v)= l(z_2)$, where $z_1$, $z_2$ and $u$ are vertices of different partite sets of $V(B)$. \par
Now, we shall study the structure of $I$ and the corresponding properties 
of the bigraph $B$ forced by a bad pair $(u',u)$, where $u$ and $u'$ are 
vertices of the same partite set of $V(B)$. A vertex $u$ is said to be left of a 
vertex $v$ (in representation $I$) if $r(u)<l(v)$. Similarly, $v$ is said 
to be right of $u$, where $u, v\in X\cup Y$. Next, $I(u)$ and $I(v)$ are 
two intersecting intervals such that $l(u)<l(v)$ and $r(u)<r(v)$, then 
$v$ is \textit{distinguishable} from $u$ to the left if, there is a vertex $u_1$ 
such that $r(u_1)<l(v)$ and $I(u)$ intersects $I(u_1)$, here $u_1$ and 
$v$ are vertices of different partite sets. Similarly, $u$ is 
\textit{distinguishable} from $v$ to the right if, there is a vertex $v_1$ such 
that $I(v_1)$ intersects $I(v)$ and $r(u)<l(v_1)$, where $u$ and $v_1$ 
are vertices of different partite sets. In short, we say that the 
vertices $u_1$ and $v_1$ distinguishe $u$ and $v$ from the left and 
right respectively. Let $(u',u)$ be a bad pair, where $u$ and $u'$ 
belong to the same partite set. Then Fig.~\ref{bad} shows that 
there is a structure in $I$ corresponding to this bad pair. Let us 
introduce two positive integers $k_r(u',u)$ and $k_l(u',u)$, which 
indicates how large this structure in the right and left respectively.\par
\begin{figure}[H]
\centering
\begin{tikzpicture}
\pgfmathsetmacro{\d}{0.15}
\draw \foreach \p/\q/\r/\s in {0/0/2/Z'_2,.2/.5/2/,1.7/1/2/Z'_1,2.4/1.5/2/,3.5/2/4/,4.7/1.5/1.5/u',6.5/1.5/2/,7.2/1/2/Z_1,8.7/.5/2/,8.9/0/2/Z_2}
{
(\p,\q)--(\p+\r,\q)
(\p+\d,\q-\d)--(\p,\q-\d)--(\p,\q+\d)--(\p+\d,\q+\d)
(\p+\r-\d,\q-\d)--(\p+\r,\q-\d)--(\p+\r,\q+\d)--(\p+\r-\d,\q+\d)
(\p+0.5*\r,\q-0.4) node{$\s$}
}
(5.5,2.3)node{$u$};

{\tikzstyle{every node}=[circle, draw, fill=black,
                        inner sep=0pt, minimum width=2pt]
\draw (11.5,0)node{} (12,0)node{} (12.5,0)node{} (-.5,0)node{} (-1,0)node{} (-1.5,0)node{};}
\end{tikzpicture}
\caption{The structure of the representation $I$ of $B$ forced by the bad pair $(u',u)$.}
\label{bad}
\end{figure}
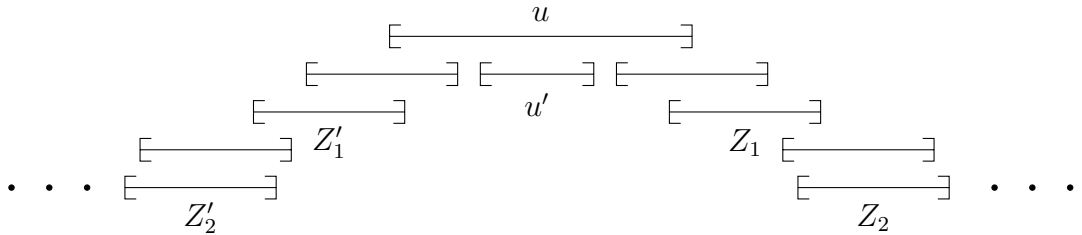

For a bad pair $(u',u)$, let $Z_1$ be the set of vertices to the right 
of $r(u')$ and $Z_1'$ be the set of vertices to the left of $l(u')$. By Claim~1, we conclude that 
$|Z_1|\geq 1$ and $|Z_1'|\geq 1$. The vertices of $Z_1$ (or $Z_1'$) 
and $u'$  belong to opposite partite sets. Without loss of generality, 
we consider the bad pair $(x',x)$ (since for the bad pair $(y,x)$, it 
can be shown that if $|Z_1|> 1$ and/or $|Z_1'|> 1$ then the corresponding bigraph $B$ has a $\mathcal{U}$-representation). Then the vertices of $Z_1$ (or $Z_1'$) 
belong to $Y$. If, $|Z_1|=1$, then $k_r(u',u)=1$ and we stop here (similar 
situation arises if, $|Z_1'|=1)$. So, we assume $|Z_1|\geq 2$ and 
$|Z_1'|\geq 2$.\par
Next, let $|Z_1|=2$ and $Z_1=\{y_1, y_2\}$ such that $l(y_1)<l(y_2)$. Suppose, 
there exists a vertex $x''$ which distinguishes $y_1$ and $y_2$ from the 
left i.e. $l(y_1)<r(x'')<l(y_2)$. Now, $l(x'')>l(x)$ is not possible by 
Claim~2. Again, $l(x'')\leq l(x)$ is not possible by Claim~3. Thus existence 
of the vertex $x''$ is not possible. Again, $(y_2, y_1)$ is not a bad 
pair, otherwise we have a vertex $x_0$ by Claim~1, such that $l(y_1)<
r(x_0)<l(y_2)$, which is a contradiction. Thus, we have $l(y_1)<l(y_2)
<r(y_1)<r(y_2)$. If $Z_1'=\{y_1', y_2'\}$ then similarly we have 
$l(y_2')<l(y_1')<r(y_2')<r(y_1')$.\par
Now, we shall study different possibilities forced by the bad pair 
$(x', x)$. First, we shall show that $|Z_1|\leq 2$ or $|Z_1'|\leq 2$. By 
contradiction, we assume that $Z_1=\{y_1, y_2, y_3\}$ and 
$Z_1'=\{y_1', y_2', y_3'\}$. Then, as in the case of $|Z_1|=2$ and 
$|Z_1'|=2$, we have $l(y_1)<l(y_2)<l(y_3)<r(y_1)
<r(y_2)<r(y_3)$ and $l(y_3')<l(y_2')<l(y_1')<r(y_3')<r(y_2')<r(y_1')$. 
Now, no two vertices  of $Z_1$ (or $Z_1'$) form a bad pair, there exist 
vertices $x_1$ and $x_2$ such that $r(y_1)<l(x_1)<r(y_2)$ and 
$r(y_2)<l(x_2)<r(y_3)$. Also, $r(x_1)<r(x_2)$ since $(x_2, x_1)$ is not a 
bad pair. Again as before,  $(x_1,y_3)$ and $(x_2,y_3)$ are not bad pairs. So, we 
have $r(y_3)<r(x_1)<r(x_2)$.\par
By a similar argument, there exist $x_1'$ and $x_2'$ such that $l(y_2')<
r(x_1')<l(y_1')$ and $l(y_3')<r(x_2')<l(y_2')$ also $l(x_2')<l(x_1')<l(y_3')$. Again, $x'$ is not an 
isolated vertex, there exist a vertex $y'$ such that $I(y')$ intersects 
$I(x')$ and it does not form any bad pair with any intervals. Now, the vertices $x, x', x_1, 
x_2, x_1', x_2', y', y_1, y_2, y_3, y_2', y_3'$ induce $B_2$, which is 
a contradiction.\par
Next, assume $Z_1=\{y_1, y_2, y_3\}$ and $Z_1'=\{y_1', y_2'\}$. Then, as 
before there exists a vertex $x_1$ distinguishes $y_1$ and $y_2$ from 
the right and $x_2$ distinguishes $y_2$ and $y_3$ from the right. Also, 
there exists a vertex $x_1'$ distinguishes $y_1'$ and $y_2'$ from the 
left. Again, $y_1$ and $y_1'$ are copies, there exist a vertex $x''$ such 
that $l(y_1')<l(x'')<l(x)$ and $r(y_1')<r(x'')<l(x')$. Now, the vertices 
$x, x', x'', x_1, x_2, x_1', y', y_1, y_2, y_3, y_1', y_2'$ induce 
$B_2$, which is a contradiction.\par
Next, assume $Z_1=\{y_1, y_2, y_3\}$ and $Z_1'=\{y_1'\}$. Also, as before 
let the vertex $x_1$ distinguishes $y_1$ and $y_2$ from the right and 
$x_2$ distinguishes $y_2$ and $y_3$ from the right. Since $x'$ is not an 
isolated vertex, there exists a vertex $y$ such that $I(y)$ intersects 
$I(x')$ and $(x', y)$ is not a bad pair. By Claim~2, suppose $l(y)<l(x)$. 
Now, $y_1$ and $y_1'$ are copies, there exists a vertex $x_2'$ such that 
$I(x_2')$ intersects $I(y_1')$. Then, we can assume that $I(y)$  intersects 
$I(x_2')$. Again if, $(x_2', y)$ is a bad pair then by Claim~1, there exist a 
vertex $y'$ such that $r(x_2')<l(y')<r(y)$. Then, $I(y')$ intersects 
$I(x')$ and $I(x)$ and assume $r(x)<r(y')<r(y_1)$. Now, the vertices 
$x, x', x_1, x_2, x_2', y, y', y_1', y_1, y_2$ and $y_3$ induce $H_0$, 
which is a contradiction. In the other case if, $l(x_2')<l(y)$ such that 
$(x_2', y)$ is not a bad pair. Then, there exists a vertex $x''$, such 
that $r(x'')<l(y)$ and $I(x'')$ intersects $I(y_1')$. Now, the vertices 
$x, x', x_1, x_2, x'', x_2', y, y_1', y_1, y_2, y_3$ induce the bigraph 
$M$, which is again a contradiction.\\
Now, we consider the case where $Z_1=\{y_1, y_2\}$, $Z_1'=\{y_1'\}$ 
and $Z_2=\{x_1, x_2\}$. Then, we have $l(y_1)<l(y_2)<r(y_1)<r(y_2)$. The 
intervals $I(x_1)$ and $I(x_2)$ intersect $I(y_2)$ only since $y_1$ and 
$y_2$ are not copies. Assume $l(x_1)<l(x_2)$ and $r(x_1)<r(x_2)$ as 
$(x_2, x_1)$ is not a bad pair. Again, $x'$ is not an isolated vertex, 
there exist a vertex $y$ such that $I(y)$ intersects $I(x')$ and 
$l(y_1')<l(y)<l(x)$, also $(x',y)$ is not a bad pair by claim 3. Again, $y_1$ and $y_1'$ are 
 copies but the bigraph is copy-free, there exist a vertex $x_1'$ such that $l(x_1')<l(y_1)$ and 
$r(x)<r(x_1')<r(y_1)$. So $I(x_1')$ intersects $I(y_1)$ and $I(y_2)$. 
Again, $(y_1',y)$ is not a bad pair, there exists a vertex $x_2'$ such 
that $I(x_2')$ intersects $I(y_1')$ and $r(x_2')<l(y)$. Next, let 
$Z_3=\{y_3\}$ where $r(x_1)<l(y_3)<r(x_2)$ and $r(y_3)>r(x_2)$, i.e. 
$k_r(x',x)=3$ and we stop here. Now, if there exists a vertex $y'$ 
distinguishes $x_1$ and $x_2$ from the left then $l(x_1)<r(y')<l(x_2)$ 
and $r(x_1')<l(y')$. Now, the vertices $x, x', x_1', x_2', x_1, x_2, y, 
y_1, y_2, y_1', y_3$ and $y'$ induce $S_1$, which is a contradiction. 
Next, let $Z_2=\{x_1, x_2, x_3\}$ i.e. $|Z_2|=3$. Then the intervals 
$I(x_1)$, $I(x_2)$ and $I(x_3)$ intersects $I(y_2)$ only since the bigraph 
is copy-free. Again, no two intervals of $Z_2$ forms a bad pair, we have 
$l(x_1)<l(x_2)<l(x_3)$ and $r(x_1)<r(x_2)<r(x_3)$. Consider a vertex 
$y_3$ distinguishes $x_1$ and $x_2$ from the right and a vertex $y_4$ 
distinguishes $x_2$ and $x_3$ from the right. Then, the vertices $x, x', 
x_1', x_2', x_1, x_2, x_3, y, y_1', y_1, y_2, y_3$ and $y_4$ induce 
$\widetilde{\mathcal{S}}_1$, which is a contradiction.\par
Next, we consider the case, where $Z_1=\{y_1, y_2\}$ and 
$Z_1'=\{y_1', y_2'\}$. Also, let $Z_2'=\{x_1'\}$ and $Z_2=\{x_1, x_2\}$. 
Then $|Z_2'|=1$ i.e. $k_l(x',x)=2$ and we stop here. Now, assume 
$l(x_1)<l(x_2)$ and $r(x_1)<r(x_2)$ since $(x_2,x_1)$ is not a bad 
pair. Again, $x'$ is not an isolated vertex, there exists a 
vertex $y'$ such that $I(y')$ intersects $I(x')$ and by claim 2 and claim 3, $l(x')<l(y')$ 
and $r(x)<r(y')<r(y_1)$. Now, $y_1$ and $y_1'$ are not copies, there 
exists a vertex $x''$ such that $I(x'')$ intersects $I(y_1')$, 
consequently $I(y_2')$ and $l(y_1')<l(x'')<l(x)$ and $r(y_1')<r(x'')<l(y')$. Let $y$ be a vertex 
distinguishes $x_1$ and $x_2$ from the left. If $l(y)<r(x)$, then 
$\{y, y_1, y_2\}$ forms $Z_1$, i.e. $|Z_1|=3$, which a contradiction. 
So, let $l(y)>r(x)$. Again, let the vertex $y_3$ distinguishes $x_1$ 
and $x_2$ from the right. Then the vertices $x, x', x'', x_1', x_1, x_2, 
y, y', y_1', y_2', y_1, y_2$ and $y_3$ induce $R_1$, which is a 
contradiction. Finally, if 
$Z_2=\{x_1, x_2, x_3\}$ then assume $l(x_1)<l(x_2)<l(x_3)$ and $I(x_1)$, 
$I(x_2)$, $I(x_3)$ intersect $I(y_2)$. Also, since no two vertices of $Z_2$ form 
a bad pair, we assume $r(x_1)<r(x_2)<r(x_3)$. Suppose $y_3$ distinguishes 
$x_1$ and $x_2$ from the right and $y_4$ distinguishes $x_2$ and $x_3$ 
from the right. Now, the vertices $x, x', x'', x_1', x_1, x_2, x_3, y', 
y_1', y_2', y_1, y_2, y_3$ and $y_4$ induce $\widetilde{R}_1$, which is 
also a contradiction.\par
Next, consider the case, where $Z_1=\{y_1, y_2\}$, $Z_1'=\{y_1'\}$, 
$Z_2=\{x_1, x_2\}$ and $Z_3=\{y_4\}$. Thus, $k_l(x',x)=1$ and we stop 
here. Also, $|Z_3|=1$ and $k_r(x',x)=3$ and we stop here. As before we 
conclude $I(x_1)$ and $I(x_2)$ intersect $I(y_2)$ and $l(x_1)<l(x_2)<
r(x_1)<r(x_2)$. Since $y_1$ and $y_1'$ are copies, there exists a vertex 
$x''$ such that $I(x'')\cap I(y_1')\neq \emptyset$ and $r(y_1')<r(x'')$ 
and $l(y_1')<l(x'')$. Again, $x'$ is not an isolated vertex, there a 
vertex $y$ such that $I(y)\cap I(x')\neq \emptyset$. Assume $r(y)<r(x')$ 
and $l(y_1')<l(y)$. Now, if $(x'',y)$ is a bad pair, by Claim~1, we have 
vertex $y'$ such that $r(x'')<l(y')<r(y)$ and $r(y')>r(x)$. Next, if 
there exists a vertex $y_3$ distinguishes $x_1$ and $x_2$ from the 
left and $I(y_3)\cap I(x)=\emptyset$, otherwise $|Z_1|=3$ and as before we can produce a 
contradiction. Now, the vertices $x, x', x'', x_1, x_2, y, y', y_1', 
y_1, y_2, y_3, y_4$ induce $H_1'$, which is a contradiction. Again, let 
$(x'', y)$ is not a bad pair. Then $l(x'')<l(y)$. Also, $(y_1', x'')$ is 
not a bad pair, then $l(y_1')<l(x'')$. Since $l(y)$ cannot be prolonged to 
the left, there exist a vertex $x_0$ such that $r(x_0)<l(y)$ and 
$I(x_0)\cap I(y_1')\neq \emptyset$. Now, the vertices $x, x', x'', x_0, 
x_1, x_2, y, y_1', y_1, y_2, y_3, y_4$ induce $M_1$, which is also a 
contradiction.\par 
Again, if $Z_2=\{x_1, x_2, x_3\}$ and no two vertices of $Z_2$ form a 
bad pair, we assume $l(x_1)<l(x_2)<l(x_3)$ and $r(x_1)<r(x_2)<r(x_3)$. 
Let, $y_3$ distinguishes $x_1$ and $x_2$ from the right and $y_4$ 
distinguishes $x_2$ and $x_3$ from the right. Then, the vertices $x, x', 
x'', x_1, x_2, x_3, y, y', y_1', y_1, y_2, y_3, y_4$ induce 
$\widetilde{H}_1'$ in the first case and the vertices $x, x', x'', x_0, 
x_1, x_2, x_3, y, y_1', y_1, y_2, y_3$ and $y_4$ induce 
$\widetilde{M}_1$ in the second case. Thus we have produced a contradiction in both the cases.\par
Next, suppose for the bad pair $(x',x)$, $Z_1=\{y_1, y_2\}$, $Z_2=\{x_1, x_2\}$, $Z_3=\{y_4\}$ and $Z_1'=\{y_1'\}$. Thus $k_r(x',x)=3$ and 
$k_l(x',x)=1$ and we stop here. Since $x'$ is not an isolated vertex, there 
exist a vertex $y$ such that $I(y)$ intersects $I(x')$, also $r(y)<r(x')$ 
and $l(y)<l(x)$. As $l(x)$ cannot be prolonged to the left, there 
exists a vertex $y'$ such that $r(y')<l(x)$ and $l(y')<l(y_1')<l(y)$. 
Also $l(y)$ (and $r(y)$) should not be shorten to the right (and left), 
there exists a vertex $x''$ (and $x_0$) such that $l(y)<r(x'')$ and 
$l(x_0)<r(y)$. Also, $l(x'')<l(y')$ and $r(x_0)>r(x)$. If, there exists a 
vertex $y_3$ distinguishes $x_1$ and $x_2$ from the left, then as before 
$l(y_3)>r(x)$. Now, the vertices $x, x', x'', x_0, x_1, x_2, y, y', y_1',  y_1, 
y_2, y_3$ and $y_4$ induce $N_1$. Again, if $Z_2=\{x_1, x_2, x_3\}$, 
as before we assume $l(x_1)<l(x_2)<l(x_3)$ and $r(x_1)<r(x_2)<r(x_3)$ and 
each vertex of $Z_2$ intersects $y_2$ only. Now, suppose $y_3$ distinguishes $x_1$ and $x_2$ from the right and $y_4$ distinguishes 
$x_2$ and $x_3$ from the right then the vertices $x, x', x'', x_0, x_1, 
x_2, x_3, y, y', y_1', y_1, y_2, y_3$ and $y_4$ induce $\widetilde{N}_1$. 
Which is also a contradiction.\par
Next, we consider the another possibility. As before, we assume $Z_1=\{y_1, y_2\}$, $Z_2=\{x_1, x_2\}$, $Z_3=\{y_4\}$ and $Z_1'=\{y_1'\}$. Now, 
since $x'$ is not an isolated vertex, we have a vertex $y$ where $I(y)\cap 
I(x')\neq \emptyset$ 
such that $l(x')<l(y)$ and $r(x)<r(y)<r(y_1)$. Since $y_1$ and $y_1'$ are 
copies, there exists a vertex $x''$ such that $l(y_1)<l(x'')<l(y_2)$ and 
$r(y_1)<r(x'')<r(y_2)$. Then $I(x'')$ intersects $I(y)$, $I(y_1)$ and 
$I(y_2)$. Now, if there exists a vertex $y_3$ distinguishes $x_1$ and 
$x_2$ from the left and a vertex $x_0$ distinguishes $y_1$ and $y_2$ from the left then the vertices $x, x', x'', x_0, x_1, x_2, y, y_1', 
y_1, y_2, y_3$ and $y_4$ induce $P_1$, which is a contradiction. Again, 
as before if, $Z_2=\{x_1, x_2, x_3\}$ and $y_3$ distinguishes $x_1$ and 
$x_2$ from the right and $y_4$ distinguishes $x_2$ and $x_3$ from the 
right. Then, the vertices $x, x', x'', x_0, x_1, x_2, x_3, y, y_1', y_1, y_2, 
y_3$ and $y_4$ induce $\widetilde{P}_1$, which is also a contradiction.\par
Next, suppose for the bad pair $(x',x)$, $Z_1=\{y_1,y_2\}$, 
$Z_2=\{x_1, x_2\}$, $Z_3=\{y_3, y_4\}$, $Z_4=\{x_3\}$ and 
$Z_1'=\{y_1'\}$. Now, $x'$ is not an isolated vertex, there exists a 
vertex $y_1''$ such that $I(y_1'')$ intersects $I(x')$ where 
$r(x)<r(y_1'')$ and $l(x')<l(y_1'')$. Again, $y_1'$ and $y_1$ are not 
copies, there exists a vertex $x_1'$ such that $l(y_1)<l(x_1')<l(y_2)$ 
and $r(y_1)<r(x_1')<r(y_2)$. Since $r(x_1')$ cannot be shorten to the left, 
there exist a vertex $y_2'$ such that $r(x_1')>l(y_2')$ and 
$r(y_2')>r(y_2)$. Next, since $l(y_2')$ cannot be extended to the left, 
there exists a vertex $x_2'$ such that $r(x_2')<l(y_2')<r(x_1')$ and 
$l(y_1)<l(x_2')<l(y_2)$. Finally, let there exists a vertex $x_3'$ 
distinguishes $y_3$ and $y_4$ from the left. Now the vertices $x, x', 
x_1', x_2', x_3', x_1, x_2, x_3, y_1', y_2', y_1'', y_1, y_2, y_3, y_4$ 
induce $Q_1$, which is a contradiction. As before if, $Z_3=\{y_3,y_4, y_5\}$ then we can similarly derive $\widetilde{Q}_1$, which is again a contradiction.\par

Next, suppose for the bad pair $(x',x)$, $Z_1=\{y_1, y_2\}$, $Z_2=\{x_1, 
x_2\}$, $Z_3=\{y_4\}$ and $Z_1'=\{y_1', y_2'\}$, $Z_2'=\{x_1', x_2'\}$, 
$Z_3'=\{y_4'\}$. Since $x'$ is not an isolated vertex, there exists a vertex 
$y'$ intersecting $x'$ and $l(y')<l(x)$ also $r(y')<r(x')$. If, there 
exists a vertex $y_3$ distinguishes $x_1$ and $x_2$ from the left and 
also a vertex $y_3'$ distinguishes $x_1'$ and $x_2'$ from the right then 
the vertices $x, x', x_1, x_2, x_1', x_2', y', y_1, y_2, y_2', y_3, y_4, 
y_3'$ and $y_4'$ induce $L_{1,1}$, which is a contradiction.\par
By induction, let $Z_i=\{u_k, u_{k+1}\}$ and $v_m$ and $v_{m+1}$ are 
vertices intersecting $I(u_{k+1})$ and to the right of $u_k$, where 
$l(v_m)<l(v_{m+1})<r(v_m)<r(v_{m+1})$. Also, suppose there exist vertices 
$u'$ and $u''$ distinguish $v_m$ and $v_{m+1}$ from the left and right 
respectively.\par
Also, by induction, let $Z_j'=\{u_p', u_{p+1}'\}$ and $v_n'$, $v_{n+1}'$ 
are vertices intersecting $I(u_{p+1}')$ and to the left of $u_p'$, where 
$l(v_{n+1}')<l(v_{n}')<r(v_{n+1}')<r(v_{n}')$. Suppose there exists vertices 
$u_0'$ and $u_0''$ distinguish $v_n'$ and $v_{n+1}'$ from the right and 
left respectively. Now, the vertices $x, x', x_1, x_2, x_1', x_2', 
\ldots , u_k, u_{k+1}, v_m, v_{m+1}, u', u'', y', y_1, y_2', y_3', y_4', 
\ldots , u_p', u_{p+1}', v_n', v_{n+1}', u_0', u_0''$ induce $L_{i, j}$. 
Which is a contradiction.\par
Again, suppose there are vertices $v_m, v_{m+1}, v_{m+2}$ intersecting 
$I(u_{k+1})$ and to the right of $u_k$ satisfying $l(v_m)<l(v_{m+1})<
l(v_{m+2})$ and $r(v_m)<r(v_{m+1})<r(v_{m+2})$. Also, the vertex $u'$ 
distinguishes $v_m$ and $v_{m+1}$ from the right and the vertex $u''$ 
distinguishes $v_{m+1}$ and $v_{m+2}$ from the right. Then, the vertices 
$x, x', x_1, x_2, x_1', x_2', \ldots , u_k, u_{k+1}, v_m, v_{m+1}, 
v_{m+2}, u', u'', y, y_1, y_2', y_3', y_4', \ldots , u_p', u_{p+1}', 
v_n', v_{n+1}', u_0'$ and $u_0''$ induce $\widetilde{L}_{i, j}$, which is 
also a contradiction.\par
Next, suppose in the left there are vertices $v_n', v_{n+1}', v_{n+2}'$ 
intersecting $I(u_{p+1}')$ and to the left of $u_p'$ satisfying 
$l(v_{n+2}')<l(v_{n+1}')<l(v_n')$ and $r(v_{n+2}')<r(v_{n+1}')<r(v_n')$. 
Also, suppose there exists vertex $u_0'$ distinguishes $v_n'$ and 
$v_{n+1}'$ from the left and the vertex $u_0''$ distinguishes $v_{n+1}'$ 
and $v_{n+2}'$ from the left. Then, the vertices $x, x', x_1, x_2, x_1', 
x_2', \ldots , u_k, u_{k+1}, v_m, v_{m+1}, v_{m+2}, u', u'', y, y_1, 
y_2', y_3', y_4', \ldots , u_p', u_{p+1}', v_n', v_{n+1}', v_{n+2}', 
u_0', u_0''$ induce $T_{i, j}$, which is again a contradiction.\par
Finally, consider the case where for the bad pair $(x',x)$, we suppose 
$Z_1=\{y_1, y_2\}$, $Z_2=\{x_1, x_2\}$, $Z_3=\{y''\}$ and $Z_1'=\{y_1', 
y_2'\}$, $Z_2'=\{x_1', x_2'\}$, $Z_3'=\{y_0''\}$. Since $x'$ is not an 
isolated vertex, there exists a vertex $y_0$ such that $l(x')<r(y_0)<
r(x')$ and $l(y_0)<l(x)$ (by Claim~2). Again, $y_1$ and $y_1'$ are copies, 
there exists a vertex $x_0$ such that $I(x_0)$ intersects $I(y_1')$ and 
obviously intersects $I(y_2')$ and let $r(y_2')<r(x_0)<r(y_1')$ and 
$l(x_0)<l(y)<l(y_1')$. Also, $l(x)$ should not be extended to 
the left, there exists a vertex $y$ such that $r(y)<l(x)$ and 
$l(x_0)<l(y)<l(y_1')$. Now, if $(y,x_0)$ is not a bad pair, then we can derive $F_4$ as an induced subgraph. Next, let $(y,x_0)$ is a bad pair. Then for the bad pair $(y,x_0)$, by Claim~1 we 
have a vertex $x''$, such that $r(y)<l(x'')<r(x_0)$. Since $I(x'')$ 
is not contained in $I(x)$, we have $r(x)<r(x'')$. Also, $I(x'')$ intersects 
$I(y_1')$, $I(y_1)$, and $I(y_2)$ but not $I(y_2')$. Obviously $I(y)$ intersects $I(x_0)$. Next, suppose there 
exists a vertex $y'$ distinguishes $x_1$ and $x_2$ from the left and 
there exists a vertex $y_0'$ distinguishes $x_1'$ and $x_2'$ from the 
right. Now, the vertices $x'', x', x_0, x_1, x_2, x_1', x_2', y, y_0, y_1, 
y_2, y_1', y_2', y', y'', y_0'$ and $y_0''$ induce $K_{1,1}$, which is a 
contradiction. Next, if $Z_2 =\{x_1, x_2, x_3\}$ and $y'$ distinguishes $x_1$ and $x_2$ from the right and $y''$ distinguishes $x_2$ and $x_3$ from the right. Now the vertices $x', x'', x_0, x_1, x_2, x_3, x_1', x_2', y, y_0,y_1, y_2, y_1', y_2', y', y'', y_0'$ and $y_0''$ induces $\widetilde{K}_{1,1}$. We have seen before that we can construct $L_{i,j}$ and $\widetilde{L}_{i,j}$ by induction. Similarly, by induction, we can derive $K_{i,j}$ and 
$\widetilde{K}_{i,j}$.\par
\par
Then from the method of construction of the bigraphs $M_1, H_1', N_1, P_1, Q_1, R_1$ and $S_1$, we conclude that as the case of  bigraphs $L_{i,j}$ and $ T_{i,j}$ we can also construct the bigraphs $M_i, H_i', N_i, P_i, Q_i,R_i$ and $S_i$. Thus we have the following properties of $I$ (representation of bigraph $B$) as state in the Claim~4\par
\begin{claim}
For the bad pair $(u',u)$, $i\in [k_r(u',u)-1]$, we have the 
following results:\par
\begin{enumerate}[(a)]
\item $|Z_i|=2$ and if, $Z_i=\{u_k, u_{k+1}\}$, then we have 
$l(u_k)<l(u_{k+1})<r(u_k)<r(u_{k+1})$.
\item The vertices of $Z_{i+1}$ belong to different partite set and are to the right of $u_k$.
\item The vertices $u_k$ and $u_{k+1}$ are not distinguishable from the 
left. 
\item If $i=k_r(u',u)$, then $|Z_i|=1$.
\end{enumerate}

\end{claim}
\par
Also, from symmetry we have the similar results for the vertices of 
$Z_j'$ where $j\in [k_l(u',u)-1]$ and if $j=k_l(u',u)$, then 
$|Z_{j}'|=1$. Also if $Z_j'=\{u_k', u_{k+1}'\}$, then $u_k'$ and $u_{k+1}'$ are not distinguishable from the right.\par
\begin{claim}
Let for a bad pair $(u',u)$, $i\in [k_r(u',u)-1]$. Then, there exists no 
vertex $z\in V(B)$ such that $(u_i,z)$ or $(u_{i+1},z)$ is a bad pair where $Z_i=\{u_i, 
u_{i+1}\}$. 
\end{claim}
\begin{proof}\renewcommand{\qedsymbol}{}
Let $i=1$, then $Z_i=\{u_1, u_2\}$. Now,  without loss of generality $(u_1,z)$ is a bad pair 
then by Claim~2, $(u_2,z)$ is not a bad pair. Again by Claim~1, there 
exists a vertex $v_0$ such that $l(z)<r(v_0)<l(u_1)$. Now, if $I(v_0)$ 
contained in $I(u)$ we have contradiction by Claim~2, and if, $l(v_0)<l(u)$ 
then we have a contradiction by Claim~3.\par
Next if $i=k\geq 2$ and let $Z_i=\{u_k, u_{k+1}\}$, then applying 
Claim~1 to the bad pair $(u_k,z)$, we have a vertex $v_0$ such that 
$l(z)<r(v_0)<l(u_k)$. Now, if $l(v_0)<l(v_{k-1})$ then we have $|Z_{i-1}|=3$, 
which is a contradiction. In the other case if, $(v_0,v_{k-1})$ is a bad pair 
then by Claim~1, there exists a vertex $u_0$ such that $r(v_0)<l(u_0)<
r(v_{k-1})$. Then, we have $|Z_i|=3$, which is again a contradiction. 
Similarly we can prove that $(u_{i+1},z)$ is not a bad pair. This 
completes the proof of Claim~5.\par
Note that from the proof the above Claim we have the next 
Corollary.
\end{proof}
\noindent\textbf{Corollary .}
\textit{Let for a bad pair $(u',u)$, $i\in [k_r(u',u)-1]$. Then there exists no vertex $z\in V(B)$ such that $(z,u_i)$ or $(z,u_{i+1})$ is a bad pair.}\par
Now, we modify the interval representation $I$ of the bigraph $B$ step by 
step. First, we modify the intervals of $Z_k=\{u_k, u_{k+1}\},\ k\in 
[k_r(x',x)-1]$ for the bad pair $(x',x)$ as follows; $I':V\to I^{++}
\cup I^{+-}$, where $I(y_2)=[r(x),r(y_2)]$, $I(y_1)=[r(x),r(y_2))$ and 
in general, $I'(u_{k+1})=[r(u_{k-1}),r(u_{k+1})]$ and 
$I'(u_k)=[r(u_{k-1}),r(u_{k+1}))$. If $m=k_r(x',x)$, then $I'(u_m)=
[r(u_{m-1}),r(u_m)]$ and $I'(u)=I(u)$, for other $u\in V(B)$.\par
\begin{claim}
$I'$ is an interval representation of $B$.
\end{claim}
\begin{proof}\renewcommand{\qedsymbol}{}
We shall show that all the intersections and non intersections of $I$ 
are also preserved in $I'$. Assume for some $x \in V(B)$, $I(x)\cap I(u_{k+1})=\emptyset$ then 
trivially $I'(x)\cap I'(u_{k+1})=\emptyset$ and $I'(x)\cap I'(u_k) 
=\emptyset$. Next, suppose $I(x)\cap I(u_{k+1})\neq \emptyset$ but 
$I(x)\cap I(u_k)=\emptyset$, then it appears that it may possible that $I'(x)\cap I'(u_k) 
\neq \emptyset$. Obviously $l(x)<l(u_{k+3})$. Now, if $r(x)<r(u_{k+3})$. 
Then, $Z_{k+1}=\{u_{k+2}, x, u_{k+3}\}$ i.e. $|Z_{k+1}|=3$, which is 
contradiction. In the other case if, $r(x)>r(u_{k+3})$ then $(u_{k+3},x)$ 
is a bad pair ,which is also a contradiction. Next, let 
$I(x)\cap I(u_{k+1})\neq \emptyset$ but $I'(x)\cap I'(u_{k+1})=
\emptyset$ (then also $I'(x)\cap I'(u_k)=\emptyset$). Now if, $l(u_{k-1})<
l(x)$ then $(x,u_{k-1})$ is a bad pair. Again, if $l(u_{k-2})<l(x)
<l(u_{k-1})$ then $|Z_{k-1}|=3$, $(Z_{k-1}=\{u_{k-2}, x, u_{k-1}\})$. Thus we have a contradiction in both cases. Finally, if $l(x)<l(u_{k-2})$ then $(u_{k-2},x)$ is a 
bad pair, which is also a contradiction by Claim~5. This completes the 
proof of the claim. 
\end{proof}
Next, we modify the intervals of $Z'_k, k\in [k_l(x',x)-1]$ for the 
bad pair $(x',x)$ as follows. $I'':V\to I^{++}\cup I^{-+}$, where 
$I''(y'_2)=[l(y'_2),l(x)]$, $I''(y'_1)=(l(y'_2),l(x)]$ and in general, 
for $Z'_k=\{u'_k, u'_{k+1}\}$, $I''(u'_{k+1})=[l(u'_{k+1}),l(u'_{k-1})]$, 
$I''(u'_k)=(l(u'_{k+1}),l(u'_{k-1})]$. Also, for $n=k_l(x',x)$, 
$I''(u_n)=[l(u_n),l(u_{n-1}]$ and $I''(u)=I(u)$ for other $u\in V(B)$.\par
Now, in the following Claim we show that $I''$ is an interval representation of $B$. \par 
\begin{claim}
$I''$ is an interval representation of $B$.
\end{claim}
\begin{proof}
The proof is similar to the proof of Claim~6. If, there exists $x\in V(B)$ 
such that $I(x)\cap I(u'_{k+1})=\emptyset$, then trivially $I''(x) \cap 
I''(u'_{k+1})=\emptyset$. Now, we shall show that if, $I(x)\cap I(u'_k)=\emptyset$ then also $I''(x)\cap I''(u'_k)=\emptyset$. Suppose $I(x)\cap 
I(u'_{k+1})\neq \emptyset$. So $l(u'_{k+1})<r(x)$. Next, assume 
$r(u'_{k+3})<r(x)<r(u'_{k+2})$. Now, if $l(x)<l(u'_{k+3})$ then $Z'_{k+1}=\{u_{k+2}, x, u_{k+3}\}$ i.e. $|Z'_{k+1}|=3$, which is a 
contradiction. 
 Again, if $l(x)>l(u'_{k+3})$ then $(x,u_{k+3})$ is a bad pair, which is a contradiction.
 In the remaining cases we can similarly derive a 
contradiction. Thus $I''(x)\cap I''(u'_k)=\emptyset$.\par
Next, let $I(x)\cap I(u'_{k+1})\neq \emptyset$ such that 
$l(u'_{k-1})<l(x)$. Assume $I''(x)\cap I''(u'_{k+1})=\emptyset$. 
Now, if $r(x)<r(u'_{k-1})$ then $(x,u'_{k-1})$, is a bad pair, which is 
a contradiction. Again, if $r(u'_{k-1})<r(x)<r(u'_{k-2})$, then 
$Z'_{k-1}=\{u'_{k-2}, x, u'_{k-1}\}$ i.e. $|Z'_{k-1}|=3$, which is also a 
contradiction. Finally, if $r(x)>r(u'_{k-2})$ then $(u_{k-2},x)$ is a 
bad pair, which is also a contradiction. Hence $I''(x)\cap I''(u'_{k+1})\neq \emptyset $. Thus $I''$ is an interval 
representation of $B$ and this completes the proof of Claim~7.\par
It can be easily observed from the construction of $I''$  that no new bad pair is formed. Since the bad pair $(x',x)$ is 
clean we replace $I(x')$ by the open copy of $I(x)$. Thus we get a 
mixed-proper interval representation of $B$. This completes the proof of the Theorem 17.
\end{proof}
\section{Conclusion}
The problem of finding a recognition algorithm for interval graphs (or bigraphs) of open and closed intervals is still open. However, in a recent paper \cite{tk} Talon and Kratochv\'il have given a quadratic-time algorithm to recognize the class of mixed unit interval graphs. We do hope that our work will be a motivating factor to find an efficient algorithm to recognize mixed unit interval bigraphs.

\end{document}